\pdfoutput=1

\documentclass[12pt]{article}
\usepackage[utf8]{inputenc}
\usepackage{amsmath, bm}
\usepackage{amssymb}
\usepackage{bbm}
\usepackage{amsthm}
\usepackage{fullpage}
\usepackage{natbib}
\setcitestyle{authoryear,open={(},close={)}}
\usepackage{xcolor}
\usepackage{hyperref}
\usepackage{cleveref}
\usepackage{IEEEtrantools} 
\usepackage{graphicx}
\usepackage{subcaption}
\usepackage[affil-it]{authblk}
\usepackage{cancel}
\usepackage{enumitem}

\DeclareMathOperator*{\argmin}{arg\,min}

\def\KL{\text{KL}}

\def\KSD{\text{KSD}}
\def\S{\mathcal{S}}
    \def\P{\mathbb{P}}

\def\B{\Theta}
\def\E{\mathbb{E}}
\def\Q{\mathbb{Q}}
\def\F{\mathcal{F}}

\def\U{\mathcal{H}}
\def\H{\mathcal{H}}

\def\N{\mathbb{N}}

\def\R{\mathbb{R}}
\def\X{\mathcal{X}}

\newtheorem{assumption}{Assumption}
\newtheorem{theorem}{Theorem}

\newtheorem{lemma}{Lemma}

\newtheorem{proposition}{Proposition}

\newtheorem{definition}{Definition}

\newtheorem{remark}{Remark}

\crefname{assumption}{assumption}{assumptions}

\newenvironment{talign*}
 {\csname align*\endcsname}
 {\endalign}

\title{Robust Generalised Bayesian Inference \\ for Intractable Likelihoods}

\author[1,4]{Takuo Matsubara}
\author[2,4]{Jeremias Knoblauch}
\author[3,4]{Fran\c{c}ois-Xavier Briol}
\author[1,4]{Chris. J. Oates}

\affil{Newcastle University, $^2$University of Warwick,}
\affil[3]{University College London, $^4$The Alan Turing Institute}

\begin{document}

\maketitle

\begin{abstract}

Generalised Bayesian inference updates prior beliefs using a loss function, rather than a likelihood, and can therefore be used to confer robustness against possible mis-specification of the likelihood. 
Here we consider generalised Bayesian inference with a Stein discrepancy as a loss function, motivated by applications in which the likelihood contains an intractable normalisation constant.
In this context, the Stein discrepancy circumvents evaluation of the normalisation constant and produces generalised posteriors that are either closed form or accessible using standard Markov chain Monte Carlo. 
On a theoretical level, we show consistency, asymptotic normality, and bias-robustness of the generalised posterior, highlighting how these properties are impacted by the choice of Stein discrepancy. 
Then, we provide numerical experiments on a range of intractable distributions, including applications to kernel-based exponential family models and non-Gaussian graphical models.

\vspace{5pt}
\end{abstract}


\section{Introduction} \label{sec:introduction}

A considerable proportion of statistical modelling deviates from the idealised approach of fine-tuned, expertly-crafted descriptions of real-world phenomena, in favour of default models fitted to a large dataset.
If the default model is a good approximation to the data-generating mechanism this strategy can be successful, but things can quickly go awry if the default model is misspecified. 
Generalised Bayesian updating \citep{Bissiri2016}, and in particular using divergence-based loss functions \citep{Jewson2018}, has been shown to mitigate some of the risks involved when working with a model that is misspecified. 
Unlike other robust modelling strategies, these methods do \textit{not} change the  statistical model. 
Instead, they change how the model's parameters are scored, affecting how ``good'' parameter values are discerned from ``bad'' ones.
This is a key practical advantage, as it implies that such strategies do not require precise knowledge about how the model is misspecified.
\textcolor{black}{This paper considers generalised Bayesian inference in the context of intractable likelihood.}
An \textit{intractable likelihood}, in this paper, takes the form $p_\theta(x) = q(x, \theta) / Z(\theta)$, where $q(x, \theta)$ is an analytically tractable function and $Z(\theta)$ is an \textit{intractable} normalising constant, each depending on the value of the unknown parameter $\theta$ of interest. 
Classical Bayesian posteriors resulting from intractable likelihood models are sometimes called \emph{doubly intractable}, due to the computational difficulties they entail \citep{Murray2006}.
For example, standard Markov chain Monte Carlo (MCMC) methods cannot be used in this setting, since they typically require explicit evaluation of the likelihood.
Doubly intractable posteriors  appear in many important statistical applications, including spatial models \citep{Besag1974, Besag1986, Diggle1990}, exponential random graph models \citep{Park2018}, \textcolor{black}{models for gene expression \citep{jiang2021bayesian}, and hidden Potts models for satellite data \citep{moores2020scalable}}.

This paper proposes the first generalised Bayesian approach to inference for models that involve an intractable likelihood.
To achieve this, we propose to employ a loss function based on a \textit{Stein discrepancy} \citep{Gorham2015}. 
As such, this research can be thought of as a Bayesian alternative to the minimum Stein discrepancy estimators of \citet{Barp2019}.
The methodology is developed for a particular Stein discrepancy called \textit{kernel} Stein discrepancy (\KSD), and we call the resulting generalised Bayesian approach \emph{KSD-Bayes}. 
It is shown in this paper that KSD-Bayes (1) provides robustness to misspecified likelihoods; (2) produces a generalised posterior that is tractable for standard MCMC, or even closed-form when an appropriate conjugate prior (which we identify) is used together with an exponential family likelihood; (3) satisfies several desirable theoretical properties, including a Bernstein--von Mises result which holds irrespective of whether the likelihood is correctly specified.
These results appear to represent a compelling case for the use of KSD-Bayes as an alternative to standard Bayesian inference with intractable likelihood.
However, KSD-Bayes is no panacea and caution must be taken to avoid certain pathologies of KSD-Bayes, which we highlight in \Cref{subsec: pathologies}.

The paper is structured as follows:
\Cref{sec:background} contains necessary background on generalised Bayesian inference, Stein discrepancy, and robustness in the Bayesian context.
\Cref{sec:methodology} presents the KSD-Bayes methodology, including conjugacy of the generalised posterior under an exponential family likelihood.
\Cref{sec:theory} elucidates the robustness and asymptotic properties of KSD-Bayes.
Guidance for practical application of KSD-Bayes is contained in \Cref{subsec: choice of k}.
The experimental results and empirical assessments are outlined in \Cref{sec:experiment}, and we draw our conclusions in \Cref{sec:conclusion}.
Code to reproduce all results in this paper can be downloaded from: \url{https://github.com/takuomatsubara/KSD-Bayes}.


\section{Background} \label{sec:background}

First we provide a short summary of generalised Bayesian inference and Stein discrepancies, putting in place a standing assumption on the domains in which data and parameters are contained:

\vspace{5pt}
\noindent \textbf{Standing Assumptions 1:}
The topological space $\X$, in which the data are contained, is locally compact and Hausdorff.
The set $\Theta \subseteq \R^p$, in which parameters are contained, is Borel.

\subsection{Notation} \label{sec:notation}

\textit{Measure theoretic notation:}
For a locally compact Hausdorff space such as $\X$, we let $\mathcal{P}(\X)$ denote the set of all Borel probability measures on $\X$. 
A point mass at $x$ is denoted $\delta_{x} \in \mathcal{P}(\X)$.
If $\mathcal{X}$ is equipped with a reference measure, then we abuse notation by writing $p \in \mathcal{P}(\X)$ to indicate that the distribution with p.d.f. $p$ is an element of $\mathcal{P}(\X)$.
For $\P \in \mathcal{P}(\X)$, we occasionally overload notation by denoting by $L^q(\X, \P)$ both the set of functions $f: \X \to \R$ for which $\| f \|_{L^q(\X, \P)} := ( \int_{\X} | f |^q \mathrm{d} \P )^{1/q} < \infty$ and the normed space in which two elements $f,g \in L^q(\X, \P)$ are identified if they are $\P$-almost everywhere equal. 
If $\P$ is a Lebesgue measure, we simply write $L^q(\X)$ instead of $L^q(\X, \P)$.
Let $\mathcal{P}_{\text{S}}(\R^d)$ be the set of all Borel probability measures $\P$ supported on $\R^d$, admitting an everywhere positive p.d.f. $p$ and continuous partial derivatives $x \mapsto (\partial/\partial x_{(i)}) p(x)$.

\textit{Real analytic notation:}
The Euclidean norm on $\R^d$ is denoted $\| \cdot \|_2$.
The set of continuous functions $f: \X \to \R$ is denoted $C(\X)$.
We denote by $C_b^{1}(\R^d)$ the set of functions $f: \R^d \to \R$ such that both $f$ and the partial derivatives $x \mapsto (\partial/\partial x_{(i)}) f(x)$ are bounded and continuous on $\R^d$.
We also denote by $C_b^{1,1}(\R^d \times \R^d)$ the set of bivariate functions $f: \R^d \times \R^d \to \R$ such that both $f$ and the partial derivatives $(x, x') \mapsto (\partial/\partial x_{(i)}) (\partial/\partial x'_{(j)}) f(x,x')$ are bounded and continuous on $\R^d \times \R^d$.
For an arbitrary set $\S(\X)$ of functions $f: \X \to \R$, denote by $\S(\X; \R^k)$ the set of $\R^k$-valued functions whose components belong to $\S(\X)$.
Let $\nabla$ and $\nabla \cdot$ be the gradient and the divergence operators \textcolor{black}{in} $\R^d$.
For functions with multiple arguments, we sometimes use subscripts to indicate the argument \textcolor{black}{to} which the operator is applied (e.g. $\nabla_x f(x,y)$).
For $f$ an $\R^d$-valued function, $[ \nabla f(x) ]_{(i,j)} := ( \partial / \partial x_{(i)} ) f_{(j)}(x)$ and $\nabla \cdot f(x) := \sum_{i=1}^{d} ( \partial / \partial x_{(i)} ) f_{(i)}(x)$.
For $f$ an $\R^{d \times d}$-valued function,  $[\nabla f(x)]_{(i,j,k)} := ( \partial / \partial x_{(i)} ) f_{(j,k)}(x)$ and $[ \nabla \cdot f(x) ]_{(i)} := \sum_{j=1}^{d} ( \partial / \partial x_{(j)} ) f_{(i,j)}(x)$.

\subsection{Generalised Bayesian Inference}
\label{sec:gen_bayes}

Consider a dataset consisting of independent random variables $\{x_i\}_{i=1}^n$ generated from $\P \in \mathcal{P}(\X)$, together with a statistical model $\P_\theta \in \mathcal{P}(\X)$ for the data, with p.d.f. $p_\theta$, indexed by a parameter of interest $\theta \in \Theta$.
\textcolor{black}{The Bayesian statistician elicits a prior $\pi \in \mathcal{P}(\Theta)$, which may reflect \textit{a priori} belief about the parameter $\theta \in \Theta$, and determines their \textit{a posteriori} belief according to }
\begin{align}
    \pi_n(\theta) \propto \pi(\theta) \prod_{i=1}^n p_\theta(x_i) = \pi(\theta) \exp\left\{ \sum_{i=1}^n \log p_\theta(x_i) \right\}. \label{eq:CBP}
\end{align}
In the \textit{M-closed} setting there exists $\theta_0 \in \Theta$ for which $\P = \P_{\theta_0}$, and the Bayesian update is optimal from an information-theoretic perspective \citep[see][]{Williams, Zellner1988}. 
Optimal processing of information is a desirable property, but in applications the assumption of adequate prior and model specification is often violated.
This has inspired several lines of research, including (but not limited to) strategies for the robust specification of prior belief
\citep{Berger1994}, the so-called \textit{safe Bayes} approach \citep{SafeLearning, SafeBayesian}, \textit{power posteriors} \citep[e.g.][]{holmes2017assigning}, \textit{coarsened posteriors} \citep{DunsonCoarsening} and Bayesian inference based on scoring rules \citep{Giummole2019}. 
A particularly versatile approach to robustness, which encompasses most of the above, is \textit{generalised Bayesian inference} \citep{Bissiri2016} \citep[see also the earlier work of][]{chernozhukov2003mcmc}. 
This approach constructs a distribution, denoted $\pi_n^L$, using a \textit{loss function} $L_n: \Theta \rightarrow \R$, which may be data-dependent, and a scaling parameter $\beta>0$, according to
\begin{align}
    \pi_n^{L}(\theta) \propto \pi(\theta) \exp \left\{-\beta n L_n(\theta) \right\} .
    \label{eq:gen-posterior}
\end{align}
The so-called \textit{generalised posterior} $\pi_n^L$ coincides with the Bayesian posterior $\pi_n$ when $\beta = 1$ and the loss function is the negative average log-likelihood; $L_n(\theta) = - \frac{1}{n} \sum_{i=1}^n\log p_\theta(x_i)$. 
As discussed in \cite{Knoblauch2019}, generalised Bayesian inference admits an optimisation-centric interpretation:
\begin{align}
\pi_n^{L} = \argmin_{\rho \in \mathcal{P}(\Theta)}\bigg\{ \beta n \ \mathbb{E}_{\theta \sim \rho}\left[ L_n(\theta) \right] + \KL(\rho \|\pi) \bigg\} \label{eq:gen-posterior-opt}
\end{align}
where $\text{KL}(\rho \| \pi)$ denotes the Kullback--Leibler (KL) divergence between two distributions $\rho,\pi \in \mathcal{P}(\Theta)$.
This perspective reveals that the standard Bayesian posterior is an implicit commitment to a particular loss function -- the negative log-likelihood -- and that the weighting constant $\beta$ controls the influence of this loss relative to the prior $\pi$.
In particular, under mild conditions $L_n(\theta) \stackrel{\text{a.s.}}{\rightarrow} \KL(\P\|\P_{\theta}) + C$ as $n \rightarrow \infty$, for a constant $C$ independent of $\theta$, which reveals that standard Bayesian posterior concentrates around the value of $\theta$ that minimizes the $\KL$ divergence between the data-generating distribution $\P$ and the model $\P_{\theta}$.
Outside of the M-closed setting such concentration is problematic, often leading to over-confident predictions \citep[]{Bernardo}.

The use of alternative, divergence-based loss functions has been demonstrated to mitigate the negative consequences of a misspecified statistical model, as pioneered in the work on $\alpha$- and $\beta$-divergences in \cite{hooker2014bayesian, Ghosh2016} and extended to $\gamma$-divergence in  \cite{Nakagawa2020}.
The properties of the divergence, including any potentially undesirable pathologies associated with it, determine the properties of the generalised posterior \citep{Jewson2018,Knoblauch2019}.
These compelling theoretical results have led to considerable interest in generalised Bayesian inference with divergence-based loss functions, yet the divergences that have been considered to-date cannot be computed in the important setting of intractable likelihood.

\subsection{Stein Discrepancy}

In an independent line of research, \textit{Stein discrepancies} were proposed in \cite{Gorham2015} to provide statistical divergences that are both computable and capable of providing various forms of distributional convergence control.
The approach is based on the method of \citet{stein1972bound}, which requires the identification of a linear operator $\S_\Q: \U \to L^1(\X, \Q)$, depending on a probability distribution $\Q \in \mathcal{P}(\X)$ and acting on a Banach space $\U$, such that
\begin{align}
\E_{X \sim \Q}[\S_\Q [h](X)] = 0 \quad \forall h \in \U. \label{eq:stein_identity}
\end{align}
Such an operator $\S_\Q$ is called a \emph{Stein operator} and $\H$ is called a \textit{Stein set}.
Given a distribution $\Q \in \mathcal{P}(\X)$, there are infinitely many operators $\S_\Q$ satisfying \eqref{eq:stein_identity}.
A convenient example is the \emph{Langevin Stein operator} \citep{Gorham2015}, defined for $\X = \R^d$, $\Q \in \mathcal{P}_{\text{S}}(\R^d)$ and a Banach space $\U$ of differentiable functions $h : \R^d \rightarrow \mathbb{R}^d$, as
\begin{align}
\S_\Q [h](x) = h(x) \cdot \nabla \log q(x) + \nabla \cdot h(x)  \label{eq:stein_operator}
\end{align}
where $q$ is the p.d.f. of $\Q$.
Under suitable regularity conditions on $\nabla \log q$ and $\U$, the Langevin Stein operator satisfies Equation \ref{eq:stein_identity}; see \citet[Proposition 1]{Gorham2015}.
Given $\P,\Q \in \mathcal{P}(\mathcal{X})$ and a Stein operator $\S_\Q : \U \rightarrow L^1(\X,\Q)$ whose image is contained in $L^1(\X,\P)$, the \textit{Stein discrepancy} (SD) 
is defined as 
\begin{align}
\operatorname{SD}(\Q \| \P) := \sup_{\| h \|_\U \le 1} \Big| \E_{X \sim \P}\left[ \S_\Q[h](X) \right] - \E_{X \sim \Q}\left[ \S_\Q[h](X) \right] \Big| = \sup_{\| h \|_\U \le 1} \Big| \E_{X \sim \P}\left[ \S_\Q[h](X) \right] \Big|, \label{eq:SD}
\end{align}
where the last equality follows directly from \eqref{eq:stein_identity}.
Under mild assumptions, SD defines a statistical divergence between two probability distributions $\P,\Q \in \mathcal{P}(\X)$, meaning that $\text{SD}(\Q \| \P) \geq 0$ with equality if and only if $\P = \Q$; see Proposition 1 and Theorem 2 in \cite{Barp2019}. 
Under slightly stronger assumptions SD provides convergence control, meaning that a sequence $(\Q_n)_{n=1}^\infty \subset \mathcal{P}(\X)$ converges in a specified sense to $\Q$ whenever $\text{SD}(\Q \| \Q_n) \rightarrow 0$; see \citet[][Theorem 2, Proposition 3]{Gorham2015} and \citet[][Theorem 8, Proposition 9]{Gorham2017}. 
An important property of SDs that we exploit in this work is that, unlike other divergences, SDs can often be computed with an un-normalised representation of $\Q$.
For example, the Stein operators in \eqref{eq:stein_operator} depend on $\Q$ only through $\nabla \log q$, which can be computed when $q$ is provided in a form that involves an intractable normalisation constant.
The suitability of SD for use in generalised Bayesian inference has not previously been considered, and this is our focus next.


\section{Methodology} \label{sec:methodology}

Highly structured data, or data belong to a high-dimensional domain $\mathcal{X}$, are often associated with an intractable likelihood.
Moreover, the difficulty of modelling such data means that models will typically be misspecified.
Thus there is a pressing need for Bayesian methods that are both robust and compatible with intractable likelihood.
To this end, in \Cref{subsec: sd bayes} we introduce \emph{SD-Bayes}, a generalised Bayesian procedure with a loss function based on SD.
There are numerous SDs that can be considered, and in \Cref{subsec: kernel SDs} we focus in detail on KSD due to the possibility of performing fully conjugate inference in the context of exponential family models, as described in \Cref{subsec: exp fam mod}.
\textcolor{black}{Non-conjugate inference and its computational cost are discussed in \Cref{subsec: computation}.}
However, all statistical divergences have their pathologies, and one must bear in mind the pathologies of KSD when using KSD-Bayes; see the discussion in \Cref{subsec: pathologies}.

\subsection{SD-Bayes} \label{subsec: sd bayes}

Suppose we are given a prior p.d.f. $\pi \in \mathcal{P}(\Theta)$ and a statistical model $\{ \P_{\theta} \mid \theta \in \Theta \} \subset \mathcal{P}(\X)$.
Let $\{ x_i \}_{i=1}^{n}$ be independent observations generated from $\P \in \mathcal{P}(\X)$ and let $\P_n := \frac{1}{n} \sum_{i=1}^n \delta_{x_i}$ be the empirical measure associated to this dataset.
In this context, the SD-Bayes generalised posterior can now be defined:

\begin{definition}[SD-Bayes]
\label{def:KSD_Bayes} 
	For each $\theta \in \Theta$, select  a Stein operator  $\S_{\P_{\theta}}$ and denote the associated Stein discrepancy $\operatorname{SD}(\P_\theta \| \cdot)$.
	Let $\beta \in (0,\infty)$.
	Then the SD-Bayes generalised posterior is defined as
	\begin{align}
	\pi_n^D(\theta) \propto \pi(\theta) \exp\left\{- \beta n \operatorname{SD}^2(\P_{\theta} \| \P_n) \right\}
    \label{eq: sd bayes def}
	\end{align}
	where $\theta \in \Theta$.
\end{definition}

\noindent Here the `$D$' superscript stands for \emph{discrepancy}.
Comparing \eqref{eq: sd bayes def} to \eqref{eq:gen-posterior} confirms that SD-Bayes is a generalised Bayesian method with loss function $L_n(\theta) = \text{SD}^2(\P_\theta \| \P_n)$.
\textcolor{black}{There is an arbitrariness to using squared discrepancy, as opposed to another power of the discrepancy, but this choice turns out to be appropriate for the discrepancies considered in \Cref{subsec: kernel SDs}, ensuring that fluctuations of $L_n(\theta)$ about its expectation are $\mathcal{O}(n^{-1/2})$, analogous to the standard Bayesian loss, and permitting tractable computation (\Cref{subsec: exp fam mod}) and analysis (\Cref{sec:theory}). }
A discussion of how the weight $\beta$ should be selected is deferred until after our theoretical analysis, in \Cref{subsec: choice of k}.

\subsection{KSD-Bayes}
\label{subsec: kernel SDs}

Compared to other Stein discrepancies, KSDs are attractive because they enable the supremum in \eqref{eq:SD} to be be explicitly computed.
To define KSD, we require the concept of a (matrix-valued) \textit{kernel} $K: \X \times \X \to \R^{d \times d}$; the precise definition is contained in \Cref{sec:kernel_RKHS}.
For our purposes in the main text, it suffices to point out that any kernel $K$ \textcolor{black}{has a} uniquely associated Hilbert space of functions $f: \X \to \R^d$, called a \emph{vector-valued reproducing kernel Hilbert space} (v-RKHS).
This v-RKHS constitutes the Stein set in KSD, and we therefore denote this v-RKHS as $\H$.
The associated norm and inner product will respectively be denoted $\| \cdot \|_{\H}$ and $\langle \cdot, \cdot \rangle_{\H}$.

Let $\S_\Q$ be a Stein operator and denote the action of $\S_\Q$ on both the first and second argument\footnote{More precisely, denoting the $j$-th column of $K(x, x') \in \R^{d \times d}$ by $K_{-,j}(x, x') \in \R^d$, we define $\S_\Q K(x, x') := [\S_\Q K_{-,1}(x, x'), \dots, \S_\Q K_{-,d}(x, x')] \in \R^d$ where $\S_\Q K_{-,j}(x, x') := \S_\Q[ K_{-,j}(\cdot, x') ](x)$ is an action of $\S_\Q$ for the $\R^d$-valued function $K_{-,j}(\cdot, x')$ at each $x' \in \X$.
We further define  $\S_\Q \S_\Q K(x,x') := \S_\Q[\ \S_\Q K(x, \cdot)\ ](x')$ as an action of $\S_\Q$ for the $\R^d$-valued function $\S_\Q K(x, \cdot)$ at each $x \in \X$.} of a kernel $K$ as $\S_\Q \S_\Q K$.
The following result is a generalisation of the original construction of KSD \citep{Chwialkowski2016a,Liu2016b} to general Stein operators. 

\begin{assumption} \label{asmp:derivation_KSD}
	Let $\H$ be a v-RKHS with kernel $K: \X \times \X \to \mathbb{R}^{d \times d}$. 
	For $\Q \in \mathcal{P}(\X)$, let $\S_\Q$ be a Stein operator with domain $\H$.
	For each fixed $x \in \X$, we assume $h \mapsto \S_\Q[h](x)$ is a continuous linear functional on $\H$.
	Further, we assume that $\E_{X \sim \P}\left[ \S_\Q \S_\Q K(X,X) \right] < \infty$.
\end{assumption}

\begin{proposition}[Closed form of SD] \label{prop:derivation_KSD}
	Under \Cref{asmp:derivation_KSD}, we have
	\begin{align*}
	\operatorname{SD}^2(\Q \| \P) =
	\operatorname{KSD}^2(\Q \| \P)
	:= \E_{X,X' \sim \P}\left[ \S_\Q \S_\Q K(X,X') \right] 
	\end{align*}
	where $X$ and $X'$ are independent.
\end{proposition}

The proof is in \Cref{apx:proof_KSD_derivation}. 
Note that it is straightforward to verify the assumption that $h \mapsto \S_\Q[h](x)$ is a continuous linear functional for each fixed $x \in \X$ once the form of $\S_\Q$ is specified;
see \Cref{sec:verify_asmp_1}.
KSD is attractive for SD-Bayes since it enables the generalised posterior in \Cref{def:KSD_Bayes} to be explicitly computed:
\begin{IEEEeqnarray}{rCl}
\operatorname{KSD}^2(\P_\theta \| \P_n) = \frac{1}{n^2} \sum_{i=1}^{n} \sum_{j=1}^{n} \S_{\P_\theta}\S_{\P_\theta} K(x_i,x_j) .
\label{eq:empirical_KSD}
\end{IEEEeqnarray}
The resulting generalised posterior will be referred to as \textit{KSD-Bayes} in the sequel.
The explicit form of $\S_{\P_\theta} \S_{\P_\theta} K$ depends on $\S_{\P_\theta}$. 
The case of $\X = \R^d$ and the Langevin Stein operator in \eqref{eq:stein_operator} is given by
\begin{align}
\S_{\P_\theta} \S_{\P_\theta} K(x,x') & = \nabla \log p_\theta(x) \cdot K(x,x') \nabla \log p_\theta(x') + \nabla_{x} \cdot \left( \nabla_{x'} \cdot K(x, x') \right) \nonumber \\
& \hspace{30pt} + \nabla \log p_\theta(x) \cdot \left( \nabla_{x'} \cdot K(x, x') \right) + \nabla \log p_\theta(x') \cdot \left( \nabla_{x} \cdot K(x, x') \right)  \label{eq:langevin_k0}
\end{align}
where $p_\theta$ is a p.d.f. for $\P_\theta \in \mathcal{P}_{\text{S}}(\R^d)$.
Clearly, this expression is straightforward to evaluate\footnote{For maximum clarity, the vector calculus notation is expanded as follows: 
\begin{align*}
    \S_{\P_\theta} \S_{\P_\theta}K(x,x') & = \sum_{i,j=1}^{d} \frac{\partial}{\partial x_{(i)}} \log p_\theta(x) \left[ K(x, x') \right]_{(i,j)} \frac{\partial}{\partial x_{(j)}} \log p_\theta(x) + \frac{\partial^2}{\partial x_{(i)} \partial x'_{(j)}} \left[ K(x, x') \right]_{(i,j)} \nonumber \\
& \hspace{60pt} + \frac{\partial}{\partial x_{(i)}} \log p_\theta(x) \frac{\partial}{\partial x'_{(j)}} \left[ K(x, x') \right]_{(i,j)} + \frac{\partial}{\partial x'_{(j)}} \log p_\theta(x') \frac{\partial}{\partial x_{(i)}} \left[ K(x, x') \right]_{(i,j)}
\end{align*}} whenever we have access to derivatives of the kernel and the log density. 
If the derivatives are analytically intractable, the expression above is amenable to the use of automatic differentiation tools \citep{Baydin2018}. 

Whether KSD-Bayes is reasonable or not hinges crucially on whether KSD is a meaningful way to quantify the difference between the discrete distribution $\P_n$ and the parametric model $\P_\theta$.
Sufficient conditions for convergence control have been established for the Langevin Stein operator, under which the convergence of $\operatorname{KSD}(\P_\theta \| \P_n)$ implies the weak convergence of $\P_n$ to $\P_\theta$ \cite[Theorem~8]{Gorham2017}.
This provides some preliminary assurance that KSD-Bayes may work; we present formal theoretical guarantees in \Cref{sec:theory}.
These theoretical results motivate specific choices of $K$ for use in KSD-Bayes, which we discuss in \Cref{subsec: choice of k}.

\subsection{Conjugate Inference for Exponential Family Models} \label{subsec: exp fam mod}

The generalised posterior can be exactly computed in the case of an natural exponential family model when a conjugate prior is used.
Let $\eta: \Theta \to \R^k$ and $t: \X \to \R^k$ be any sufficient statistic for some $k \in \mathbb{N}$ and let $a: \Theta \to \R$ and $b: \X \to \R$.
An exponential family model has p.m.f. or p.d.f. (with respect to an appropriate reference measure on $\X$) of the form
\begin{align}
p_\theta(x) = \exp( \eta(\theta) \cdot t(x) - a(\theta) + b(x) ). \label{eq:expfam}
\end{align} 
This includes a wide range of distributions with an intractable normalisation constant $\exp( a(\theta) )$, used in statistical applications such as random graph estimation \citep{yang2015graphical}, spin glass models \citep{Besag1974} and the kernel exponential family model \citep{canu2006kernel}.
The model in \eqref{eq:expfam} is called \emph{natural} when the canonical parametrisation $\eta(\theta) = \theta$ is employed.

\begin{proposition} \label{prop:post_expfam}
	Consider $\X = \R^d$ and the Langevin Stein operator $\S_{\P_\theta}$ in \eqref{eq:stein_operator}, where $\P_\theta$ is the exponential family in \eqref{eq:expfam}, and a kernel $K \in C_b^{1,1}(\R^d \times \R^d; \R^{d \times d})$.
	Assuming the prior has a p.d.f. $\pi$, the KSD-Bayes generalised posterior has a p.d.f.
	\begin{align*}
	\pi_n^D(\theta) \propto \pi(\theta) \exp\left( - \beta n \{ \eta(\theta) \cdot \Lambda_{n} \eta(\theta) + \eta(\theta) \cdot \nu_{n} \} \right) ,
	\end{align*}
	where $\Lambda_{n} \in \R^{k \times k}$ and $\nu_{n} \in \R^{k}$ are defined as
	\begin{align*}
	\Lambda_{n} & := \frac{1}{n^2} \sum_{i,j=1}^{n}  \nabla t(x_i) \cdot K(x_i, x_j) \nabla t(x_j) , \\
	\nu_n & := \frac{1}{n^2} \sum_{i,j=1}^{n} \nabla t(x_i) \cdot \big( \nabla_{x_j} \cdot K(x_i, x_j) \big) + \nabla t(x_j) \cdot \big( \nabla_{x_i} \cdot K(x_i, x_j) \big) + 2 \nabla t(x_i) \cdot K(x_i, x_j) \nabla b(x_j) .
	\end{align*}
	For a natural exponential family we have $\eta(\theta) = \theta$, and the prior $\pi(\theta) \propto \exp(- \frac{1}{2} (\theta - \mu) \cdot \Sigma^{-1} (\theta - \mu))$ leads to a generalised posterior
	\begin{align*}
	\pi_n^D(\theta) \propto \exp\left( - \frac{1}{2} ( \theta - \mu_n ) \cdot \Sigma_n^{-1} ( \theta - \mu_n )\right) ,
	\end{align*}
	where $\Sigma_n^{-1} := \Sigma^{-1} + 2 \beta n \Lambda_n$ and $\mu_n := \Sigma_n^{-1} ( \Sigma^{-1} \mu - \nu_n )$.
\end{proposition}

The proof is in \Cref{apx:proof_post_expfam}.
That the Gaussian distribution will be conjugate in KSD-Bayes, even in the presence of intractable likelihood, is remarkable and notably different from the classical Bayesian case, \textcolor{black}{albeit at a $\mathcal{O}(n^2)$ computational cost}.
\textcolor{black}{Strategies to further reduce this computational cost are discussed in \Cref{subsec: computation}.}
It is well known that certain minimum discrepancy estimators, such as the \emph{score matching estimator} \citep{Hyvarinen2005} and the \emph{minimum KSD estimator} \citep{Barp2019}, have closed forms in the case of an exponential family models; it is similar reasoning that has led us to \Cref{prop:post_expfam}.

\color{black}
\subsection{Non-Conjugate Inference and Computational Cost}
\label{subsec: computation}

To access the generalised posterior in the non-conjugate case, existing MCMC algorithms for \textit{tractable} likelihood can be used.
The per-iteration computational cost appears to be $\mathcal{O}(n^2)$ since, for each state $\theta$ visited along the sample path, the KSD in \eqref{eq:empirical_KSD} must be evaluated.
However, various strategies enable this computational cost to be mitigated.
For concreteness of the discussion that follows, we consider the Langevin Stein operator, for which 
\begin{align*}
	\eqref{eq:empirical_KSD} \stackrel{+C}{=} \frac{1}{n^2} \sum_{i=1}^{n} \sum_{j=1}^{n} \left\{ \begin{array}{l} \nabla \log p_\theta(x_i) \cdot K(x_i, x_j) \nabla \log p_\theta(x_j) + \nabla \log p_\theta(x_i) \cdot \nabla_{x_j} \cdot K(x_i, x_j)  \\
	\qquad + \nabla \log p_\theta(x_j) \cdot \nabla_{x_i} \cdot K(x_i, x_j)  \end{array} \right\}
\end{align*}
where the equality holds up to a $\theta$-independent constant.

\vspace{5pt}
\noindent
\textbf{Memoisation:} 
The above expression depends on $\theta$ only through the terms $\{ \nabla \log p_\theta(x_i) \}_{i=1}^{n}$, of which there are $\mathcal{O}(n)$, while all other terms involving $K$, of which there are $\mathcal{O}(n^2)$, can be computed once and memoised.
The double summation still necessitates $\mathcal{O}(n^2)$ computational cost but this operation is \textit{embarrassingly parallel}.

\vspace{5pt}
\noindent
\textbf{Finite rank kernel:}
Computational cost can be reduced from $\mathcal{O}(n^2)$ to $\mathcal{O}(n)$ using a finite rank kernel. 
A useful and important example is the rank one kernel $K(x,x') =  I_{d}$, which reduces \eqref{eq:empirical_KSD} to
\begin{equation*} 
	\eqref{eq:empirical_KSD} \stackrel{+C}{=} \left\| \frac{1}{n} \sum_{i=1}^{n} \nabla \log p_\theta(x_i)  \right\|^2 
\end{equation*}
and is closely related to divergences used in \emph{score matching} \citep{Hyvarinen2005}.
Random finite rank approximations of the kernel can also considered in this context \citep{huggins2018random}.

\vspace{5pt}
\noindent
\textbf{Stochastic approximation:}
The construction of low-cost unbiased estimators for \eqref{eq:empirical_KSD} is straight-forward via sampling \textit{mini-batches} from the dataset.
This enables a variety of exact and approximate algorithms for posterior approximation to be exploited \citep[e.g.][]{Ma2015}.
Alternatively, \citet{huggins2018random,gorham2020stochastic} argued for stochastic approximations of KSD that could be used.

\color{black}

\subsection{Limitations of KSD-Bayes} \label{subsec: pathologies}

A divergence $\text{D}(\Q || \P)$ induces an information geometry \citep{amari1997information}, encoding a particular sense in which $\Q$ can be considered to differ from $\P$.
As such, all divergence exhibit \emph{pathologies}, meaning that certain characteristics that distinguish $\Q$ from $\P$ are less easily detected.
A documented pathology of gradient-based discrepancies, including the Langevin KSD, is their insensitivity to the existence of high-probability regions which are well-separated; see \citet[Section 5.1][]{Gorham2019} and \citet{wenliang2020blindness}.
To see this, consider a Gaussian mixture model 
\begin{align}
	p_\theta(x) = \frac{\theta}{\sqrt{2 \pi}} \exp \left( - \frac{(x - \mu)^2}{2} \right) + \frac{(1 - \theta)}{\sqrt{2 \pi}} \exp\left( - \frac{(x + \mu)^2}{2} \right)  \label{eq: gmm}
\end{align}
where $\theta \in [0, 1]$ specifies the mixture ratio and $\mu \in \R$ controls the separation between the two components.
If the two components are well-separated i.e. $\mu \gg 1$, the gradient $\nabla \log p_\theta$ becomes insensitive to $\theta$ and hence a gradient-based divergence such as KSD will be insensitive to $\theta$, as demonstrated in \Cref{fig:score_mm_example}. 
For this reason, caution is warranted when gradient-based discrepancies are used.
However, in practice direct inspection of the dataset and knowledge of how $\P_\theta$ is parametrised can be used to ascertain whether either distribution is multi-modal.
Our applications in \Cref{sec:experiment} are not expected to be multi-modal (with the exception of the kernel exponential family in \Cref{subsec: kernel family} which was selected to demonstrate the insensitivity to mixing proportions of KSD-Bayes).

\begin{figure}[t!]
	\centering
	\subcaptionbox{$\theta=0.2, \mu = 5$}{\includegraphics[width = 0.24\textwidth]{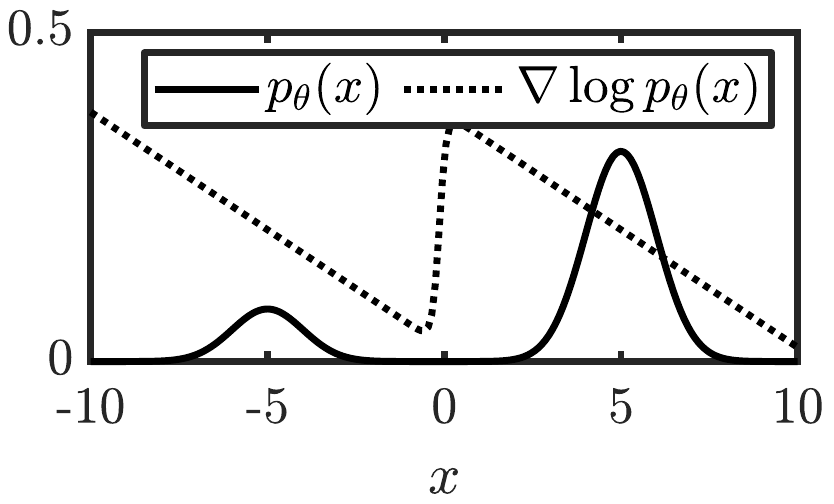}}
	\subcaptionbox{$\theta=0.5, \mu = 5$}{\includegraphics[width = 0.24\textwidth]{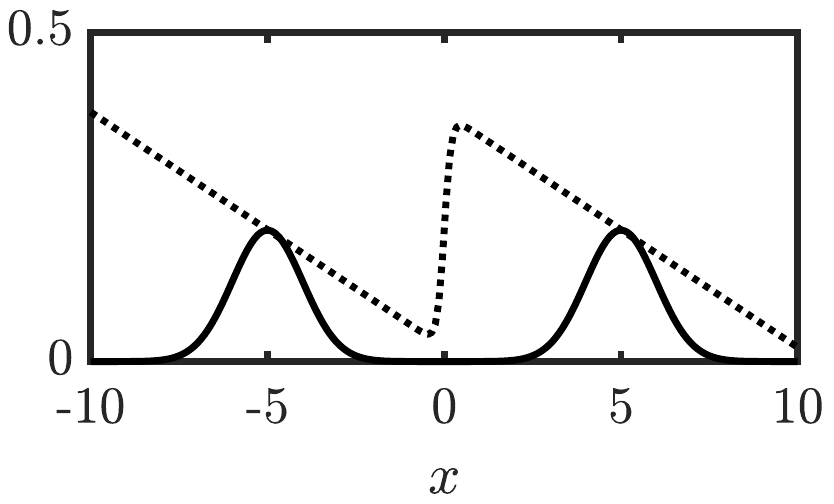}}
	\subcaptionbox{$\theta=0.8, \mu = 5$}{\includegraphics[width = 0.24\textwidth]{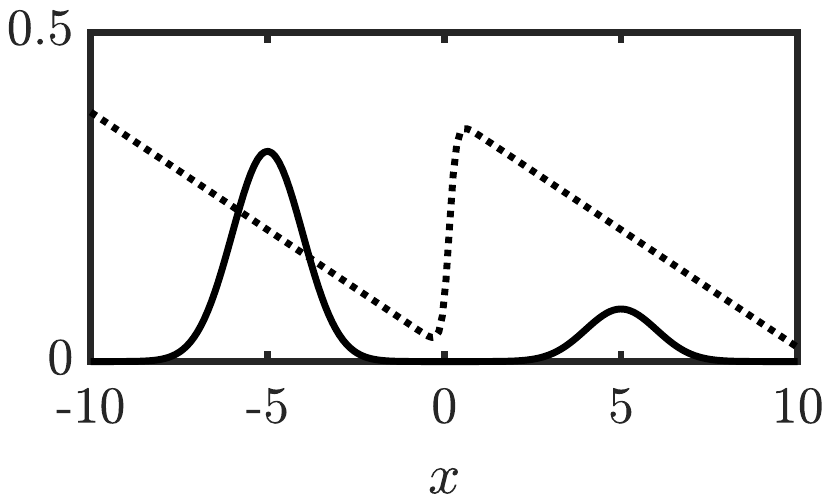}}
	\subcaptionbox{$\operatorname{KSD}^2$ for $\mu = 5$}{\includegraphics[width = 0.24\textwidth]{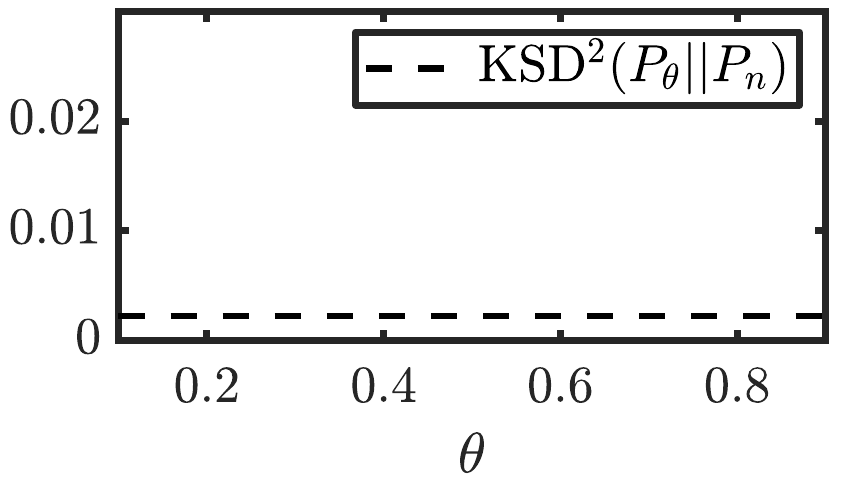}}
	\subcaptionbox{$\theta=0.2, \mu = 2$}{\includegraphics[width = 0.24\textwidth]{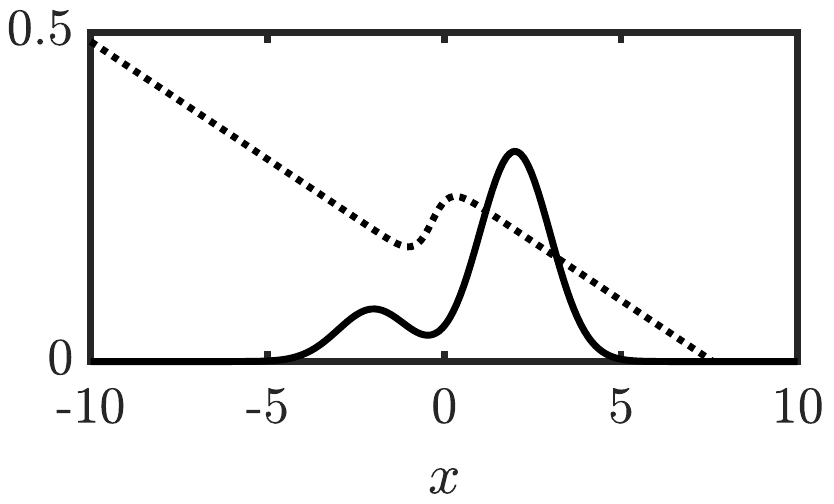}}
	\subcaptionbox{$\theta=0.5, \mu = 2$}{\includegraphics[width = 0.24\textwidth]{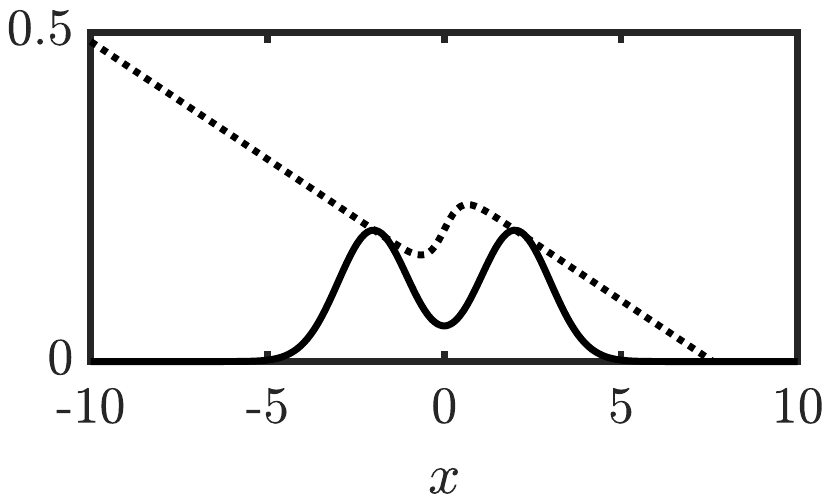}}
	\subcaptionbox{$\theta=0.8, \mu = 2$}{\includegraphics[width = 0.24\textwidth]{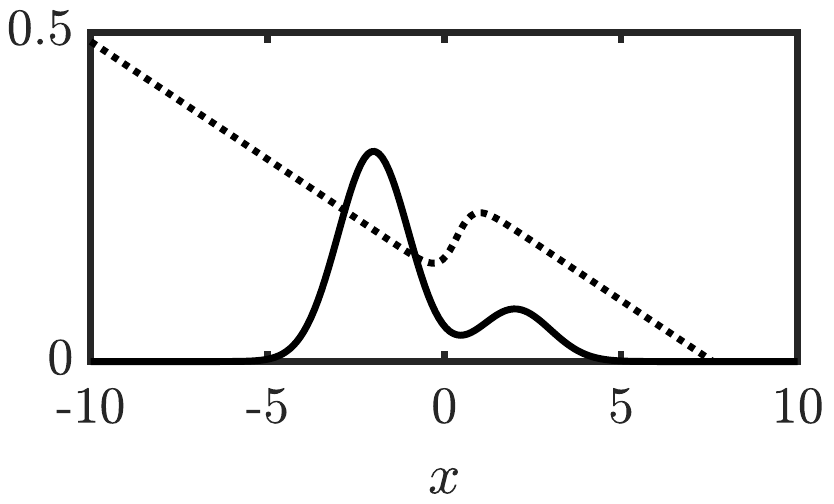}}
	\subcaptionbox{$\operatorname{KSD}^2$ for $\mu = 2$}{\includegraphics[width = 0.24\textwidth]{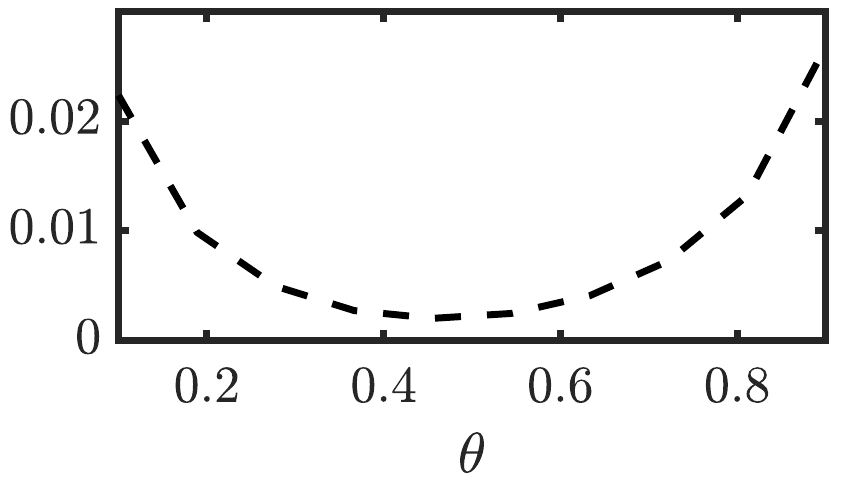}}
	\caption{Illustrating the insensitivity to mixture proportions of KSD.
	Panels (a-c,e-g) display the density function $p_\theta(x)$ from \eqref{eq: gmm} together with the gradient $\nabla \log p_\theta(x)$, the latter rescaled to fit onto the same plot. 
	Panels (d,h) display the discrepancy $\KSD^2(\P_\theta \| \P_n)$, where $\P_n$ is an empirical distribution of $n=1000$ samples from the model with $\theta=0.5$. 
	}
	\label{fig:score_mm_example}
\end{figure}

A second limitation of KSD-Bayes is non-invariance to a change of coordinates in the dataset.
This is a limitation of loss-based estimators in general. \textcolor{black}{In \Cref{subsec: default SO and K} we recommend a data-adaptive choice of kernel, which serves to provide approximate invariance to affine transformations of the dataset.
As usual in statistical analyses, we recommend \textit{post-hoc} assessment of the sensitivity of inferences to perturbations of the dataset.}

Despite these two limitations, KSD-Bayes represents a flexible and effective procedure for generalised Bayesian inference in the context of an intractable likelihood.
Our attention turns next to theoretical analysis of KSD-Bayes.


\section{Theoretical Assessment} \label{sec:theory}

This section contains a comprehensive theoretical treatment of KSD-Bayes.
The main results are \emph{posterior consistency} and a \emph{Bernstein--von Mises} theorem in \Cref{sec:pc_bvm}, and \emph{global bias-robustness} of the generalised posterior in \Cref{sec:robustness}.
In obtaining these results we have developed novel intermediate results concerning an important V-statistic estimator for KSD; these are anticipated to be of independent interest, so we present these in \Cref{sec:prop_KSD} of the main text.
Note that all theory is valid for the misspecified regime where $\P$ need not be an element of $\{\P_\theta : \theta \in \Theta\}$.
Moreover, the results in \Cref{sec:prop_KSD} and \Cref{sec:pc_bvm} hold for general data domains $\X$.
For the entirety of this section we set $\beta = 1$, with all results for $\beta \ne 1$ immediately recovered by replacing $K$ with $\beta K$.
The results of this section motivate a specific choice for $\beta$ that is described in \Cref{subsec: choice of k}.

\vspace{5pt}
\noindent \textbf{Standing Assumptions 2:}
The dataset $\{ x_i \}_{i=1}^{n}$ consists of independent samples generated from $\P \in \mathcal{P}(\X)$, with empirical distribution denoted $\P_n := (1/n) \sum_{i=1}^n \delta_{x_i}$.
The set $\Theta \subseteq \R^p$ is open, convex and bounded\footnote{It simplifies presentation to assume the parameter set $\Theta$ is bounded; there is no loss of generality since re-parametrisation can be performed.}.
\Cref{asmp:derivation_KSD} holds with $\Q = \P_{\theta}$ for every $\theta \in \Theta$.

\vspace{5pt}
\noindent \textbf{Notation:}
For shorthand, let $\partial^1$, $\partial^2$ and $\partial^3$ denote the partial derivatives $(\partial / \partial \theta_{(h)})$, $(\partial^2 / \partial \theta_{(h)} \partial \theta_{(k)})$ and $(\partial^3 / \partial \theta_{(h)} \partial \theta_{(k)} \partial \theta_{(l)})$ for $h,k,l \in \{ 1, \dots, p \}$, where to reduce notation the indices $(h,k,l)$ are left implicit.
The gradient and Hessian operators are $[\nabla_{\theta}]_{(h)} = (\partial / \partial \theta_{(h)})$ and $[\nabla_{\theta}^2]_{(h,k)} = (\partial^2 / \partial \theta_{(h)} \partial \theta_{(k)})$.

\subsection{Minimum KSD Estimators} \label{sec:prop_KSD}

First we present novel analysis of the V-statistic in \eqref{eq:empirical_KSD}.
Note that a U-statistic estimator of KSD was analysed in \cite{Barp2019}, but only for \textcolor{black}{the so-called \textit{diffusion} Stein operator, a variant (or \textit{standardisation}) of the Langevin Stein operator in \eqref{eq:stein_operator}.}
Our results for the V-statistic do not depend on a specific form of $\S_{\P_{\theta}}$, and may hence be of independent interest.

Despite the bias present in a V-statistic, our standing assumptions are sufficient to derive the following consistency result:

\begin{lemma}[a.s. Pointwise Convergence] \label{thm:pw}
For each $\theta \in \Theta$,
	\begin{align*}
	\operatorname{KSD}^2(\P_{\theta} \| \P_n) - \operatorname{KSD}^2(\P_{\theta} \| \P) \overset{a.s.}{\longrightarrow} 0 .
	\end{align*}
\end{lemma}

\noindent The proof is contained in \Cref{apx:proof_pw}.
If we impose further regularity, we can obtain a uniform convergence result.
\textcolor{black}{It will be convenient to introduce a collection of assumptions that are indexed by $r_{\max} \in \{0,1,2,\dots\}$, as follows: }
\begin{assumption}[$r_{\max}$] \label{asmp:an_cnd}
	\textcolor{black}{For all integers $0 \leq r \leq r_{\max}$, the following conditions hold: }
	\begin{enumerate}
		\item[(1)] the map $\theta \mapsto \partial^r \S_{\P_{\theta}}[h](x)$ exists and is continuous, for all $h \in \H$ and $x \in \X$;
		\item[(2)] the map $h \mapsto ( \partial^r \S_{\P_\theta} )[h](x)$ is a continuous linear functional on $\H$, for each $x \in \X$;
		\item[(3)] $\E_{X \sim \P}[ \sup_{\theta \in \B} ( ( \partial^r \S_{\P_{\theta}} ) ( \partial^r \S_{\P_\theta} ) K(X, X) ) ] < \infty$,
	\end{enumerate}
	where $( \partial^0 \S_{\P_\theta} ) := \S_{\P_\theta}$; note that (2) with $r = 0$ is implied from Standing Assumption 2.
\end{assumption}
\noindent In the expression above, the first and second $( \partial^r \S_{\P_{\theta}} )$ are applied, respectively, to the first and second argument of $K$, as with $\S_{\P_{\theta}} \S_{\P_{\theta}} K(x, x)$.
These assumptions become concrete when considering a specific Stein operator; the case of the Langevin Stein operator is presented in \Cref{ap: verify A3}.

\begin{lemma}[a.s. Uniform Convergence] \label{thm:uc}
	\textcolor{black}{Suppose \Cref{asmp:an_cnd} ($r_{\max}=1$) holds.}
	Then
	\begin{align*}
	\sup_{\theta \in \B} \left| \operatorname{KSD}^2(\P_{\theta} \| \P_n) - \operatorname{KSD}^2(\P_{\theta} \| \P) \right| \overset{a.s.}{\longrightarrow} 0 .
	\end{align*}
\end{lemma}

\noindent The proof is contained in \Cref{apx:proof_uc}.

Our next results concern consistency and asymptotic normality of the estimator $\theta_n$ that minimises the V-statistic in \eqref{eq:empirical_KSD}.

\begin{assumption} \label{asmp:sc_cnd}
	There exist minimisers $\theta_n$ of $\operatorname{KSD}(\P_\theta \| \P_n)$ for all sufficiently large $n \in \mathbb{N}$,
	and there exists a unique $\theta_*$ s.t. $\operatorname{KSD}(\P_{\theta_*} \| \P) < \inf_{\{ \theta \in \Theta : \| \theta - \theta_*\|_2 \geq \epsilon\}} \operatorname{KSD}(\P_\theta \| \P)$ for any $\epsilon > 0$.
\end{assumption}	

\begin{lemma}[Strong Consistency] \label{lem:sc_ksd}
	\textcolor{black}{Suppose Assumptions \ref{asmp:an_cnd} ($r_{\max}=1$) and \ref{asmp:sc_cnd} hold. }
	Then 
	$$
	\theta_n \overset{a.s.}{\longrightarrow}  \theta_*.
	$$
\end{lemma}

\noindent The proof is contained in \Cref{apx:proof_sc_ksd}.
For the well-specified case where $\exists \theta_0$ s.t. $\P_{\theta_0} = \P$, the uniqueness of $\theta_*$ holds automatically if $\operatorname{KSD}$ is a proper divergence i.e. $\operatorname{KSD}(\P \| \Q) = 0 \Longleftrightarrow \P = \Q$.
For example, if the preconditions of \citet[Proposition~1]{Barp2019} are satisfied and the parametrisation $\theta \mapsto \P_{\theta}$ is injective, the minimum is uniquely attained.

Asymptotic normality of $\theta_n$ can be established if further regularity is imposed: 

\begin{lemma}[Asymptotic Normality] \label{lem:an_ksd}
	\textcolor{black}{Suppose Assumptions \ref{asmp:an_cnd} ($r_{\max}=3$) and \ref{asmp:sc_cnd} hold.}
	Let $H_* := \nabla_{\theta}^2 \operatorname{KSD}^2(\P_{\theta} \| \P) |_{\theta = \theta_*}$ and $J_* := \mathop{\E}_{X \sim \P}[ S(X, \theta_*) S(X, \theta_*)^\top]$, where we define the column vector $S(x, \theta) := \mathop{\E}_{X \sim \P}[ \nabla_{\theta} ( \S_{\P_{\theta}} \S_{\P_{\theta}} K(x, X) ) ]$. 
	If $H_*$ is non-singular, $$\sqrt{n} \left( \theta_n - \theta_* \right) \overset{d}{\to} \mathcal{N}(0, H_*^{-1} J_* H_*^{-1})$$ where $\overset{d}{\to}$ denotes the convergence in distribution.
\end{lemma}

\noindent The proof is contained in \Cref{apx:proof_an_ksd}.
These preliminaries on minimum KSD estimation are required for our main results on KSD-Bayes, presented next.

\subsection{Posterior Consistency and Bernstein-von-Mises} \label{sec:pc_bvm}

Armed with the technical results of \Cref{sec:prop_KSD}, we can now establish consistency of KSD-Bayes and a Bernstein--von Mises result. 
Our consistency result requires a \textit{prior mass condition}, similar to that of \citet{Cherief-Abdellatif2019}:

\begin{assumption} \label{asmp:pc_cnd}
The prior is assumed to
\begin{enumerate}
    \item admit a p.d.f. $\pi$ that is continuous at $\theta_*$, with $\pi(\theta_*) > 0$;
    \item satisfy $\int_{B_n(\alpha_1)}  \pi(\theta) \mathrm{d} \theta \ge e^{- \alpha_2 \sqrt{n}}$ for some constants $\alpha_1, \alpha_2 >0$,
\end{enumerate}
where we define $B_n(\alpha_1) := \{ \theta \in \Theta : | \operatorname{KSD}^2(\P_\theta \| \P) - \operatorname{KSD}^2(\P_{\theta_*} \| \P) | \le \alpha_1/\sqrt{n} \}$.
\end{assumption}

\noindent \Cref{asmp:pc_cnd} specifies the amount of prior mass in a neighbourhood around the population-optimal value $\theta_*$ that is required.
This is not a strong assumption and \Cref{apx: expfam_asmp} demonstrates how each of \Cref{asmp:sc_cnd,asmp:an_cnd,asmp:pc_cnd} can be verified in the case of an exponential family model.

\begin{theorem}[Posterior Consistency] \label{thm:pc}
	Suppose \Cref{asmp:sc_cnd,asmp:pc_cnd} hold.
	Let $\sigma(\theta) := \E_{X \sim \P}\left[ \S_{\P_{\theta}} \S_{\P_{\theta}} K(X, X) \right]$.
	Then, for all $\delta \in (0,1]$, 
	\begin{align*}
	\mathbb{P} \left( \left| \int_\Theta  \operatorname{KSD}^2(\P_\theta \| \P) \pi_{n}^{D}(\theta) \mathrm{d} \theta - \operatorname{KSD}^2(\P_{\theta_{*}} \| \P) \right| > \delta \right) \le \frac{\alpha_1 + \alpha_2 + 8 \sup_{\theta \in \Theta} \sigma(\theta)}{\delta \sqrt{n}} 
	\end{align*}
	where the probability is with respect to realisations of the dataset $\{ x_i \}_{i=1}^{n} \overset{i.i.d.}{\sim} \P$.
\end{theorem}

\noindent The proof is contained in \Cref{apx:proof_pc}.

Next, we derive a Bernstein--von Mises result.
The pioneering work of \cite{hooker2014bayesian} and \cite{Ghosh2016} established Bernstein--von Mises results for generalised posteriors defined by $\alpha$- and $\beta$-divergences.
Unfortunately, the form of KSD is rather different and different theoretical tools are required to tackle it.
\textcolor{black}{
\cite{Miller2019} introduced a general approach to deriving Bernstein--von Mises results for generalised posteriors, demonstrating how the assumptions can be verified for several additive loss functions $L_n$.
Our proof builds on \cite{Miller2019}, demonstrating that the required assumptions can also be satisfied by the non-additive KSD loss function in \eqref{eq:empirical_KSD}.
}

\begin{theorem}[Bernstein--von Mises] \label{thm:lan}
	\textcolor{black}{Suppose \Cref{asmp:an_cnd} ($r_{\max}=3$), \ref{asmp:sc_cnd}, and part (1) of \ref{asmp:pc_cnd} hold.}
	Let $\hat{\pi}_n^D$ the p.d.f. of the random variable $\sqrt{n} (\theta - \theta_n)$ for $\theta \sim \pi_n^D$, viewed as a p.d.f. on $\R^p$.
	Let $H_* := \nabla_{\theta}^2 \operatorname{KSD}^2(\P_{\theta} \| \P) |_{\theta = \theta_*}$.
	If $H_*$ is nonsingular,
	\begin{align*}
	\int_{\R^p} \left| \hat{\pi}_n^D(\theta) - \frac{1}{\det( 2 \pi H_*^{-1}  )^{1/2}} \exp\left( - \frac{1}{2} \theta \cdot H_* \theta \right) \right| \mathrm{d} \theta  \overset{a.s.}{\longrightarrow} 0 ,
	\end{align*}
	where the a.s. convergence is with respect to realisations of the dataset $\{ x_i \}_{i=1}^{n}$.
\end{theorem}

\noindent The proof is contained in \Cref{apx:proof_lan}.
These positive results are encouraging, as they indicate the limitations of KSD-Bayes described in \Cref{subsec: pathologies} are at worst a finite sample size effect.
\textcolor{black}{However, we note that the asymptotic precision matrix $H_*$ from \Cref{thm:lan} differs to the precision matrix $H_* J_*^{-1} H_*$ of the minimum KSD estimator from \Cref{lem:an_ksd}; this is analogous to fact that Bayesian credible sets can have asymptotically incorrect frequentist coverage if the statistical model is mis-specified \citep{kleijn2012bernstein}.
This point will be addressed in \Cref{subsec: beta setting}.
}

\textcolor{black}{
	\begin{remark}
		The analysis in \Cref{sec:prop_KSD,sec:pc_bvm} covers general domains $\X$ and Stein operators $\S_\P$.
		Henceforth, in the main text we restrict attention to $\X = \R^d$, but the case of a discrete domain $\X$, and the identification of an appropriate Stein operator in this context, are discussed in \Cref{subsec: other-space}.
	\end{remark}
}

\subsection{Global Bias-Robustness of KSD-Bayes} \label{sec:robustness} 

An important property of KSD-Bayes is that, through a suitable choice of kernel, the generalised posterior can be made robust to contamination in the dataset.
This robustness will now be rigorously established.

Consider the \emph{$\varepsilon$-contamination model} $\P_{n,\epsilon,y} = (1 - \epsilon) \P_n + \epsilon \delta_{y}$, where $y \in \mathcal{X}$ and $\epsilon \in [0,1]$ \citep[see][]{Huber2009}.
In other words, the datum $y$ is considered to be contaminating the dataset $\{x_i\}_{i=1}^n$.
Robustness in the generalised Bayesian setting has been considered in \cite{hooker2014bayesian,Ghosh2016,Nakagawa2020}. 
In what follows we write $L_n(\theta) = L(\theta; \P_n)$ to make explicit the dependence of the loss function $L_n$ on the dataset $\P_n$.
Following \cite{Ghosh2016}, we consider a generalised posterior based on a (contaminated) loss $L(\theta; \P_{n,\epsilon,y})$ with density $\pi_n^L(\theta ; \P_{n,\epsilon,y})$, and define the \emph{posterior influence function}
\begin{align}
\operatorname{PIF}(y, \theta, \P_n) & :=  \frac{\mathrm{d}}{\mathrm{d} \epsilon} \pi_n^L(\theta ; \P_{n,\epsilon,y}) |_{\epsilon = 0} . \label{eq: def pif}
\end{align}
Here the notation $\pi_n^L(\theta ; \P_{n,\epsilon,y})$ emphasises the dependence of the generalised posterior on the (contaminated) dataset $\P_{n,\epsilon,y}$.
A generalised posterior $\pi_n^L$ is called \emph{globally bias-robust} if $\sup_{\theta \in \Theta} \sup_{y \in \mathcal{X}} | \operatorname{PIF}(y, \theta, \P_n) | < \infty$, meaning that the sensitivity of the generalised posterior to the contaminant $y$ is limited.
The following lemma provides general sufficient conditions for global bias-robustness to hold:

\begin{lemma} \label{lem:bias-robust_cnd}
	Let $\pi_n^L$ be a generalised Bayes posterior for a fixed $n \in \N$ with a loss $L(\theta; \P_n)$ and a prior $\pi$. 
	Suppose $L(\theta; \P_n)$ is lower-bounded and $\pi(\theta)$ is upper-bounded over $\theta \in \Theta$, for any $\P_n$.
	Denote $\operatorname{D}L(y, \theta, \P_n) := (\mathrm{d} / \mathrm{d} \epsilon) L(\theta; \P_{n,\epsilon,y}) |_{\epsilon=0}$.
	Then $\pi_n^L$ is globally bias-robust if, for any $\P_n$,
	\begin{enumerate}
		\item $\sup_{\theta \in \Theta} \sup_{y \in \X} \left| \operatorname{D}L(y, \theta, \P_n ) \right| \pi(\theta) < \infty$, and
		\item $\int_{\Theta} \sup_{y \in \X} \left| \operatorname{D}L(y, \theta , \P_n ) \right| \pi(\theta) \mathrm{d} \theta < \infty$.
	\end{enumerate}
\end{lemma}

\noindent The proof is contained in \Cref{apx:proof_bias-robust_cnd}.
\textcolor{black}{Note that standard Bayesian inference does not satisfy the conditions of \Cref{lem:bias-robust_cnd} in general.
Indeed, when $L(\theta; \P_n)$ is the negative log likelihood, $\operatorname{D}L(y, \theta, \P_n) = \log p_\theta(y) - \sum_{i=1}^{n} \log p_\theta(x_i)$, and the term $\log p_\theta(y)$ can be unbounded over $y \in \X$.
This can occur even if the statistical model is not heavy-tailed, e.g. for a normal location model $p_\theta$ on $\X = \R^d$.
In contrast, the kernel $K$ in KSD-Bayes provides a degree of freedom which can be leveraged to ensure that the conditions of \Cref{lem:bias-robust_cnd} \textit{are} satisfied; }
the specific form of $\operatorname{D}L(y, \theta, \P_n)$ for KSD-Bayes is derived in \Cref{sec:GD_KSD}.
This enables us to derive sufficient conditions \textcolor{black}{on $K$} for global bias-robustness of KSD-Bayes, which we now present.

\begin{theorem}[Globally Bias-Robust] \label{thm:bias-robust}
	For each $\theta \in \Theta$, let $\P_{\theta} \in \mathcal{P}_{\normalfont \text{S}}(\R^d)$ and let $\S_{\P_{\theta}}$ denote \textcolor{black}{the Langevin Stein operator in \eqref{eq:stein_operator}}.
	Let $K \in C_b^{1,1}(\R^d \times \R^d; \R^{d \times d})$. 
	Suppose that $\pi$ is bounded over $\Theta$.
	If there exists a function $\gamma: \Theta \to \R$ such that
	\begin{align}
	\sup_{y \in \R^d} \Big( \nabla_{y} \log p_{\theta}(y) \cdot K(y, y) \nabla_{y} \log p_{\theta}(y) \Big) \le \gamma(\theta)  \label{eq: main robust cdn}
	\end{align}
	and, in addition, $\sup_{\theta \in \Theta} | \pi(\theta) \gamma(\theta) | < \infty$ and $\int_\Theta \pi(\theta) \gamma(\theta) \mathrm{d} \theta < \infty$, then KSD-Bayes is globally bias-robust.
\end{theorem}

\noindent The proof is contained in \Cref{sec:proof_bias-robust}.
\textcolor{black}{The preconditions of \Cref{thm:bias-robust} can be satisfied through an appropriate choice of kernel $K$; see \Cref{subsec: default SO and K}.}
\textcolor{black}{A comparison of KSD-Bayes to existing robust generalised Bayesian methodologies for tractable likelihood can be found in \Cref{subsec: gen bayes comparison}.}
The difference in performance of robust and non-robust instances of KSD-Bayes is explored in detail in \Cref{sec:experiment}.


\section{Default Settings for KSD-Bayes} \label{subsec: choice of k}

The previous section considered $\beta$ to be fixed, but an appropriate selection of $\beta$ is essential to ensure the generalised posterior is calibrated.
The choice of $\beta$ is closely related to the choice of a Stein operator $\S_{\P_\theta}$ and kernel $K$; the purpose of this section is to recommend how these quantities are selected.
If the recommendations of this section are followed, then KSD-Bayes has no remaining degrees of freedom to be specified.

\subsection{Default Settings for $\S_{\P_\theta}$ and $K$}
\label{subsec: default SO and K}

For Euclidean domains $\X = \R^d$, we advocate the default use of \textcolor{black}{the Langevin Stein operator $\S_{\P_\theta}$ in \eqref{eq:stein_operator} and a kernel of the form
\begin{equation} 
	K(x,x') = \frac{M(x) M(x')^\top}{ \left(1 + (x-x')^\top \Sigma^{-1} (x-x') \right)^{\gamma} } , \label{eq: suggested kernel}
\end{equation}
where $\Sigma$ is a positive definite matrix, $\gamma \in (0,1)$ is a constant, and $M \in C_b^1(\R^d; \R^{d \times d})$ will be called a matrix-valued \textit{weighting function}\footnote{\color{black} The use of a non-constant weighting function is equivalent to replacing the Langevin Stein operator with a \textit{diffusion} Stein operator whose \textit{diffusion matrix} is $M(x)$; see \citet{Gorham2019}.}.}
\textcolor{black}{
	For $M(x) = I_d$, \eqref{eq: suggested kernel} is called an \textit{inverse multi-quadratic} (IMQ) kernel.
}
The IMQ kernel and the Langevin Stein operator have appealing properties in the context of KSD.
Firstly, under mild conditions on $\mathbb{P}$, $\text{KSD}(\mathbb{P} || \mathbb{P}_n) \rightarrow 0$  implies that $\mathbb{P}_n$ converges weakly to $\mathbb{P}$ \citep[][Theorem 4]{chen2019stein}.
This convergence control ensures that small values of $\text{KSD}(\P_\theta \| \P_n)$ imply similarity between $\P_\theta$ and $\P_n$ in the topology of weak convergence, so that minimising KSD is meaningful\footnote{Note that other common kernels (e.g., Gaussian or Mat\'{e}rn kernels) fail to provide convergence control \citep[][Theorem 6]{Gorham2017}.}.
Secondly, and on a more practical level, \textcolor{black}{the combination of Stein operator and IMQ kernel}, with $\gamma = 1/2$, was found to work well in previous studies \citep{chen2019stein,riabiz2020optimal}; we therefore also recommend $\gamma = 1/2$ as a default.
\textcolor{black}{
	The weighting function $M(x)$ facilitates an efficiency-robustness trade-off:
	If global bias robustness is \emph{not} required then we recommend setting $M(x) = I_d$ as a default, which enjoys the aforementioned properties of KSD.
	If global bias-robustness \emph{is} required then we recommend selecting $M(x)$ such that the supremum in \eqref{eq: main robust cdn} exists and the preconditions of \Cref{thm:bias-robust} are satisfied; see the worked examples in \Cref{sec:experiment} and the further discussion in \Cref{subsec: eff robust}.
}

\textcolor{black}{The theoretical analysis of \Cref{sec:theory} assumed that $K$ is fixed, but in our experiments we follow standard practice in the kernel methods community and recommend a data-adaptive choice of the matrix $\Sigma$.
All experiments we report used the $\ell_1$-regularised sample covariance matrix estimator of \citet{ollila2019optimal}.
The sensitivity of KSD-Bayes to the choice of kernel parameters is investigated in \Cref{subsec: sensitivity kernel choice}.
}

\subsection{Default Setting for $\beta$}
\label{subsec: beta setting}

For a simple normal location model, as described in \Cref{sec:experiment-nl}, and in a well-specified setting, the asymptotic variance of the KSD-Bayes posterior with $\beta = 1$ is never smaller than that of the standard posterior.
This provides a heuristic motivation for the default $\beta = 1$.
However, in a misspecified setting smaller values of $\beta$ are needed to avoid over-confidence in the generalised posterior, taking misspecification into account; see the recent review of \cite{Wu2020}.
Here we aim to pick $\beta$ such that the scale of the asymptotic precision matrix of the generalised posterior ($H_*$; \Cref{thm:lan}) matches that of the minimum KSD point estimator ($H_* J_*^{-1} H_*$; \Cref{lem:an_ksd}), an approach proposed in \cite{Lyddon2018}.
\textcolor{black}{This ensures the scale of the generalised posterior matches the scale of the sampling distribution of a closely related estimator whose frequentist properties can be analysed when the statistical model is misspecified.}
Since $\P$ is unknown, estimators of $H_*$ and $J_*$ are required. 
We propose the following default for $\beta$:
\begin{align}
\beta = \min\left( 1, \beta_n \right) \quad \text{ where } \quad \beta_n = \frac{\text{tr}(H_n J_n^{-1} H_n)}{\text{tr}(H_n)} , \label{eq:beta_choice_approx}
\end{align}
where the matrix $H_*$ is approximated using $H_n := \nabla_\theta^2 \operatorname{KSD}^2(\mathbb{P}_{\theta} \| \mathbb{P}_n) \big|_{\theta = \theta_n}$, and the matrix $J_*$ is approximated using
\begin{align*}
J_n := \frac{1}{n} \sum_{i=1}^n S_n(x_i, \theta_n) S_n(x_i, \theta_n)^\top ,
\quad S_n(x,\theta) := \frac{1}{n} \sum_{i=1}^n \nabla_\theta ( \mathcal{S}_{\mathbb{P}_\theta} \mathcal{S}_{\mathbb{P}_\theta} K(x,x_i) ) .
\end{align*}
\textcolor{black}{
The minimum of $\beta = 1$ and $\beta = \beta_n$ taken in \eqref{eq:beta_choice_approx} provides a safeguard against selecting a value of $\beta$ that over-shrinks the posterior covariance matrix --- a phenomenon that we observed for the experiments reported in \Cref{subsec: scale est,subsec: kernel family,subsec: exponential model}, due to poor quality of the approximations $H_n$ and $J_n$ when $n$ is small.
}
The above expressions are derived for the exponential family model in \Cref{apx: expfam_asmp}.

\vspace{5pt}
This completes our methodological and theoretical development, and next we turn to empirical performance assessment.


\section{Empirical Assessment} \label{sec:experiment}

In this section four distinct experiments are presented.
The first experiment, in \Cref{sec:experiment-nl}, concerns a normal location model, allowing the standard posterior and our generalised posterior to be compared and confirming our robustness results are meaningful.
\Cref{subsec: scale est} presents a two-dimensional precision estimation problem, where standard Bayesian computation is challenging but computation with KSD-Bayes is trivial. 
Then, \Cref{subsec: kernel family} presents a 25-dimensional kernel exponential family model, and \Cref{subsec: exponential model} presents a 66-dimensional exponential graphical model; in both cases a Bayesian analysis has not, to-date, been attempted due to severe intractability of the likelihood.
In addition, the kernel exponential family model allows us to explore a multi-modal dataset and to understand the potential limitations of KSD-Bayes in that context (c.f. \Cref{subsec: pathologies}). 
For all experiments, the default settings of \Cref{subsec: choice of k} were used.
\textcolor{black}{An example of KSD-Bayes applied to a discrete dataset is presented in \Cref{subsec: other-space}.}

\subsection{Normal Location Model} \label{sec:experiment-nl}

For expositional purposes we first consider fitting a normal location model $\mathbb{P}_\theta = \mathcal{N}(\theta,1)$ to a dataset \textcolor{black}{$\{x_i\}_{i=1}^n$}.
Our aim is to illustrate the robustness properties of KSD-Bayes, and we therefore \textcolor{black}{generated the dataset using a} contaminated data-generating model where, \textcolor{black}{for each index $i = 1,\dots,n$ independently, with probability $1 - \epsilon$ the datum $x_i$ was drawn from $\P_\theta$ with ``true'' parameter $\theta = 1$, otherwise $x_i$ was drawn from $\mathbb{P}_y = \mathcal{N}(y,1)$, } so that $y$ and $\epsilon$ control, respectively, the nature and extent of the contamination in the dataset.
The task is to make inferences for $\theta$ based on a contaminated dataset of size $n = 100$.
The prior on $\theta$ was $\mathcal{N}(0,1)$.

\begin{figure}[t!]
	\centering
	\includegraphics[width = 0.9\textwidth,clip,trim = 0cm 10cm 0.5cm 10cm]{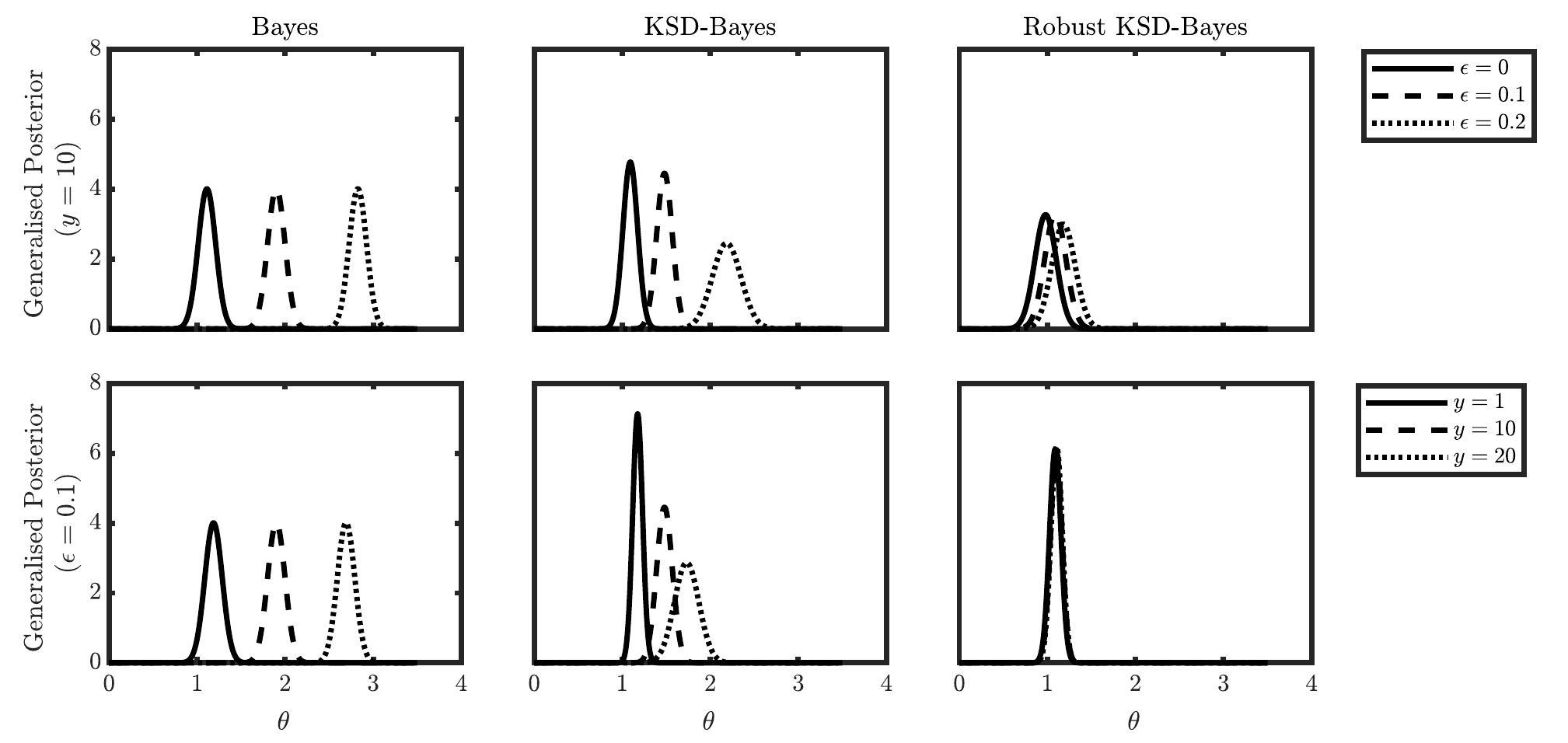}
	\caption{\textcolor{black}{Posteriors and generalised posteriors for the normal location model. The true parameter value is $\theta = 1$, while a proportion $\epsilon$ of the data were contaminated by noise of the form $\mathcal{N}(y,1)$.
	In the top row $y=10$ is fixed and $\epsilon \in \{0,0.1,0.2\}$ are considered, while in the bottom row $\epsilon = 0.1$ is fixed and $y \in \{1,10,20\}$ are considered.
	}
	}
	\label{fig:normal location model}
\end{figure}

The standard Bayesian posterior is depicted in the leftmost panels of \Cref{fig:normal location model}, for varying $\epsilon$ (top row) and varying $y$ (bottom row).
Straightforward calculation shows that the expected posterior mean is $\frac{n}{n+1} \left[ \theta  + \epsilon (y - \theta) \right]$, which increases linearly as either $y$ or $\epsilon$ are increased, with the other fixed.
This behaviour is evident in the leftmost panels of \Cref{fig:normal location model}.
The generalised posterior from KSD-Bayes is depicted in the central panels of \Cref{fig:normal location model}.
This generalised posterior is slightly less sensitive to contamination compared to the standard posterior.
Moreover, the variance slightly increases  whenever either $\epsilon$ or $y$ are increased, as a result of estimating $\beta$ (c.f. \Cref{subsec: beta setting}).
In the rightmost panels of \Cref{fig:normal location model} we display the robust generalised posterior using the weighting function $M(x) = (1 + x^2)^{-1/2}$, intended to bound the influence of large values in the dataset.
\textcolor{black}{This choice of $M(x)$ vanishes just fast enough as $|x| \rightarrow \infty$ to ensure that the bias-robustness conditions of \Cref{thm:bias-robust} are satisfied; see \Cref{subsec: eff robust}.}
The effect is clear from the bottom right panel of \Cref{fig:normal location model}, where even for $y = 20$ (and $\epsilon$ fixed to a small value, $\epsilon = 0.1$) the robust generalised posterior remains centred close to the true value $\theta = 1$.
While our theoretical results relate to $y$ and do not guarantee robustness when $\epsilon$ is increased, the top right panel in \Cref{fig:normal location model} suggests that the robust generalised posterior is indeed robust in this regime as well.
\Cref{fig:PIFy_ex} displays the posterior influence function \eqref{eq: def pif} for this normal location model.
This reveals that the standard Bayesian posterior is not bias-robust, since the tails of the posterior are highly sensitive to the contaminant $y$.
In contrast, the tails of the generalised posterior are insensitive to the contaminant.
This appears to be the case for both weighting functions, despite only one weighting function satisfying the conditions of \Cref{thm:bias-robust}.

\begin{figure}[t!]
	\centering
	\subcaptionbox{$\theta \mapsto | \operatorname{PIF}(y = 2.0, \theta, \P_n) |$}{\includegraphics[width=0.45\textwidth]{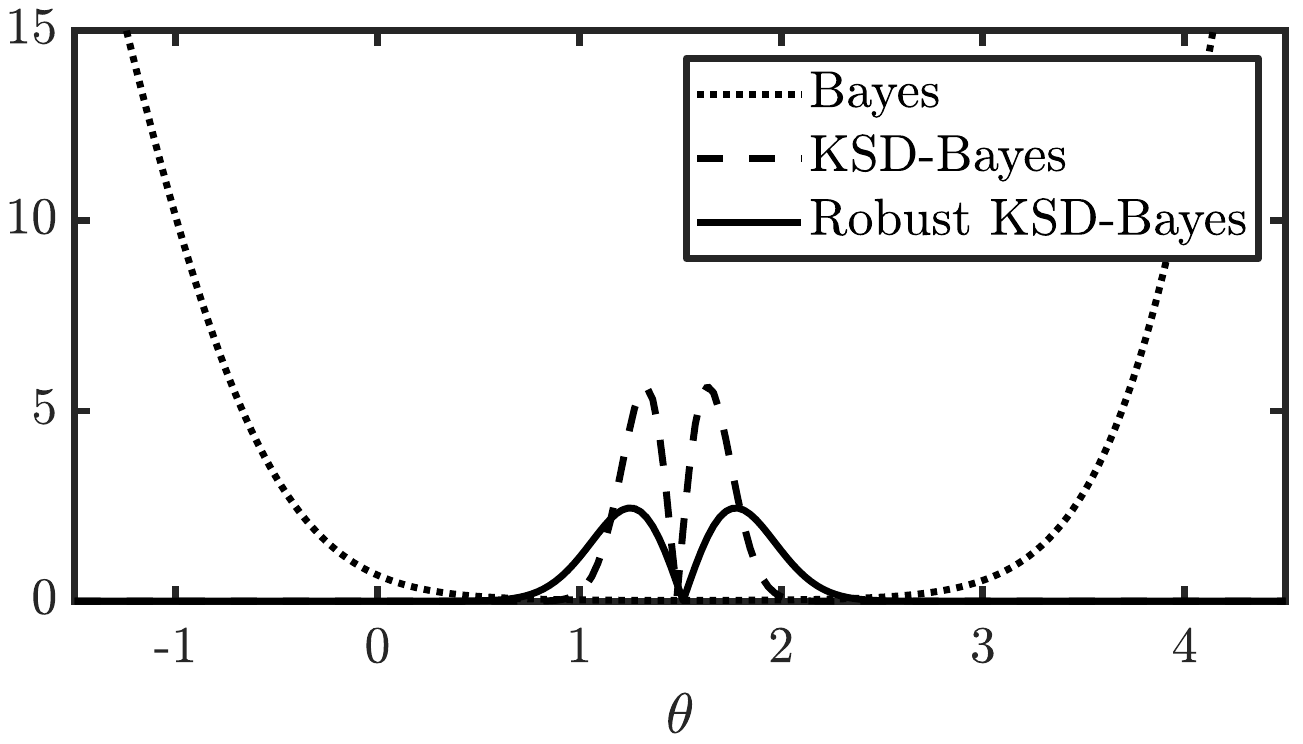}}
	\subcaptionbox{$\theta \mapsto | \operatorname{PIF}(y = 20, \theta, \P_n) |$}{\includegraphics[width=0.45\textwidth]{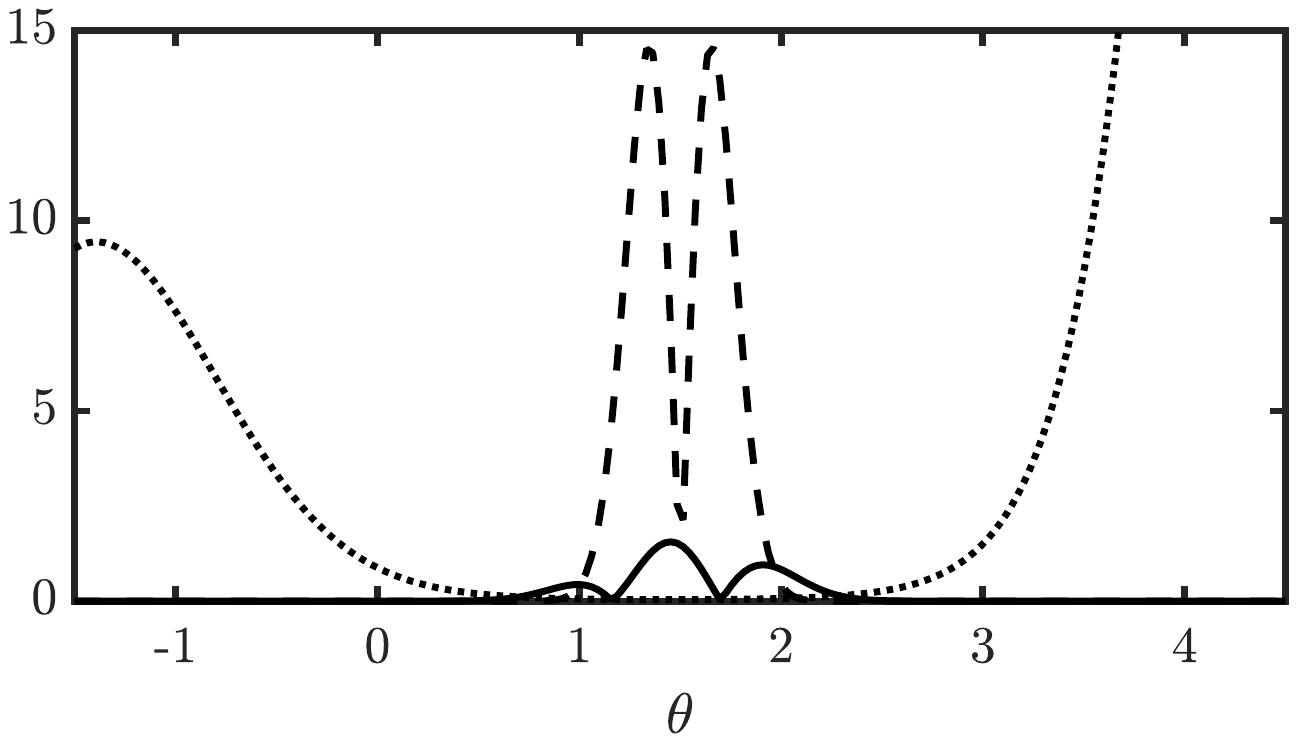}}
	\caption{Posterior influence function for the normal location model.}
	\label{fig:PIFy_ex}
\end{figure}

\subsection{Precision Parameters in an Intractable Likelihood Model} \label{subsec: scale est}

Our second experiment is due to \cite{liu2019fisher}, and concerns an exponential family model $p_\theta(x) = \exp( \theta \cdot t(x) - a(\theta) + b(x) )$, where $\theta \in \mathbb{R}^2$ are parameters to be inferred and $x \in \mathbb{R}^5$.
The model specification is completed with
\begin{align*}
t(x) & = (\tanh(x_{(4)}),\tanh(x_{(5)})) , \quad    b(x)  = \textstyle -0.5 \sum_{i=1}^5 x_{(i)}^2 + 0.6 x_{(1)} x_{(2)} + 0.2 \sum_{i=3}^5 x_{(1)} x_{(i)} .
\end{align*}
Despite the apparent simplicity of this model, the term $a(\theta)$, which determines the normalisation constant, is analytically intractable and exact simulation from this data-generating model is not straightforward (excluding the case $\theta = 0$). 
As a consequence, standard Bayesian analysis is not practical without, for example, the development of model-specific numerical methods, such as cubature rules to approximate the intractable normalisation constant.
In sharp contrast, the generalised posterior produced by KSD-Bayes is available in closed form for this model.
Our aim here is to assess robustness of the generalised posterior, focusing on the setting where $y$ is fixed and $\epsilon$ is increased, since this is the regime for which our theoretical results do \textit{not} hold.
A dataset of size $n=500$ was generated from 
the model $\P_\theta$ with true parameter $\theta = (0,0)$, so that $\mathbb{P}_\theta$ has the form $\mathcal{N}(0,\Sigma)$ and can be exactly sampled.
Each datum $x_i$ was, with probability $\epsilon$, shifted to $x_i + y$ where $y = (10,\dots,10)$.
The prior on $\theta$ was $\mathcal{N}(0,10^2 I)$.

\begin{figure}[t!]
	\centering
	\includegraphics[width = 0.8\textwidth,clip,trim = 4cm 10.5cm 4cm 8.5cm]{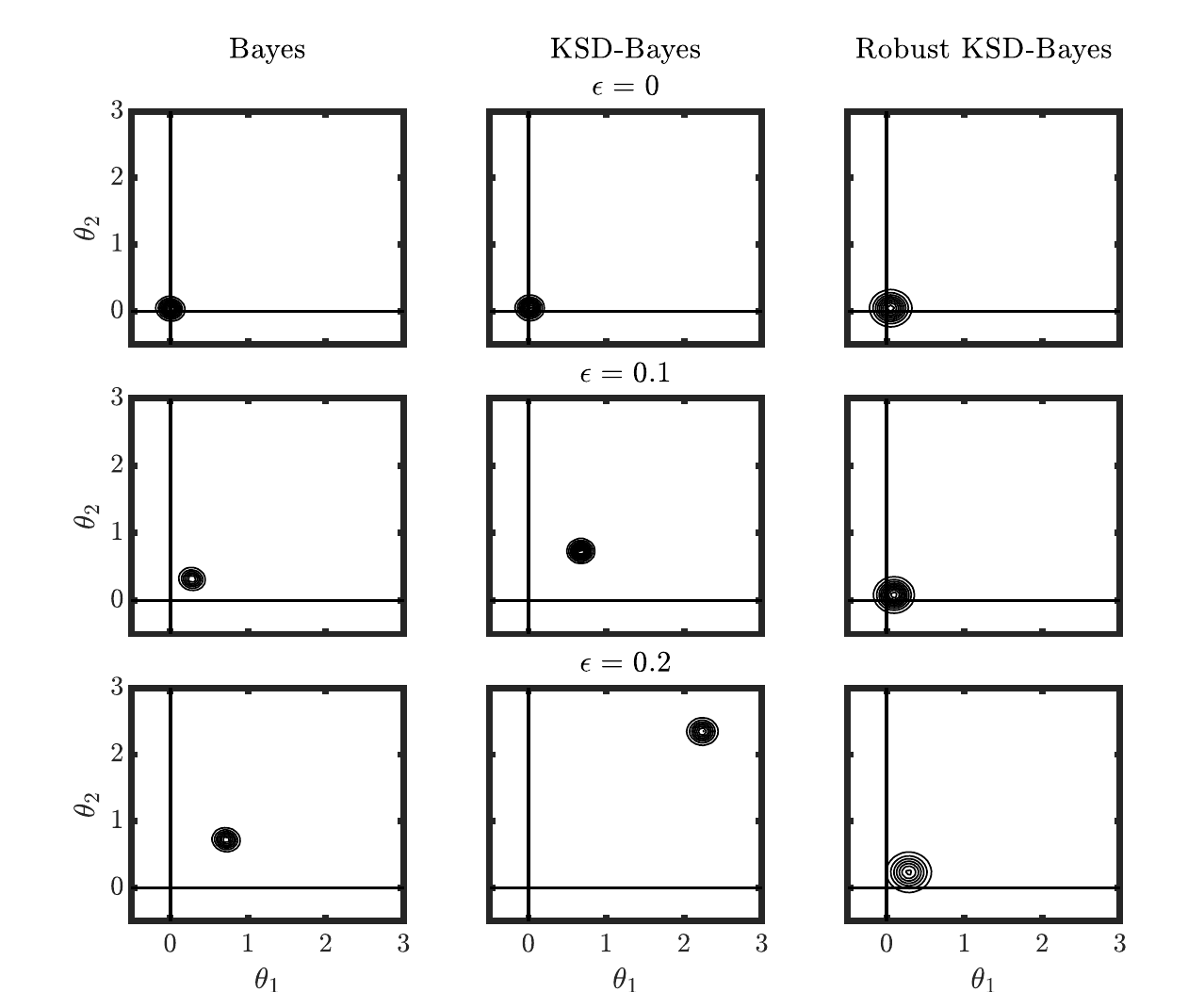}
	\caption{\textcolor{black}{Posteriors and generalised posteriors for the \cite{liu2019fisher} model. The true parameter value is $\theta = 0$, while a proportion $\epsilon$ of the data were contaminated by being shifted by an amount $y = (10,10)$.}}
	\label{fig: scale param estimation}
\end{figure}

The left column in \Cref{fig: scale param estimation} displays the standard posterior\footnote{To obtain these results, the intractable normalisation constant was approximated using a numerical cubature method.
To do this, we recognise that $p_\theta(x) = \mathcal{N}(x;0,\Sigma) r_\theta(x) / C_\theta$ where $r_\theta(x) = \exp(\theta_1 \tanh(x_4) + \theta_2 \tanh(x_5))$.
Then $C_\theta = \int r_\theta(x) \mathrm{d}\mathcal{N}(x;0,\Sigma)$, which was approximated using (polynomial order 10) Gauss-Hermite cubature in 2D.}, which is seen to be sensitive to contamination in the dataset, in much the same way observed for the normal location model in \Cref{sec:experiment-nl}.
The generalised posterior with $M(x) = I_d$ is depicted in the middle column of \Cref{fig: scale param estimation}, and is seen to be \textit{more} sensitive to contamination compared to the standard Bayesian posterior, in that the mean moves further from 0 as $\epsilon$ is increased.
Finally, in the right column of \Cref{fig: scale param estimation} we display the robust generalised posterior obtained with weighting function
\begin{align*}
M(x) = \text{diag}\left( (1+x_{(1)}^2+\dots+x_{(5)}^2)^{-1/2} , (1+x_{(1)}^2+x_{(2)}^2)^{-1/2} , \dots , (1+x_{(1)}^2+x_{(5)}^2)^{-1/2} \right) ,
\end{align*}
\textcolor{black}{which ensures the criteria for bias-robustness in \Cref{thm:bias-robust} are satisfied.}
From the figure, we observe that the robust generalised posterior remains centred close to the data-generating value $\theta = 0$, even for the largest contamination proportion considered ($\epsilon = 0.2$), with a variance that increases as $\epsilon$ is increased.
At $\epsilon = 0$, the spread of the robust generalised posterior is almost twice that of the standard posterior, which reflects the trade-off between robustness and efficiency.

\subsection{Robust Nonparametric Density Estimation} \label{subsec: kernel family}

Our third experiment concerns density estimation using the kernel exponential family, and explores the performance of KSD-Bayes when the dataset is multi-modal (c.f. \Cref{subsec: pathologies}).
Let $q$ denote a reference p.d.f. on $\mathbb{R}^d$, and let $\kappa : \mathbb{R}^d \times \mathbb{R}^d \rightarrow \mathbb{R}$ be a reproducing kernel.
The \textit{kernel exponential family} model  \citep{canu2006kernel}
\begin{align}
p_\theta(x) \propto q(x) \exp( \langle f , \kappa(\cdot,x) \rangle_{\mathcal{H}(\kappa)} ) \label{eq: KEF}
\end{align}
is parametrised by $f$, an element of the RKHS $\mathcal{H}(\kappa)$.
The implicit normalisation constant of \eqref{eq: KEF}, if it exists, is typically an intractable function of $f$.
There appears to be no Bayesian or generalised Bayesian treatment of \eqref{eq: KEF} in the literature, which may be due to intractability of the likelihood.
Indeed, we are not aware of a computational algorithm that would easily facilitate Bayesian inference for \eqref{eq: KEF}, so a standard Bayesian analysis will not be presented.
As the theory in this paper is finite-dimensional, we consider a finite-rank approximation of elements in $\mathcal{H}(\kappa)$ of the form $f(x) = \sum_{i=1}^p \theta_{(i)} \phi_{(i)}(x)$, with coefficients $\theta_{(i)} \in \mathbb{R}$ and basis functions $\phi_{(i)} \in \mathcal{H}(\kappa)$, where we will take $\theta$ to be $p=25$ dimensional.
Finite rank approximations have previously been considered for frequentist learning of kernel exponential families in \cite{strathmann2015gradient,sutherland2018efficient}.
In our case, the finite rank approximation ensures that any prior we induce on $f$ via a prior on the coefficients $\theta_{(i)}$ will be supported on $\mathcal{H}(\kappa)$.
If one is interested in a well-defined limit as $p \rightarrow \infty$ then one will need to ensure a.s. convergence of the sum in this limit.
If the $\phi_i$ are orthonormal in $\mathcal{H}(\kappa)$, and if the $\theta_{(i)}$ are a priori independent, then  $\mathbb{E}[\|f\|_{\mathcal{H}(\kappa)}^2] = \sum_{i=1}^p \mathbb{E}[\theta_{(i)}^2]$ so a sufficient condition, for example, is $\mathbb{E}[\theta_{(i)}^2] = O(n^{-1-\delta})$ for some $\delta > 0$.

Our interest is in the performance of KSD-Bayes applied to a multi-modal dataset, and to explore these we considered the \emph{galaxy data} of \cite{postman1986probes,roeder1990density}, comprising $n=82$ velocities in km/sec of galaxies from 6 well-separated conic sections of a survey of the \textit{Corona Borealis}. 
The data were whitened prior to computation, but results are reported with the original scale restored.
For the kernel exponential family we use $q(x) = \mathcal{N}(0, 3^2)$ and the kernel $\kappa(x,y) = \exp(- (x-y)^2 / 2)$, which ensures that \eqref{eq: KEF} is normalisable due to Proposition 2 of \cite{wenliang2019learning}. 
For basis functions we use $\phi_{(i+1)}(x) = (x^i / \sqrt{i!}) \exp( - x^2 / 2 )$, \textcolor{black}{$i = 0,\dots,24$}, which are orthonormal in $\mathcal{H}(\kappa)$ \citep{steinwart2006explicit}.
For our prior we let $\theta_{(i)} \sim \mathcal{N}(0,10^2 i^{-1.1})$, which is weakly informative within the constraint of having a well-defined $p \rightarrow \infty$ limit.
Our contamination model replaces a proportion $\epsilon$ of the dataset with values independently drawn from $\mathcal{N}(y  , 0.1^2 )$, with $y=5$, shown as black bars in the top row of \Cref{fig: KEF}.

\begin{figure}[t!]
	\centering
	\includegraphics[width = 0.9\textwidth,clip,trim = 3cm 10cm 2cm 9cm]{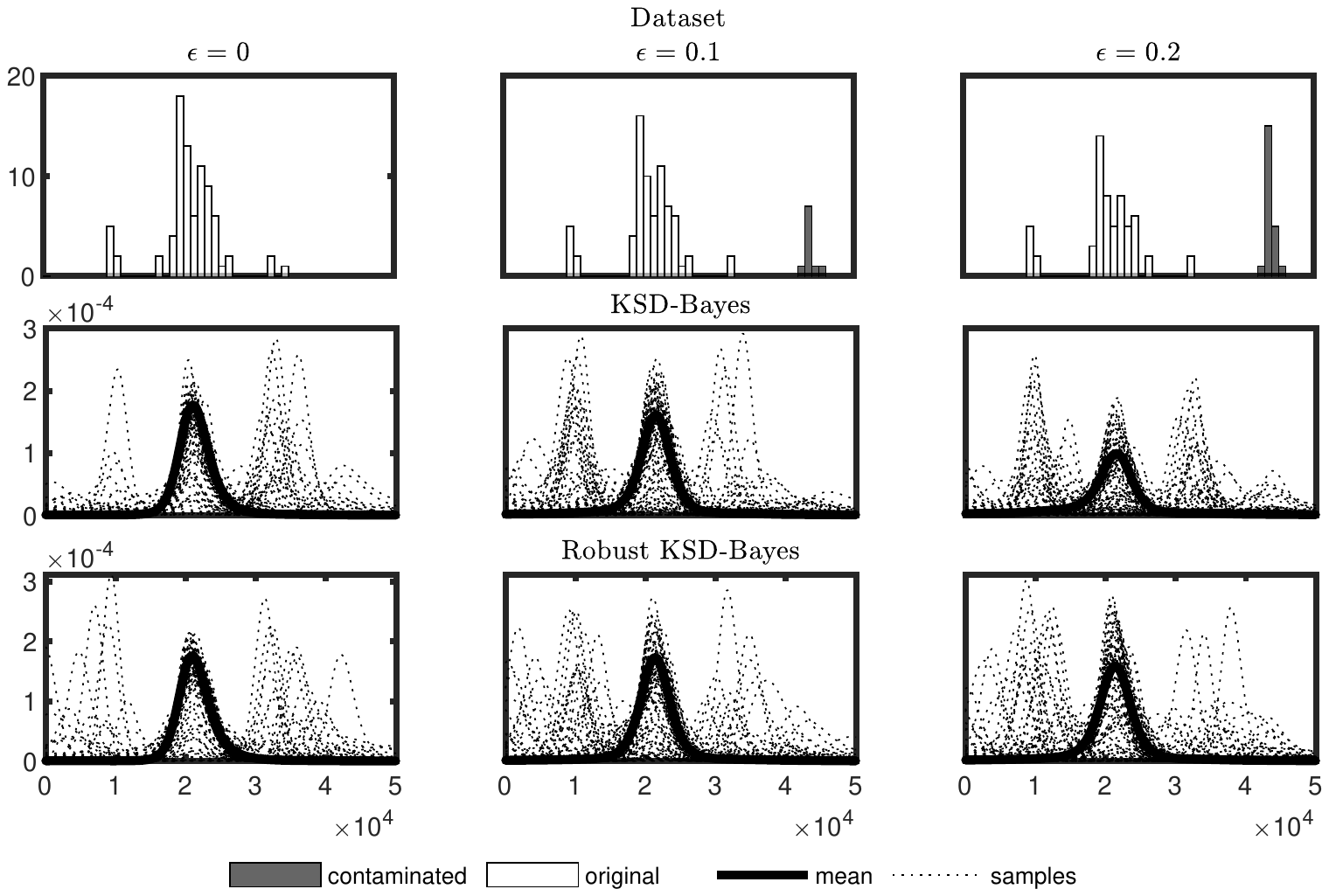}
	\caption{\textcolor{black}{Generalised posteriors for the kernel exponential family model. A proportion $\epsilon$ of the data (top row) were contaminated. }}
	\label{fig: KEF}
\end{figure}

The generalised posterior with $M(x) = 1$ is displayed in the second row of \Cref{fig: KEF}, with the bottom row presenting a robust generalised posterior based on the weighting function $M(x) = (1+x^2)^{-1/2}$, which ensures the conditions of \Cref{thm:bias-robust} are satisfied.
The results we present are for fixed $y$ and increasing $\epsilon$, since this regime is \textit{not} covered by \Cref{thm:bias-robust}.
The generalised posterior mean is a uni-modal density, which we attribute to the insensitivity of KSD to mixture proportions discussed in \Cref{subsec: pathologies}, but multi-modal densities are evident in sampled output.
Our results indicate that the robust weighting function reduces sensitivity to contamination in the dataset (note how the mass in the central mode of the generalised posterior decreases when $\epsilon = 0.2$, when the identity weighting function is used).
Whether this insensitivity of KSD to well-separated regions in the dataset is desirable or not will depend on the application, but in this case it happens to be beneficial.

\subsection{Network Inference with Exponential Graphical Models} \label{subsec: exponential model}

Our final example concerns an exponential graphical model, representing negative conditional relationships among a collection of random variables $W = (W_1,\dots,W_d)$, described in \citet[][Sec. 2.5]{yang2015graphical}.
The likelihood function is
\begin{align}
p_{W | \theta}(w | \theta) & \propto \exp\Big( - \sum_{i} \theta_{(i)} w_{(i)} - \sum_{i < j} \theta_{(i,j)} w_{(i)} w_{(j)} \Big) ,  
\label{eq: EGM defn}
\end{align}
where $w \in (0,\infty)^d$ and $\theta_{(i)} > 0, \theta_{(i,j)} \geq 0$.
The total number of parameters is $p = d(d+1)/2$.
Simulation from this model is challenging and the normalisation constant is an intractable integral, so in what follows a standard Bayesian analysis is not attempted.
Our aim is to fit \eqref{eq: EGM defn} to a protein kinase dataset, mimicking an experiment presented by \cite{Yu2016} in the score-matching context.
This dataset, originating in \cite{sachs2005causal}, consists of quantitative measurements of $d=11$ phosphorylated proteins and
phospholipids, simultaneously measured from single cells using a fluorescence-activated cell sorter, so the parameter $\theta$ is 66-dimensional.
Nine stimulatory or inhibitory interventional conditions were combined to give a total of $7,466$ cells in the dataset.
The data were square-root transformed and samples containing values greater than 10 standard deviations from their mean were judged to be \textit{bona fide} outliers and were removed. 
The remaining dataset of size $n=7,449$ was normalised to have unit standard deviation.
In most cases the measurement reflects the activation state of the kinases, and scientific interest lies in the mechanisms that underpin their interaction\footnote{There is no scientific basis to expect only negative conditional dependencies in the dataset; in this sense the model is likely to be misspecified. Our interest is in assessing the robustness properties of KSD-Bayes only, and no scientific conclusions will be drawn using this model.}.
These mechanisms are often summarised as a \emph{protein signalling network}, whose nodes are the $d$ proteins and whose edges correspond to the pairs of proteins that interact.
An important statistical challenge is to \emph{estimate} a protein signalling network from such a dataset \citep{oates2013bayesian}.
However, it is known that existing approaches to \emph{network inference} are non-robust, in a general sense, with community challenges regularly highlighting the different conclusions drawn by different estimators applied to an identical dataset \citep{hill2016inferring}.
Our interest is in whether networks estimated using KSD-Bayes are robust.

For our experiment the variables $w_{(i)}$ were re-parametrised as $x_{(i)} := \log(w_{(i)})$, in order that they are unconstrained and $\P_\theta \in \mathcal{P}_{\text{S}}(\R^d)$.
For the contamination model, a proportion $\epsilon$ of the data were replaced with the fixed value $y = (10,\dots,10) \in \mathbb{R}^d$.
Parameters were \textit{a priori} independent with $\theta_{(i)} \sim \mathcal{N}_{\text{T}}(0,1)$, $\theta_{(i,j)} \sim \mathcal{N}_{\text{T}}(0,1)$, where $\mathcal{N}_{\text{T}}$ is the Gaussian distribution truncated to the positive orthant of $\mathbb{R}^p$.
This prior is conjugate to the likelihood, as explained in \Cref{subsec: exp fam mod}, and allows the \textcolor{black}{generalised posterior} to be exactly computed.
Generalised posteriors were produced both without and with the exponential weighting function $[M(x)]_{(i,i)} = \exp(-x_{(i)})$, the latter aiming to reduce sensitivity to large values in the dataset and coinciding with the identity weighting function at $x = 0$.
From these, protein signalling networks were estimated using the $s$ most significant edges, defined as the $s$ largest values of $\bar{\theta}_{(i,j)} / \sigma_{(i,j)}$, where the generalised posterior marginal for $\theta_{(i,j)}$ is $\mathcal{N}_{\text{T}}(\bar{\theta}_{(i,j)} , \sigma_{(i,j)}^2)$.
Results are shown in \Cref{fig: EGM networks}; to optimise visualisation we report results for $s=5$, though for other values of $s$ similar conclusions hold.
It is interesting to observe little agreement between the networks returned when the identity weighting function is used, which may reflect the difficulty of the network inference task.
Reduced sensitivity to $\epsilon$ was observed when the exponential weighting function was used.
In \Cref{fig: EGM networks} we report the number of edges that are consistent with the network reported in \citet[Fig. 3A]{sachs2005causal}; the use of the exponential weighting function resulted in more edges being consistent with this benchmark network.

\begin{figure}[t!]
	\centering
	\includegraphics[width = 0.9\textwidth,clip,trim = 2cm 9.6cm 1.9cm 9.9cm]{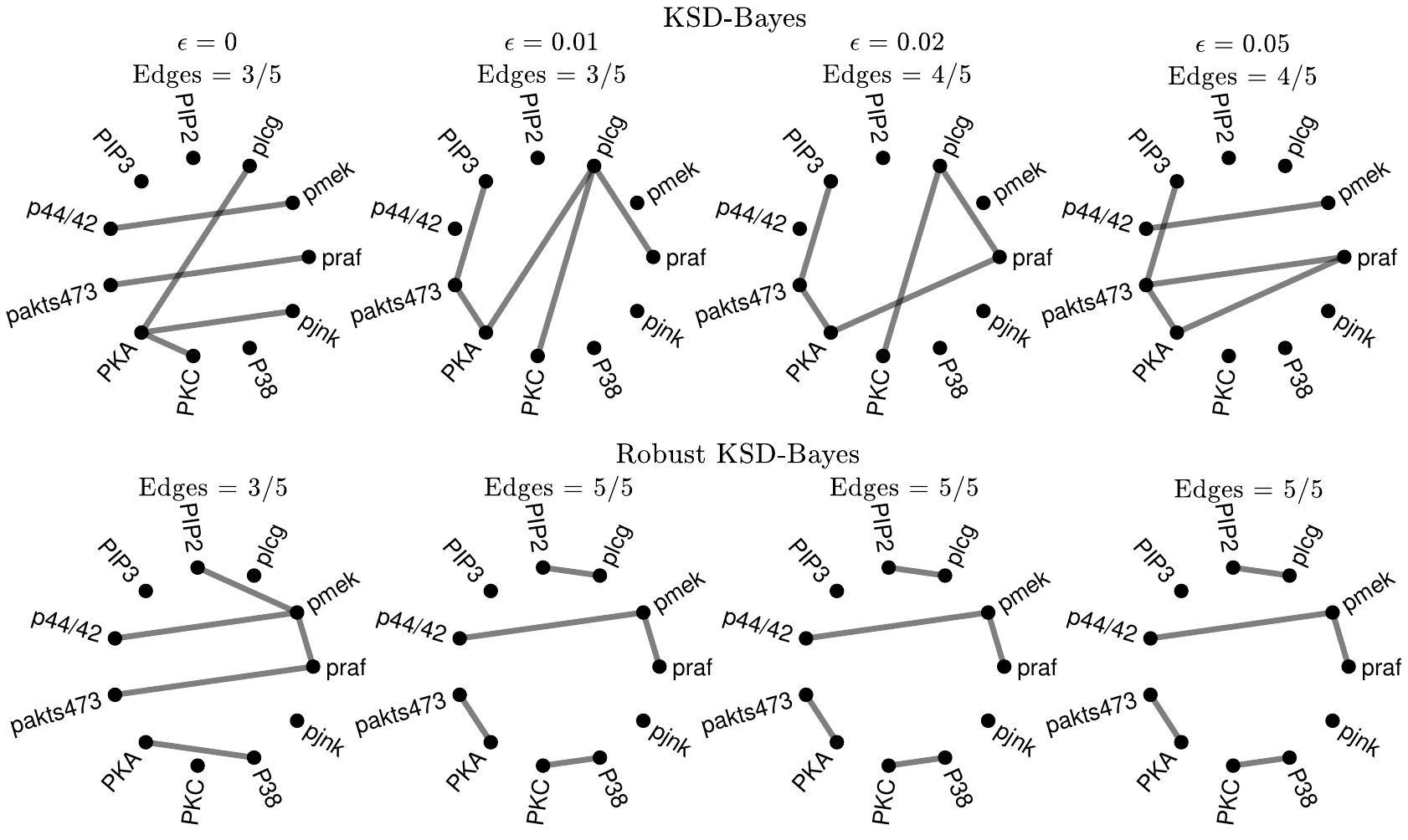}
	\caption{\textcolor{black}{Exponential graphical model; estimated protein signalling networks as a function of the proportion $\epsilon$ of contamination in the dataset.}
	}
	\label{fig: EGM networks}
\end{figure}

color{black}

\section{Conclusion} \label{sec:conclusion}

There is little existing literature concerning robust Bayesian inference in the setting of intractable likelihood.
Existing approaches to Bayesian inference for intractable likelihood fall into three categories:
(1) likelihood-free methods \citep[such as \textit{approximate Bayesian computation} and \textit{Bayesian synthetic likelihood};][]{tavare1997inferring, beaumont2002approximate, Marin2012, Price2018, Cherief-Abdellatif2019, Frazier2020b}, \textcolor{black}{(2) auxiliary variable MCMC \citep[such as the \textit{exchange algorithm} and \textit{pseudo-marginal} MCMC;][]{Moller2006, Murray2006, andrieu2009pseudo, Liang2010, Lyne2015, Doucet2015, andrieu2020metropolis}, } and (3) approximate likelihood methods \citep[such as \textit{pseudo-likelihood} and \textit{composite likelihood};][]{Besag1974, Dryden2002,eidsvik2014estimation}, which are of course also applicable beyond the Bayesian context.
Both (1) and (2) rely on either the ability to simulate from the generative model or the ability to unbiasedly estimate the data likelihood, whilst (3) represents an \textit{ad hoc} collection of approaches that are tailored to particular statistical models \citep[see the recent surveys in][]{Lyne2015, Park2018}.
These algorithms aim to approximate the standard Bayesian posterior, and do not attempt to confer robustness in situations where the model is misspecified.

This paper proposed KSD-Bayes, a generalised Bayesian procedure for likelihoods that involve an intractable normalisation constant.
KSD-Bayes provides robust generalised Bayesian inference in this context, including a theoretical guarantee of global bias-robustness over $\Theta$.
Moreover, and unlike existing Bayesian approaches to intractable likelihood, the generalised posterior can be approximated by standard sampling methods without additional levels of algorithmic complexity, even admitting conjugate analysis for the exponential family model.
From a theoretical perspective, the soundness of KSD-Bayes, in terms of consistency and asymptotic normality of the generalised posterior, was established.

Although KSD-Bayes has several appealing features, it is not a panacea for intractable likelihood.
The generalised posterior is not invariant to transformations of the dataset and, as discussed in \Cref{subsec: pathologies}, KSD can suffer from insensitivity to mixture proportions, which limits its applicability to models and datasets that are not ``too multi-modal''.
\textcolor{black}{The selection of $\beta$ remains an open problem for generalised Bayesian inference, and further regularisation may be required when the parameter $\theta$ is high-dimensional relative to the size $n$ of the dataset.}
These are challenging issues for future work.
In addition, our experiments focused on continuous data, though our theory was general.
The empirical performance of KSD-Bayes for discrete data remains to be assessed.

\paragraph{Acknowledgements:}

TM was supported EPSRC grant EP/N510129/1 at the Alan Turing Institute, UK.
JK was funded by EPSRC grant EP/L016710/1 and the Facebook Fellowship Programme.
FXB and CJO were supported by the Lloyd's Register Foundation programme on data-centric engineering at The Alan Turing Institute under the EPSRC grant EP/N510129/1.
\textcolor{black}{The authors thank the Associate Editor and three Reviewers for detailed feedback that led to an improved manuscript, and Oscar Key for pointing out an indexing error in an earlier version of the manuscript.}

{\small
\bibliography{bibliography.bib}

\begin{thebibliography}{90}
\providecommand{\natexlab}[1]{#1}
\providecommand{\url}[1]{\texttt{#1}}
\expandafter\ifx\csname urlstyle\endcsname\relax
  \providecommand{\doi}[1]{doi: #1}\else
  \providecommand{\doi}{doi: \begingroup \urlstyle{rm}\Url}\fi

\bibitem[Alquier et~al.(2016)Alquier, Ridgway, and Chopin]{Alquier2016}
P.~Alquier, J.~Ridgway, and N.~Chopin.
\newblock On the properties of variational approximations of {G}ibbs
  posteriors.
\newblock \emph{Journal of Machine Learning Research}, 17\penalty0
  (236):\penalty0 1--41, 2016.

\bibitem[Amari(1997)]{amari1997information}
S.~Amari.
\newblock Information geometry.
\newblock \emph{Contemporary Mathematics}, 203:\penalty0 81--96, 1997.

\bibitem[Andrieu and Roberts(2009)]{andrieu2009pseudo}
C.~Andrieu and G.~O. Roberts.
\newblock The pseudo-marginal approach for efficient {M}onte {C}arlo
  computations.
\newblock \emph{The Annals of Statistics}, 37\penalty0 (2):\penalty0 697--725,
  2009.

\bibitem[Andrieu et~al.(2020)Andrieu, Y{\i}ld{\i}r{\i}m, Doucet, and
  Chopin]{andrieu2020metropolis}
C.~Andrieu, S.~Y{\i}ld{\i}r{\i}m, A.~Doucet, and N.~Chopin.
\newblock {M}etropolis--{H}astings with averaged acceptance ratios.
\newblock \emph{arXiv:2101.01253}, 2020.

\bibitem[Barp et~al.(2019)Barp, Briol, Duncan, Girolami, and Mackey]{Barp2019}
A.~Barp, F.-X. Briol, A.~Duncan, M.~Girolami, and L.~Mackey.
\newblock Minimum {S}tein discrepancy estimators.
\newblock In \emph{Proceedings of the 32nd International Conference on Neural
  Information Processing Systems}, 2019.

\bibitem[Basu et~al.(2019)Basu, Shioya, and Park]{basu2019statistical}
A.~Basu, H.~Shioya, and C.~Park.
\newblock \emph{Statistical Inference: The Minimum Distance Approach}.
\newblock Chapman and Hall/CRC, 2019.

\bibitem[Baydin et~al.(2018)Baydin, Pearlmutter, Radul, and
  Siskind]{Baydin2018}
A.~G. Baydin, B.~A. Pearlmutter, A.~A. Radul, and J.~M. Siskind.
\newblock Automatic differentiation in machine learning: {A} survey.
\newblock \emph{Journal of Machine Learning Research}, 18\penalty0
  (153):\penalty0 1--43, 2018.

\bibitem[Beaumont et~al.(2002)Beaumont, Zhang, and
  Balding]{beaumont2002approximate}
M.~A. Beaumont, W.~Zhang, and D.~J. Balding.
\newblock Approximate {B}ayesian computation in population genetics.
\newblock \emph{Genetics}, 162\penalty0 (4):\penalty0 2025--2035, 2002.

\bibitem[Berger et~al.(1994)Berger, Moreno, Pericchi, Bayarri, Bernardo, Cano,
  Horra, Martín, Ríos-Insúa, Betrò, Dasgupta, Gustafson, and
  and]{Berger1994}
J.~Berger, E.~Moreno, L.~Pericchi, M.~Bayarri, J.~Bernardo, J.~Cano, J.~Horra,
  J.~Martín, D.~Ríos-Insúa, B.~Betrò, A.~Dasgupta, P.~Gustafson, and L.~W.
  and.
\newblock {An overview of robust Bayesian analysis}.
\newblock \emph{TEST: An Official Journal of the Spanish Society of Statistics
  and Operations Research}, 3\penalty0 (1):\penalty0 5--124, 1994.

\bibitem[Bernardo and Smith(2009)]{Bernardo}
J.~M. Bernardo and A.~F. Smith.
\newblock \emph{Bayesian {T}heory}.
\newblock John Wiley \& Sons, 2009.

\bibitem[Besag(1974)]{Besag1974}
J.~Besag.
\newblock Spatial interaction and the statistical analysis of lattice systems.
\newblock \emph{Journal of the Royal Statistical Society. Series B
  (Methodological)}, 36\penalty0 (2):\penalty0 192--236, 1974.

\bibitem[Besag(1986)]{Besag1986}
J.~Besag.
\newblock On the statistical analysis of dirty pictures.
\newblock \emph{Journal of the Royal Statistical Society. Series B
  (Methodological)}, 48\penalty0 (3):\penalty0 259--302, 1986.

\bibitem[Bissiri et~al.(2016)Bissiri, Holmes, and Walker]{Bissiri2016}
P.~G. Bissiri, C.~C. Holmes, and S.~G. Walker.
\newblock A general framework for updating belief distributions.
\newblock \emph{Journal of the Royal Statistical Society. Series B
  (Methodological)}, 78\penalty0 (5):\penalty0 1103, 2016.

\bibitem[Canu and Smola(2006)]{canu2006kernel}
S.~Canu and A.~Smola.
\newblock Kernel methods and the exponential family.
\newblock \emph{Neurocomputing}, 69\penalty0 (7-9):\penalty0 714--720, 2006.

\bibitem[Caponnetto et~al.(2008)Caponnetto, Micchelli, Pontil, and
  Ying]{Caponnetto2008}
A.~Caponnetto, C.~A. Micchelli, M.~Pontil, and Y.~Ying.
\newblock Universal multi-task kernels.
\newblock \emph{Journal of Machine Learning Research}, 9:\penalty0 1615–1646,
  2008.

\bibitem[Carmeli et~al.(2006)Carmeli, De~Vito, and Toigo]{Carmeli2006}
C.~Carmeli, E.~De~Vito, and A.~Toigo.
\newblock Vector valued reproducing kernel {H}ilbert spaces of integrable
  functions and {M}ercer theorem.
\newblock \emph{Analysis and Applications}, 10:\penalty0 377--408, 2006.

\bibitem[Carmeli et~al.(2010)Carmeli, De~Vito, Toigo, and
  Umanità]{Carmeli2011}
C.~Carmeli, E.~De~Vito, A.~Toigo, and V.~Umanità.
\newblock Vector valued reproducing kernel {H}ilbert spaces and universality.
\newblock \emph{Analysis and Applications}, 8\penalty0 (1):\penalty0 19--61,
  2010.

\bibitem[Chen et~al.(2019)Chen, Barp, Briol, Gorham, Girolami, Mackey, and
  Oates]{chen2019stein}
W.~Y. Chen, A.~Barp, F.-X. Briol, J.~Gorham, M.~Girolami, L.~Mackey, and C.~J.
  Oates.
\newblock Stein point {M}arkov chain {M}onte {C}arlo.
\newblock In \emph{Proceedings of the 36th International Conference on Machine
  Learning}, pages 1011--1021, 2019.

\bibitem[Cherief-Abdellatif and Alquier(2020)]{Cherief-Abdellatif2019}
B.-E. Cherief-Abdellatif and P.~Alquier.
\newblock {MMD-Bayes: Robust Bayesian estimation via maximum mean discrepancy}.
\newblock In \emph{Proceedings of the 2nd Symposium on Advances in Approximate
  Bayesian Inference}, pages 1--21, 2020.

\bibitem[Chernozhukov and Hong(2003)]{chernozhukov2003mcmc}
V.~Chernozhukov and H.~Hong.
\newblock An {MCMC} approach to classical estimation.
\newblock \emph{Journal of Econometrics}, 115\penalty0 (2):\penalty0 293--346,
  2003.

\bibitem[Chwialkowski et~al.(2016)Chwialkowski, Strathmann, and
  Gretton]{Chwialkowski2016a}
K.~Chwialkowski, H.~Strathmann, and A.~Gretton.
\newblock A kernel test of goodness of fit.
\newblock In \emph{Proceedings of the 33rd International Conference on Machine
  Learning}, pages 2606--2615, 2016.

\bibitem[Davidson(1994)]{Davidson1994}
J.~Davidson.
\newblock \emph{Stochastic Limit Theory: An Introduction for Econometricians}.
\newblock Oxford University Press, 1994.

\bibitem[Diggle(1990)]{Diggle1990}
P.~J. Diggle.
\newblock A point process modelling approach to raised incidence of a rare
  phenomenon in the vicinity of a prespecified point.
\newblock \emph{Journal of the Royal Statistical Society. Series A (Statistics
  in Society)}, 153\penalty0 (3):\penalty0 349--362, 1990.

\bibitem[Doucet et~al.(2015)Doucet, Pitt, Deligiannidis, and Kohn]{Doucet2015}
A.~Doucet, M.~K. Pitt, G.~Deligiannidis, and R.~Kohn.
\newblock {Efficient implementation of Markov chain Monte Carlo when using an
  unbiased likelihood estimator}.
\newblock \emph{Biometrika}, 102\penalty0 (2):\penalty0 295--313, 2015.

\bibitem[Dryden et~al.(2002)Dryden, Ippoliti, and Romagnoli]{Dryden2002}
I.~Dryden, L.~Ippoliti, and L.~Romagnoli.
\newblock Adjusted maximum likelihood and pseudo-likelihood estimation for
  noisy {G}aussian {M}arkov random fields.
\newblock \emph{Journal of Computational and Graphical Statistics}, 11\penalty0
  (2):\penalty0 370--388, 2002.

\bibitem[Durrett(2010)]{Durrett2010b}
R.~Durrett.
\newblock \emph{Probability: Theory and Examples (4th Edition)}.
\newblock Cambridge University Press, 2010.

\bibitem[Eidsvik et~al.(2014)Eidsvik, Shaby, Reich, Wheeler, and
  Niemi]{eidsvik2014estimation}
J.~Eidsvik, B.~A. Shaby, B.~J. Reich, M.~Wheeler, and J.~Niemi.
\newblock Estimation and prediction in spatial models with block composite
  likelihoods.
\newblock \emph{Journal of Computational and Graphical Statistics}, 23\penalty0
  (2):\penalty0 295--315, 2014.

\bibitem[Frazier(2020)]{Frazier2020b}
D.~T. Frazier.
\newblock Robust and efficient approximate {B}ayesian computation: A minimum
  distance approach.
\newblock \emph{arXiv:2006.14126}, 2020.

\bibitem[Friel and Pettitt(2008)]{friel2008marginal}
N.~Friel and A.~N. Pettitt.
\newblock Marginal likelihood estimation via power posteriors.
\newblock \emph{Journal of the Royal Statistical Society: Series B (Statistical
  Methodology)}, 70\penalty0 (3):\penalty0 589--607, 2008.

\bibitem[Ghosh and Basu(2016)]{Ghosh2016}
A.~Ghosh and A.~Basu.
\newblock {Robust Bayes estimation using the density power divergence}.
\newblock \emph{Annals of the Institute of Statistical Mathematics},
  68:\penalty0 413--437, 2016.

\bibitem[Giummol{\`e} et~al.(2019)Giummol{\`e}, Mameli, Ruli, and
  Ventura]{Giummole2019}
F.~Giummol{\`e}, V.~Mameli, E.~Ruli, and L.~Ventura.
\newblock Objective {B}ayesian inference with proper scoring rules.
\newblock \emph{Test}, 28\penalty0 (3):\penalty0 728--755, 2019.

\bibitem[Gorham and Mackey(2015)]{Gorham2015}
J.~Gorham and L.~Mackey.
\newblock Measuring sample quality with {S}tein's method.
\newblock In \emph{Proceedings of the 28th International Conference on Neural
  Information Processing Systems}, 2015.

\bibitem[Gorham and Mackey(2017)]{Gorham2017}
J.~Gorham and L.~Mackey.
\newblock Measuring sample quality with kernels.
\newblock In \emph{Proceedings of the 34th International Conference on Machine
  Learning}, pages 1292--1301, 2017.

\bibitem[Gorham et~al.(2019)Gorham, Duncan, Vollmer, and Mackey]{Gorham2019}
J.~Gorham, A.~B. Duncan, S.~J. Vollmer, and L.~Mackey.
\newblock Measuring sample quality with diffusions.
\newblock \emph{The Annals of Applied Probability}, 29\penalty0 (5):\penalty0
  2884--2928, 2019.

\bibitem[Gorham et~al.(2020)Gorham, Raj, and Mackey]{gorham2020stochastic}
J.~Gorham, A.~Raj, and L.~Mackey.
\newblock Stochastic {S}tein discrepancies.
\newblock In \emph{Proceedings of the 34th International Conference on Neural
  Information Processing Systems}, 2020.

\bibitem[Gr{\"u}nwald(2011)]{SafeLearning}
P.~Gr{\"u}nwald.
\newblock Safe learning: {B}ridging the gap between {B}ayes, {MDL} and
  statistical learning theory via empirical convexity.
\newblock In \emph{Proceedings of the 24th Annual Conference on Learning
  Theory}, pages 397--420, 2011.

\bibitem[Gr{\"u}nwald(2012)]{SafeBayesian}
P.~Gr{\"u}nwald.
\newblock The safe {B}ayesian.
\newblock In \emph{Proceedings of the 23rd International Conference on
  Algorithmic Learning Theory}, pages 169--183, 2012.

\bibitem[Hill et~al.(2016)Hill, Heiser, Cokelaer, Unger, Nesser, Carlin, Zhang,
  Sokolov, Paull, Wong, et~al.]{hill2016inferring}
S.~M. Hill, L.~M. Heiser, T.~Cokelaer, M.~Unger, N.~K. Nesser, D.~E. Carlin,
  Y.~Zhang, A.~Sokolov, E.~O. Paull, C.~K. Wong, et~al.
\newblock Inferring causal molecular networks: {E}mpirical assessment through a
  community-based effort.
\newblock \emph{Nature Methods}, 13\penalty0 (4):\penalty0 310--318, 2016.

\bibitem[Hoeffding(1961)]{Hoeffding1961}
W.~Hoeffding.
\newblock The strong law of large numbers for {U}-statistics.
\newblock \emph{Institute of Statistics Mimeo Series}, 302, 1961.

\bibitem[Holmes and Walker(2017)]{holmes2017assigning}
C.~Holmes and S.~Walker.
\newblock Assigning a value to a power likelihood in a general {B}ayesian
  model.
\newblock \emph{Biometrika}, 104\penalty0 (2):\penalty0 497--503, 2017.

\bibitem[Hooker and Vidyashankar(2014)]{hooker2014bayesian}
G.~Hooker and A.~N. Vidyashankar.
\newblock Bayesian model robustness via disparities.
\newblock \emph{Test}, 23\penalty0 (3):\penalty0 556--584, 2014.

\bibitem[Huber and Ronchetti(2009)]{Huber2009}
P.~J. Huber and E.~M. Ronchetti.
\newblock \emph{{Robust Statistics}}.
\newblock Wiley, 2009.

\bibitem[Huggins and Mackey(2018)]{huggins2018random}
J.~H. Huggins and L.~Mackey.
\newblock Random feature {S}tein discrepancies.
\newblock In \emph{Proceedings of the 32nd International Conference on Neural
  Information Processing Systems}, pages 1903--1913, 2018.

\bibitem[Hyv{{\"a}}rinen(2005)]{Hyvarinen2005}
A.~Hyv{{\"a}}rinen.
\newblock Estimation of non-normalized statistical models by score matching.
\newblock \emph{Journal of Machine Learning Research}, 6\penalty0
  (24):\penalty0 695--709, 2005.

\bibitem[Jewson et~al.(2018)Jewson, Smith, and Holmes]{Jewson2018}
J.~Jewson, J.~Q. Smith, and C.~Holmes.
\newblock {Principled Bayesian minimum divergence inference}.
\newblock \emph{Entropy}, 20\penalty0 (6):\penalty0 442, 2018.

\bibitem[Jiang et~al.(2021)Jiang, Li, and Xiao]{jiang2021bayesian}
X.~Jiang, Q.~Li, and G.~Xiao.
\newblock {B}ayesian modeling of spatial transcriptomics data via a modified
  {I}sing model.
\newblock \emph{arXiv:2104.13957}, 2021.

\bibitem[Kim et~al.(2021)Kim, Bhattacharya, and Maiti]{kim2021variational}
M.~Kim, S.~Bhattacharya, and T.~Maiti.
\newblock Variational {B}ayes algorithm and posterior consistency of {I}sing
  model parameter estimation.
\newblock \emph{arXiv:2109.01548}, 2021.

\bibitem[Kleijn and van~der Vaart(2012)]{kleijn2012bernstein}
B.~J. Kleijn and A.~W. van~der Vaart.
\newblock The {B}ernstein-von-{M}ises theorem under misspecification.
\newblock \emph{Electronic Journal of Statistics}, 6:\penalty0 354--381, 2012.

\bibitem[Knoblauch et~al.(2019)Knoblauch, Jewson, and Damoulas]{Knoblauch2019}
J.~Knoblauch, J.~Jewson, and T.~Damoulas.
\newblock Generalized variational inference: {T}hree arguments for deriving new
  posteriors.
\newblock \emph{arXiv:1904.02063}, 2019.

\bibitem[Liang(2010)]{Liang2010}
F.~Liang.
\newblock {A double Metropolis-Hastings sampler for spatial models with
  intractable normalizing constants}.
\newblock \emph{Journal of Statistical Computation and Simulation}, 80\penalty0
  (9):\penalty0 1007--1022, 2010.

\bibitem[Lindsay(1994)]{lindsay1994efficiency}
B.~G. Lindsay.
\newblock Efficiency versus robustness: The case for minimum {H}ellinger
  distance and related methods.
\newblock \emph{The Annals of Statistics}, 22\penalty0 (2):\penalty0
  1081--1114, 1994.

\bibitem[Liu et~al.(2016)Liu, Lee, and Jordan]{Liu2016b}
Q.~Liu, J.~Lee, and M.~Jordan.
\newblock A kernelized {S}tein discrepancy for goodness-of-fit tests.
\newblock In \emph{Proceedings of the 33rd International Conference on Machine
  Learning}, pages 276--284, 2016.

\bibitem[Liu et~al.(2019)Liu, Kanamori, Jitkrittum, and Chen]{liu2019fisher}
S.~Liu, T.~Kanamori, W.~Jitkrittum, and Y.~Chen.
\newblock Fisher efficient inference of intractable models.
\newblock In \emph{Proceedings of the 32nd International Conference on Neural
  Information Processing Systems}, 2019.

\bibitem[Lyddon et~al.(2019)Lyddon, Holmes, and Walker]{Lyddon2018}
S.~P. Lyddon, C.~C. Holmes, and S.~G. Walker.
\newblock {General {B}ayesian updating and the loss-likelihood bootstrap}.
\newblock \emph{Biometrika}, 106\penalty0 (2):\penalty0 465--478, 2019.

\bibitem[Lyne et~al.(2015)Lyne, Girolami, Atchadé, Strathmann, and
  Simpson]{Lyne2015}
A.-M. Lyne, M.~Girolami, Y.~Atchadé, H.~Strathmann, and D.~Simpson.
\newblock On {R}ussian roulette estimates for {B}ayesian inference with
  doubly-intractable likelihoods.
\newblock \emph{Statistical Science}, 30\penalty0 (4):\penalty0 443--467, 2015.

\bibitem[Ma et~al.(2015)Ma, Chen, and Fox]{Ma2015}
Y.-A. Ma, T.~Chen, and E.~Fox.
\newblock A complete recipe for stochastic gradient {MCMC}.
\newblock In \emph{Proceedings of the 28th International Conference on Neural
  Information Processing Systems}, pages 2917--2925, 2015.

\bibitem[Marin et~al.(2012)Marin, Pudlo, Robert, and Ryder]{Marin2012}
J.-M. Marin, P.~Pudlo, C.~P. Robert, and R.~Ryder.
\newblock {Approximate Bayesian computational methods}.
\newblock \emph{{Statistics and Computing}}, 22\penalty0 (6), 2012.

\bibitem[Miller(2021)]{Miller2019}
J.~W. Miller.
\newblock Asymptotic normality, concentration, and coverage of generalized
  posteriors.
\newblock \emph{Journal of Machine Learning Research}, 22\penalty0
  (168):\penalty0 1--53, 2021.

\bibitem[Miller and Dunson(2019)]{DunsonCoarsening}
J.~W. Miller and D.~B. Dunson.
\newblock Robust {B}ayesian inference via coarsening.
\newblock \emph{Journal of the American Statistical Association}, 114\penalty0
  (527):\penalty0 1113--1125, 2019.

\bibitem[Moores et~al.(2020)Moores, Nicholls, Pettitt, and
  Mengersen]{moores2020scalable}
M.~Moores, G.~Nicholls, A.~Pettitt, and K.~Mengersen.
\newblock Scalable {B}ayesian inference for the inverse temperature of a hidden
  {P}otts model.
\newblock \emph{Bayesian Analysis}, 15\penalty0 (1):\penalty0 1--27, 2020.

\bibitem[Murray et~al.(2006)Murray, Ghahramani, and MacKay]{Murray2006}
I.~Murray, Z.~Ghahramani, and D.~J.~C. MacKay.
\newblock {MCMC} for doubly-intractable distributions.
\newblock In \emph{Proceedings of the 22nd Annual Conference on Uncertainty in
  Artificial Intelligence}, pages 359--366, 2006.

\bibitem[Møller et~al.(2006)Møller, Pettitt, Reeves, and
  Berthelsen]{Moller2006}
J.~Møller, A.~N. Pettitt, R.~Reeves, and K.~K. Berthelsen.
\newblock An efficient {M}arkov chain {M}onte {C}arlo method for distributions
  with intractable normalising constants.
\newblock \emph{Biometrika}, 93\penalty0 (2):\penalty0 451--458, 2006.

\bibitem[Nakagawa and Hashimoto(2020)]{Nakagawa2020}
T.~Nakagawa and S.~Hashimoto.
\newblock Robust {B}ayesian inference via $\gamma$-divergence.
\newblock \emph{Communications in Statistics - Theory and Methods}, 49\penalty0
  (2):\penalty0 343--360, 2020.

\bibitem[Newey and McFadden(1994)]{Newey1994}
W.~K. Newey and D.~McFadden.
\newblock Large sample estimation and hypothesis testing.
\newblock In \emph{Handbook of Econometrics}, volume~4, pages 2111--2245, 1994.

\bibitem[Oates(2013)]{oates2013bayesian}
C.~J. Oates.
\newblock \emph{Bayesian inference for protein signalling networks}.
\newblock PhD thesis, University of Warwick, 2013.

\bibitem[Ollila and Raninen(2019)]{ollila2019optimal}
E.~Ollila and E.~Raninen.
\newblock Optimal shrinkage covariance matrix estimation under random sampling
  from elliptical distributions.
\newblock \emph{IEEE Transactions on Signal Processing}, 67\penalty0
  (10):\penalty0 2707--2719, 2019.

\bibitem[Park and Haran(2018)]{Park2018}
J.~Park and M.~Haran.
\newblock Bayesian inference in the presence of intractable normalizing
  functions.
\newblock \emph{Journal of the American Statistical Association}, 113\penalty0
  (523):\penalty0 1372--1390, 2018.

\bibitem[Paulsen and Raghupathi(2016)]{Paulsen2016}
V.~I. Paulsen and M.~Raghupathi.
\newblock \emph{An Introduction to the Theory of Reproducing Kernel Hilbert
  Spaces}.
\newblock Cambridge University Press, 2016.

\bibitem[Postman et~al.(1986)Postman, Huchra, and Geller]{postman1986probes}
M.~Postman, J.~Huchra, and M.~Geller.
\newblock Probes of large-scale structure in the corona borealis region.
\newblock \emph{The Astronomical Journal}, 92:\penalty0 1238--1247, 1986.

\bibitem[Price et~al.(2018)Price, Drovandi, Lee, and Nott]{Price2018}
L.~F. Price, C.~C. Drovandi, A.~Lee, and D.~J. Nott.
\newblock Bayesian synthetic likelihood.
\newblock \emph{Journal of Computational and Graphical Statistics}, 27\penalty0
  (1):\penalty0 1--11, 2018.

\bibitem[Riabiz et~al.(2021)Riabiz, Chen, Cockayne, Swietach, Niederer, Mackey,
  and Oates]{riabiz2020optimal}
M.~Riabiz, W.~Chen, J.~Cockayne, P.~Swietach, S.~A. Niederer, L.~Mackey, and
  C.~J. Oates.
\newblock Optimal thinning of {MCMC} output.
\newblock \emph{Journal of the Royal Statistical Society: Series B (Statistical
  Methodology)}, 2021.
\newblock To appear.

\bibitem[Roeder(1990)]{roeder1990density}
K.~Roeder.
\newblock Density estimation with confidence sets exemplified by superclusters
  and voids in the galaxies.
\newblock \emph{Journal of the American Statistical Association}, 85\penalty0
  (411):\penalty0 617--624, 1990.

\bibitem[Sachs et~al.(2005)Sachs, Perez, Pe'er, Lauffenburger, and
  Nolan]{sachs2005causal}
K.~Sachs, O.~Perez, D.~Pe'er, D.~A. Lauffenburger, and G.~P. Nolan.
\newblock Causal protein-signaling networks derived from multiparameter
  single-cell data.
\newblock \emph{Science}, 308\penalty0 (5721):\penalty0 523--529, 2005.

\bibitem[Stein(1972)]{stein1972bound}
C.~Stein.
\newblock A bound for the error in the normal approximation to the distribution
  of a sum of dependent random variables.
\newblock \emph{Proceedings of the 6th Berkeley Symposium on Mathematical
  Statistics and Probability, Volume 2: Probability Theory}, 1972.

\bibitem[Steinwart and Christmann(2008)]{Steinwart2008}
I.~Steinwart and A.~Christmann.
\newblock \emph{Support Vector Machines}.
\newblock Springer, 2008.

\bibitem[Steinwart et~al.(2006)Steinwart, Hush, and
  Scovel]{steinwart2006explicit}
I.~Steinwart, D.~Hush, and C.~Scovel.
\newblock An explicit description of the reproducing kernel {H}ilbert spaces of
  {G}aussian {RBF} kernels.
\newblock \emph{IEEE Transactions on Information Theory}, 52\penalty0
  (10):\penalty0 4635--4643, 2006.

\bibitem[Strathmann et~al.(2015)Strathmann, Sejdinovic, Livingstone, Szabo, and
  Gretton]{strathmann2015gradient}
H.~Strathmann, D.~Sejdinovic, S.~Livingstone, Z.~Szabo, and A.~Gretton.
\newblock Gradient-free {H}amiltonian {M}onte {C}arlo with efficient kernel
  exponential families.
\newblock In \emph{Proceedings of the 28th International Conference on Neural
  Information Processing Systems}, 2015.

\bibitem[Sutherland et~al.(2018)Sutherland, Strathmann, Arbel, and
  Gretton]{sutherland2018efficient}
D.~J. Sutherland, H.~Strathmann, M.~Arbel, and A.~Gretton.
\newblock Efficient and principled score estimation with {N}ystr\"{o}m kernel
  exponential families.
\newblock In \emph{Proceedings of the 21st International Conference on
  Artificial Intelligence and Statistics}, pages 652--660, 2018.

\bibitem[Tavar{\'e} et~al.(1997)Tavar{\'e}, Balding, Griffiths, and
  Donnelly]{tavare1997inferring}
S.~Tavar{\'e}, D.~J. Balding, R.~C. Griffiths, and P.~Donnelly.
\newblock Inferring coalescence times from {DNA} sequence data.
\newblock \emph{Genetics}, 145\penalty0 (2):\penalty0 505--518, 1997.

\bibitem[van~der Vaart(1998)]{Vaart1998}
A.~van~der Vaart.
\newblock \emph{Asymptotic Statistics}.
\newblock Cambridge University Press, 1998.

\bibitem[Wainwright(2019)]{Wainwright2019}
M.~J. Wainwright.
\newblock \emph{{High-Dimensional Statistics: A Non-Asymptotic Viewpoint}}.
\newblock Cambridge University Press, 2019.

\bibitem[Wang et~al.(2019)Wang, Tang, Bajaj, and Liu]{wang2019stein}
D.~Wang, Z.~Tang, C.~Bajaj, and Q.~Liu.
\newblock {S}tein variational gradient descent with matrix-valued kernels.
\newblock In \emph{Proceedings of the 32nd International Conference on Neural
  Information Processing Systems}, 2019.

\bibitem[Wenliang et~al.(2019)Wenliang, Sutherland, Strathmann, and
  Gretton]{wenliang2019learning}
L.~Wenliang, D.~J. Sutherland, H.~Strathmann, and A.~Gretton.
\newblock Learning deep kernels for exponential family densities.
\newblock In \emph{Proceedings of the 36th International Conference on Machine
  Learning}, pages 6737--6746, 2019.

\bibitem[Wenliang(2020)]{wenliang2020blindness}
L.~K. Wenliang.
\newblock Blindness of score-based methods to isolated components and mixing
  proportions.
\newblock \emph{arXiv:2008.10087}, 2020.

\bibitem[Williams(1980)]{Williams}
P.~M. Williams.
\newblock Bayesian conditionalisation and the principle of minimum information.
\newblock \emph{The British Journal for the Philosophy of Science}, 31\penalty0
  (2):\penalty0 131--144, 1980.

\bibitem[Wu and Martin(2020)]{Wu2020}
P.-S. Wu and R.~Martin.
\newblock A comparison of learning rate selection methods in generalized
  {B}ayesian inference.
\newblock \emph{arXiv:2012.11349}, 2020.

\bibitem[Yang et~al.(2015)Yang, Ravikumar, Allen, and Liu]{yang2015graphical}
E.~Yang, P.~Ravikumar, G.~I. Allen, and Z.~Liu.
\newblock Graphical models via univariate exponential family distributions.
\newblock \emph{Journal of Machine Learning Research}, 16\penalty0
  (115):\penalty0 3813--3847, 2015.

\bibitem[Yang et~al.(2018)Yang, Liu, Rao, and Neville]{Yang2018}
J.~Yang, Q.~Liu, V.~Rao, and J.~Neville.
\newblock Goodness-of-fit testing for discrete distributions via {S}tein
  discrepancy.
\newblock In \emph{Proceedings of the 35th International Conference on Machine
  Learning}, pages 5561--5570, 2018.

\bibitem[Yu et~al.(2016)Yu, Kolar, and Gupta]{Yu2016}
M.~Yu, M.~Kolar, and V.~Gupta.
\newblock Statistical inference for pairwise graphical models using score
  matching.
\newblock In \emph{Proceedings of the 29th International Conference on Neural
  Information Processing Systems}, 2016.

\bibitem[Zellner(1988)]{Zellner1988}
A.~Zellner.
\newblock {Optimal information processing and Bayes's theorem}.
\newblock \emph{The American Statistician}, 42\penalty0 (4):\penalty0 278--280,
  1988.

\end{thebibliography}
}

\appendix

\begin{center}
   \LARGE \textbf{Supplementary Material}
\end{center}

\vspace{3mm}

This electronic supplement contains proofs for all theoretical results in the main text, \textcolor{black}{as well as the additional empirical results referred to in the main text}.
First, in \Cref{sec:kernel_RKHS} a formal definition of a vector-valued RKHS is provided.
Proofs for the results in the main text are contained in \Cref{sec:appendix}, with the statements and proofs of auxiliary technical lemmas contained in \Cref{sec: proof_pre_result_4}.
\textcolor{black}{Additional empirical results are contained in \Cref{sec: emp appendix}.}


\section{Background on Vector-Valued RKHS} \label{sec:kernel_RKHS}

This appendix contains background on the matrix-valued kernels used in the main text. 
Our main references are \citet{Carmeli2006, Caponnetto2008,Carmeli2011}.
For simplicity we start with the scalar-valued case and define a scalar-valued kernel:

\begin{definition}[Scalar-valued kernel]
	A function $k: \X \times \X \to \R$ is called a \emph{(scalar-valued) kernel} if 
	\begin{enumerate}[label=(\roman*)]
	\item $k$ is \emph{symmetric}; i.e. $k(x, x') = k(x', x)$ for all $x, x' \in \X$,
	\item $k$ is \emph{positive semi-definite}; i.e. $
	\sum_{i=1}^{n} \sum_{j=1}^{n} c_i c_j k(x_i, x_j) \ge 0$ for all $n \in \N$, $c_1, \dots, c_n \in \R$ and all $x_1, \dots, x_n \in \X$.
	\end{enumerate}
\end{definition}
To every scalar-valued kernel is an associated Hilbert space $\H$ of functions $h: \X \to \R$, called the \emph{reproducing kernel Hilbert space} (RKHS) of the kernel.
\begin{definition}[Reproducing kernel Hilbert space] \label{def: rkhs app}
    A Hilbert space $\H$ is said to be \emph{reproduced} by a kernel $k : \X \times \X \rightarrow \R$ if
    \begin{enumerate}[label=(\roman*)]
        \item $k(x, \cdot) \in \H$ for all $x \in \X$,
        \item $\langle h, k(x, \cdot) \rangle_{\H} = h(x)$ for all $x \in \X$ and $h \in \H$. \label{item: repro prop}
    \end{enumerate}
    \Cref{item: repro prop} is called the \emph{reproducing property} of $k$ in $\H$.
\end{definition}
It can be shown that, for every kernel $k$, there exists a unique Hilbert space $\H$ reproduced by $k$ \citep[Theorem~2.14]{Paulsen2016}.
These definitions can be generalised in the form of a matrix-valued kernel $K: \X \times \X \to \R^{m \times m}$.

\begin{definition}[Matrix-valued kernel]
	A function $K: \X \times \X \to \R^{m \times m}$, $m > 1$, is called a \emph{(matrix-valued) kernel} if 
	\begin{enumerate}[label=(\roman*)]
	\item $K$ is \emph{symmetric}; i.e. $K(x, x') = K(x', x)$ for all $x, x' \in \X$, \label{item: symm prop}
	\item $k$ is \emph{positive semi-definite}; i.e. $
	\sum_{i=1}^{n} \sum_{j=1}^{n} c_i \cdot k(x_i, x_j) c_j \ge 0$ for all $n \in \N$, $c_1, \dots, c_n \in \R^m$ and all $x_1, \dots, x_n \in \X$.
	\end{enumerate}
\end{definition}

As a direct generalisation of the scalar-valued case, there exists a uniquely associated Hilbert space $\H$ of functions $h: \X \to \R^m$ to every matrix-valued kernel $K : \X \times \X \rightarrow \R^{m \times m}$.
To define this Hilbert space, whose inner product we denote $\langle \cdot , \cdot \rangle_\H$, some additional notation is required:
Let $F$ be a $\R^{m \times m}$-valued function and let $F_{i,-}$ denote the vector-valued function $F_{i,-} : \X \rightarrow \R^m$ defined by the the $i$-th row of $F$.
Similarly, let $G$ be a $\R^{m \times m}$-valued function and let $G_{-,j}$ denote the vector-valued function $G_{-,j} : \X \rightarrow \R^m$ defined by the $j$-th column of $G$.
Formally define the symbols $\langle F, g \rangle_{\H}$, $\langle f, G \rangle_{\H}$ and $\langle F, G \rangle_{\H}$ as follows
\begin{align*}
\langle F, g \rangle_{\H} & := \left[ \begin{array}{c}
\langle F_{1,-}, g \rangle_{\H} \\
\vdots \\
\langle F_{m,-}, g \rangle_{\H}
\end{array} \right] \in \R^m, \quad
\langle f, G \rangle_{\H} := \left[ \begin{array}{c}
\langle f, G_{-,1} \rangle_{\H} \\
\vdots \\
\langle f, G_{-,m} \rangle_{\H}
\end{array} \right] \in \R^m, \\
\langle F, G \rangle_{\H} & := \left[ \begin{array}{ccc}
\langle F_{1,-}, G_{-,1} \rangle_{\H} & \cdots & \langle F_{1,-}, G_{-,m} \rangle_{\H} \\
\vdots & & \vdots \\
\langle F_{m,-}, G_{-,1} \rangle_{\H} & \cdots & \langle F_{m,-}, G_{-,m} \rangle_{\H}
\end{array} \right] \in \R^{m \times m},
\end{align*}
where these are to be interpreted as compound symbols only (i.e. we are not attempting to define an inner product on matrix-valued functions).
Then, the generalisation of the reproducing property (\cref{item: repro prop} in \Cref{def: rkhs app}) to a matrix-valued kernel $K$ is 
\begin{align*}
h(x) = \langle h, K(x, \cdot) \rangle_{\H} = \left[ \begin{array}{c}
\langle h, K_{-,1}(x, \cdot) \rangle_{\H} \\
\vdots \\
\langle h, K_{-,m}(x, \cdot) \rangle_{\H}
\end{array} \right]
\end{align*}
for all $x \in \X$ and $h \in \H$ \citep{Carmeli2011}.
The generalisation of the symmetry property (\cref{item: symm prop} in \Cref{def: rkhs app}) is straight-forward; $K(x,x') = K(x',x)$ for all $x,x' \in \X$.
A Hilbert space $\H$ for which these two properties are satisfied is called a \emph{vector-valued RKHS} that we say is \emph{reproduced} by the matrix-valued kernel $K$.
\textcolor{black}{Matrix-valued kernels and their associated vector-valued RKHS have recently been exploited in the context of Stein's method  \citep[e.g.][]{Barp2019,wang2019stein}.}


\section{Proofs of Theoretical Results} \label{sec:appendix}

This appendix provides proofs for all theoretical results in the main text. 
On occasion we refer to auxiliary theoretical results, which are stated and proven in \Cref{sec: proof_pre_result_4}.

\subsection{Proof of Result in \Cref{sec:background}} \label{apx:proof_KSD_derivation}

The following properties of the Stein operator $\S_\Q$ will be useful:

\begin{lemma} \label{lem:ub_kd2}
	Under \Cref{asmp:derivation_KSD}, we have, for all $x, x' \in \mathcal{X}$ and $h \in \U$,
	\begin{enumerate}[label=(\roman*)]
	    \item $\S_\Q K(x, \cdot) \in \H$ , \label{item: var stein 1}
	    \item $\S_\Q[h](x) = \langle h(\cdot), \S_\Q K(x, \cdot) \rangle_{\H}$ , \label{item: var stein 2}
	    \item $| \S_\Q \S_\Q K(x, x') | 
	\le \sqrt{ \S_\Q \S_\Q K(x, x) } \sqrt{ \S_\Q \S_\Q K(x', x') }$ . \label{item: var stein 3}
	\end{enumerate}
\end{lemma} 

\begin{proof}
	First of all, since $h \mapsto \S_\Q[h](x)$ is a continuous linear functional on $\H$ for each fixed $x \in \X$ by assumption, from the Riesz representation theorem \cite[Theorem A.5.12]{Steinwart2008} there exists a \emph{representer} $g_x \in \H$ for each fixed $x \in \X$ s.t.
	\begin{align*}
	\S_\Q[h](x) = \langle h, g_x \rangle_{\H} .
	\end{align*}
	Second of all, the reproducing property $h(x') = \langle h(\cdot), K(\cdot, x') \rangle_\H$ holds for any $h \in \H$, where we recall that the inner product between $h \in \H$ and a matrix-valued function $K(x, \cdot)$ is defined in \Cref{sec:kernel_RKHS}.
	By the reproducing property, for all $x, x' \in \X$,
	\begin{align}
	g_x(x') = \langle g_x, K(\cdot, x') \rangle_{\H} = \S_\Q\left[ K(\cdot, x') \right](x) = \S_\Q K(x, x') . \label{eq:deriv_KSD_lem_eq1}
	\end{align}
	In particular, $\S_\Q K(x, \cdot) \in \H$ since $g_x \in \H$, establishing \cref{item: var stein 1}.
	Based on these two observations, we can rewrite $\S_\Q[h](x)$ at each fixed $x \in \X$ as 
	\begin{align}
	\S_\Q[h](x) & = \langle h, g_x \rangle_{\H} = \langle h(\cdot), \S_\Q K(x, \cdot) \rangle_{\H} , \label{eq:deriv_KSD_lem_eq2}
	\end{align}
	establishing \cref{item: var stein 2}.
	We now apply \eqref{eq:deriv_KSD_lem_eq2} with $h(\cdot) = \S_\Q K(x',\cdot)$ to deduce that
	\begin{align}
	\S_\Q \S_\Q K(x', x) = \S_\Q \left[ \S_\Q K(x', \cdot) \right](x) = \langle \S_\Q K(x', \cdot), \S_\Q K(x, \cdot) \rangle_{\H} .
	\label{eq: various stein}
	\end{align}
	Applying the Cauchy-Schwarz inequality, 
	\begin{align*}
	| \S_\Q \S_\Q K(x, x') | = | \langle \S_\Q K(x, \cdot), \S_\Q K(x', \cdot) \rangle_{\H} | \le \| \S_\Q K(x, \cdot) \|_\H \| \S_\Q K(x', \cdot) \|_\H .
	\end{align*}
	Here for each $x \in \X$ the norm term can computed using \eqref{eq: various stein}:
	\begin{align*}
	\| \S_\Q K(x, \cdot) \|_\H = \sqrt{ \langle \S_\Q K(x, \cdot), \S_\Q K(x, \cdot) \rangle_{\H} } = \sqrt{ \S_\Q \S_\Q K(x, x) } 
	\end{align*}
	Therefore for all $x, x' \in \X$ we have
	\begin{align*}
	| \S_\Q \S_\Q K(x, x') | \le \sqrt{ \S_\Q \S_\Q K(x, x) } \sqrt{ \S_\Q \S_\Q K(x', x') } ,
	\end{align*}
	establishing \cref{item: var stein 3}.
\end{proof}

\subsubsection{Proof of \Cref{prop:derivation_KSD}}

\begin{proof}
	From \cref{item: var stein 2} of \Cref{lem:ub_kd2}, for each $x \in \X$, $h \in \H$, we have
	\begin{align*}
	\S_\Q[h](x) & = \langle h(\cdot), \S_\Q K(x, \cdot) \rangle_{\H} .
	\end{align*}
	Taking the expectation of both sides,
	\begin{align}
	\E_{X \sim \P}\left[ \S_\Q[h](X) \right] = \E_{X \sim \P}\left[ \langle h(\cdot), \S_\Q K(X, \cdot) \rangle_{\H} \right] = \left\langle h(\cdot), {\E}_{X \sim \P}\left[ \S_\Q K(X, \cdot) \right] \right\rangle_\H . \label{eq:deriv_KSD_eq1}
	\end{align}
	Here since the inner product is continuous liner operator, the expectation and inner product can be exchanged if the function $x \mapsto \S_\Q K(x, \cdot)$ is \textit{Bochner $\P$-integrable} \citep[A.32]{Steinwart2008}.
	This is indeed the case, since from \cref{item: var stein 2} of \Cref{lem:ub_kd2} again, and Jensen's inequality,
	\begin{align*}
	{\E}_{X \sim \P}\left[ \| \S_\Q K(X, \cdot) \|_\H \right] & = {\E}_{X \sim \P}\left[ \sqrt{ \langle \S_\Q K(X, \cdot), \S_\Q K(X, \cdot) \rangle_\H } \right] \\
	& = {\E}_{X \sim \P}\left[ \sqrt{ \S_\Q \S_\Q K(X, X) } \right] \le \sqrt{ {\E}_{X \sim \P}\left[ \S_\Q \S_\Q K(X, X) \right] } < \infty 
	\end{align*}
	where the last term is finite by \Cref{asmp:derivation_KSD}.
	A standard argument based on the Cauchy--Schwarz inequality gives
	\begin{align}
	\sup_{\| h \|_{\H} \le 1} \left| \left\langle h(\cdot), \E_{X \sim \P}\big[ \S_\Q K(X, \cdot) \big] \right\rangle_\H \right| & = \left\| \E_{X \sim \P}\big[ \S_\Q K(X, \cdot) \big] \right\|_\H \nonumber \\
	& = \sqrt{ \big\langle \E_{X \sim \P}\left[ \S_\Q K(X, \cdot) \right], \E_{X' \sim \P}\left[ \S_\Q K(X', \cdot) \right] \big\rangle_\H } \nonumber \\
	& = \sqrt{ \E_{X, X' \sim \P}\left[ \big\langle \S_\Q K(X, \cdot), \S_\Q K(X', \cdot) \big\rangle_\H \right] } \nonumber \\
	& = \sqrt{ \E_{X,X' \sim \P}\left[ \S_\Q \S_\Q K(X, X') \right] } \label{eq:deriv_KSD_eq2}
	\end{align}
	where $X$ and $X'$ are independent, and we again appeal to Bochner $\P$-integrability to interchange expectation and inner product.
	Thus from \eqref{eq:deriv_KSD_eq1} and \eqref{eq:deriv_KSD_eq2} we have
	\begin{align*}
	\operatorname{KSD}^2(\Q \| \P) = \left( \sup_{\| h \|_{\H} \le 1} \Big| \E_{X \sim \P}\left[ \S_\Q[h](X) \right] \Big| \right)^2 = \E_{X,X' \sim \P}\left[ \S_\Q \S_\Q K(X, X') \right] ,
	\end{align*}
	as claimed.
\end{proof}

\subsubsection{Verifying \Cref{asmp:derivation_KSD} for the Langevin Stein Operator } \label{sec:verify_asmp_1}

This section demonstrates how to verify the assumption that $h \mapsto \S_\Q[h](x)$ is a continuous linear functional on $\H$ for each fixed $x \in \X$ in the case where $\S_\Q$ is the Langevin Stein operator \eqref{eq:stein_operator} for $\Q \in \mathcal{P}_{\text{S}}(\R^d)$.
Since a linear functional is continuous if and only if it is bounded, we aim to show that, for each fixed $x \in \X$, there exist a constant $C_x$ s.t. $| \S_\Q[h](x) | \le C_x \| h \|_\H$ for all $h \in \H$.

For each fixed $x \in \R^d$, the Langevin Stein operator $\S_\Q$ is given as
\begin{align*}
\S_\Q[h](x) = \nabla \log q(x) \cdot h(x) + \nabla \cdot h(x) .
\end{align*}
From the reproducing property $h(x) = \langle h, K(x, \cdot) \rangle_\H$ for any $h \in \H$, we have
\begin{align*}
\S_\Q[h](x) & = \nabla \log q(x) \cdot \langle h, K(x, \cdot) \rangle_\H + \nabla_x \cdot \langle h, K(x, \cdot) \rangle_\H \\
& = \langle h, K(x, \cdot) \nabla \log q(x) \rangle_\H + \langle h, \nabla_x \cdot K(x, \cdot) \rangle_\H
\end{align*}
where the order of inner product and other operators is exchangeable by the continuity of $\langle h, \cdot \rangle_{\H}: \H \to \R$ \citep[][Corollary~4.36]{Steinwart2008}.
Then by the Cauchy--Schwarz inequality, 
\begin{align*}
| \S_\Q[h](x) | & \le \Big( \| K(x, \cdot) \nabla \log q(x) \|_\H + \| \nabla_x \cdot K(x, \cdot) \|_\H \Big) \| h \|_\H \\
& = \left( \sqrt{ \nabla \log q(x) \cdot K(x, x) \nabla \log q(x) } + \sqrt{ \nabla \cdot ( \nabla \cdot K(x, x) ) } \right) \| h \|_\H =: C_x \| h \|_\H .
\end{align*}
where the first and second gradient of $\nabla \cdot ( \nabla \cdot K(x, x) )$ are taken each with respect to the first and second argument of $K$.
For the constant $C_x$ to exist, it is sufficient to require that $\nabla \log q(x)$, $K(x, x)$ and $\nabla \cdot ( \nabla \cdot K(x, x) )$ exist.
This is the case when, for example, $\Q \in \mathcal{P}_\S(\R^d)$ and $K \in C_b^{1,1}(\mathbb{R}^d \times \mathbb{R}^d ; \mathbb{R}^{d \times d} )$, as assumed in \cite{Gorham2017}.

\subsection{Proofs of Results in \Cref{sec:methodology}} \label{apx:proof_post_expfam}

\subsubsection{Proof of \Cref{prop:post_expfam}} \label{apx:proof_post_expfam_form}

\begin{proof}
	From \eqref{eq:langevin_k0}, $\S_{\P_\theta} \S_{\P_\theta} K$ is given by
	\begin{align*}
	\S_{\P_\theta} \S_{\P_\theta} K(x,x') & \stackrel{+C}{=} \underbrace{ \nabla \log p_\theta(x) \cdot K(x,x') \nabla \log p_\theta(x') }_{(*_1)} \\
	& \hspace{50pt} + \underbrace{ \nabla \log p_\theta(x) \cdot \left( \nabla_{x'} \cdot K(x, x') \right) }_{(*_2)} + \underbrace{ \nabla \log p_\theta(x') \cdot \left( \nabla_{x} \cdot K(x, x') \right) }_{(*_3)} ,
	\end{align*}
	where $\stackrel{+C}{=}$ indicates equality up to an additive term that is $\theta$-independent.
	The exponential family model in \eqref{eq:expfam} satisfies $\nabla \log p_\theta(x) = \nabla t(x) \eta(\theta) + \nabla b(x)$.
	Thus for term $(*_1)$ we have
	\begin{align}
    \sum_{i=1}^{n} \sum_{j=1}^{n} (*_1) & = \sum_{i=1}^{n} \sum_{j=1}^{n} ( \nabla t(x_i) \eta(\theta) ) \cdot K(x_i, x_j) \nabla t(x_j) \eta(\theta) + \nabla b(x_i) \cdot K(x_i, x_j) \nabla t(x_j) \eta(\theta) \nonumber \\
	& \hspace{90pt} + ( \nabla t(x_i) \eta(\theta) ) \cdot K(x_i, x_j) \nabla b(x_j) + \nabla b(x_i) \cdot K(x_i, x_j) \nabla b(x_j) \nonumber \\
	& \stackrel{+C}{=} \eta(\theta) \cdot \left(  \sum_{i=1}^{n} \sum_{j=1}^{n} \nabla t(x_i)^\top K(x_i, x_j) \nabla t(x_j) \right) \eta(\theta) \nonumber \\
	& \hspace{90pt} + \eta(\theta) \cdot \left( 2 \sum_{i=1}^{n} \sum_{j=1}^{n} \nabla t(x_i)^\top K(x_i, x_j) \nabla b(x_j) \right)  \label{eq: term1 final}
	\end{align}
	where the last equality follows from symmetry of $K$.
	For terms $(*_2)$ and $(*_3)$,
	\begin{align}
	\sum_{i=1}^{n} \sum_{j=1}^{n} (*_2) & =  \sum_{i=1}^{n} \sum_{j=1}^{n} ( \nabla t(x_i) \eta(\theta) ) \cdot ( \nabla_{x'} \cdot K(x_i, x_j) ) + \nabla b(x_i) \cdot ( \nabla_{x'} \cdot K(x_i, x_j) ) \nonumber \\
	& \stackrel{+C}{=} \eta(\theta) \cdot \left(  \sum_{i=1}^{n} \sum_{j=1}^{n} \nabla t(x_i)^\top ( \nabla_{x'} \cdot K(x_i, x_j) ) \right) , \label{eq: term2 final} \\
	\sum_{i=1}^{n} \sum_{j=1}^{n} (*_3) & =  \sum_{i=1}^{n} \sum_{j=1}^{n} ( \nabla t(x_i) \eta(\theta) ) \cdot ( \nabla_{x} \cdot K(x_i, x_j) ) + \nabla b(x_j) \cdot ( \nabla_{x} \cdot K(x_i, x_j) ) \nonumber \\
	& \stackrel{+C}{=} \eta(\theta) \cdot \left(  \sum_{i=1}^{n} \sum_{j=1}^{n} \nabla t(x_j)^\top ( \nabla_{x} \cdot K(x_i, x_j) ) \right) . \label{eq: term3 final}
	\end{align}
	From \Cref{eq:empirical_KSD}, the KSD-Bayes posterior is
	\begin{align*}
		\pi_n^D(\theta) & \propto \pi(\theta) \exp\left( - \beta n \left\{ \frac{1}{n^2} \sum_{i=1}^{n} \sum_{j=1}^{n} (*_1) + (*_2) + (*_3) \right\} \right) ,
	\end{align*}
	so we may collect together terms in \Cref{eq: term1 final,eq: term2 final,eq: term3 final} to obtain the expressions in \Cref{prop:post_expfam}.
\end{proof}

\subsection{Proofs of Results in \Cref{sec:prop_KSD}} \label{sec:ap_proof_4}

\subsubsection{Proof of \Cref{thm:pw} (a.s. Pointwise Convergence)} \label{apx:proof_pw}

\begin{proof}
	Let $f_n(\theta) := \operatorname{KSD}^2(\P_{\theta} \| \P_n)$ and $f(\theta) := \operatorname{KSD}^2(\P_{\theta} \| \P)$.
	Decomposing the double summation of $f_n(\theta)$ into the diagonal term ($i = j$) and non-diagonal term ($i \ne j$),
	\begin{align*}
	f_n(\theta) & = \frac{1}{n^2} \sum_{i=1}^{n} \S_{\P_\theta} \S_{\P_\theta} K(x_i, x_i) + \frac{1}{n^2} \sum_{i=1}^{n} \sum_{j \ne i}^{n} \S_{\P_\theta} \S_{\P_\theta}(x_i, x_j) \\
	& = \frac{1}{n} \underbrace{ \frac{1}{n} \sum_{i=1}^{n} \S_{\P_\theta} \S_{\P_\theta} K(x_i, x_i) }_{(*_a)} + \frac{n-1}{n} \underbrace{ \frac{1}{n (n-1)} \sum_{i=1}^{n} \sum_{j \ne i}^{n} \S_{\P_\theta} \S_{\P_\theta} K(x_i, x_i) }_{(*_b)} .
	\end{align*}
	Fix $\theta \in \Theta$.
	From the strong law of large number \cite[Theorem~2.5.10]{Durrett2010b},
	\begin{align*}
	(*_a) = \frac{1}{n} \sum_{i=1}^{n} \S_{\P_\theta} \S_{\P_\theta} K(x_i, x_i) \overset{a.s.}{\longrightarrow} \E_{X \sim \P}\left[ \S_{\P_\theta} \S_{\P_\theta} K(X, X) \right],
	\end{align*}
	provided that $ \E_{X \sim \P}\left[ | \S_{\P_\theta} \S_{\P_\theta} K(X, X) | \right] < \infty$.
	From the positivity of $\S_{\P_\theta} \S_{\P_\theta} K(x, x)$, we have $\E_{X \sim \P}\left[ | \S_{\P_\theta} \S_{\P_\theta} K(X, X) | \right] = \E_{X \sim \P}\left[ \S_{\P_\theta} \S_{\P_\theta} K(X, X) \right]$, which has been assumed to exist.
	The form of (b) is called an \emph{unbiased statistic} (or \emph{U-statistic} for short) and \cite{Hoeffding1961} proved the strong law of large numbers
	\begin{align*}
	(*_b) = \frac{1}{n (n-1)} \sum_{i=1}^{n} \sum_{j \ne i}^{n} \S_{\P_\theta} \S_{\P_\theta} K(x_i, x_j) \overset{a.s.}{\longrightarrow} \E_{X,X' \sim \P}\left[ \S_{\P_\theta} \S_{\P_\theta} K(X, X') \right],
	\end{align*}
	whenever $\E_{X,X' \sim \P}\left[ | \S_{\P_\theta} \S_{\P_\theta} K(X, X') | \right] < \infty$.
	From item (iii) of \Cref{lem:ub_kd2} and Jensen's inequality, we have $\E_{X,X' \sim \P}\left[ | \S_{\P_\theta} \S_{\P_\theta} K(X, X') | \right] \le \E_{X \sim \P}\left[ \S_{\P_\theta} \S_{\P_\theta} K(X, X) \right]$ where the right hand side is again assumed to exist.
	Therefore, since $1 / n \to 0$ and $(n - 1) / n  \to 1$,
	\begin{align*}
	f_n(\theta) = \frac{1}{n} (*_a) + \frac{n-1}{n} (*_b) \overset{a.s.}{\longrightarrow} \E_{X,X' \sim \P}\left[ \S_{\P_\theta} \S_{\P_\theta} K(X, X') \right] = f(\theta) ,
	\end{align*}
	where the argument holds for each fixed $\theta \in \Theta$.
\end{proof}

\subsubsection{Proof of \Cref{thm:uc} (a.s. Uniform Convergence)} \label{apx:proof_uc}

\begin{proof}
	Let $f_n(\theta) := \operatorname{KSD}^2(\P_{\theta} \| \P_n)$ and $f(\theta) := \operatorname{KSD}^2(\P_{\theta} \| \P)$.
	Recall that $\Theta \subset \R^p$ is bounded.
	Theorem 21.8 in \citet[]{Davidson1994} implies that $f_n \overset{a.s.}{\longrightarrow} f$ uniformly on $\B$ if and only if (a) $f_n \overset{a.s.}{\longrightarrow} f$ pointwise on $\B$ and (b) $\{ f_n \}_{n=1}^{\infty}$ is strongly stochastically equicontinuous on $\B$.
	The condition (a) is immediately implied by \Cref{thm:pw} and we hence show the condition (b) in the remainder.
	
	By \citet[Theorem~21.10]{Davidson1994}, $\{ f_n \}_{n=1}^{\infty}$ is strongly stochastically equicontinuous on $\B$ if there exists a stochastic sequence $\{ \mathcal{L}_n \}_{n=1}^{\infty}$, independent of $\theta$, s.t.
	\begin{align*}
	| f_n(\theta) - f_n(\theta') | & \le \mathcal{L}_n \| \theta - \theta' \|_2, \quad \forall \theta, \theta' \in \B \qquad \text{and} \qquad \limsup_{n \to \infty} \mathcal{L}_n < \infty \ \text{a.s.} 
	\end{align*}
	Since $f_n$ is continuously differentiable on $\Theta$, and $\Theta$ is assumed to be open and convex, the mean value theorem yields
	\begin{align*}
	| f_n(\theta) - f_n(\theta') | & \le \sup_{\theta \in \B} \| \nabla_{\theta} f_n(\theta) \|_2 \| \theta - \theta' \|_2, \quad \forall \theta, \theta' \in \B.
	\end{align*}
	\Cref{lem:deriv_KSD_2nd3rd_up} (the first of our auxiliary results, stated and proved in \Cref{sec: proof_pre_result_4}) implies that $\sup_{\theta \in \B} \| \nabla_{\theta} f_n(\theta) \|_2 < \infty$ a.s. for all sufficiently large $n$.
	Therefore, setting $\mathcal{L}_n = \sup_{\theta \in \B} \| \nabla_{\theta} f_n(\theta) \|_2$ concludes the proof.
\end{proof}

\subsubsection{Proof of \Cref{lem:sc_ksd} (Strong Consistency)} \label{apx:proof_sc_ksd}

The following result from real analysis will be required:
	
\begin{lemma} \label{lem:consis_det}
	Let $\Theta \subset \mathbb{R}^p$ be open and bounded. 
	Let $f_n: \Theta \rightarrow \mathbb{R}$ and $f: \Theta \rightarrow \mathbb{R}$ be continuous functions. Assume that (i) there exists an unique $\theta_* \in \Theta$ s.t. $f(\theta_*) < \inf_{\{ \theta \in \Theta : \| \theta - \theta_* \|_2 \ge \epsilon \}} f(\theta)$ for any $\epsilon > 0$, and (ii) $\sup_{\theta \in \Theta} | f_n(\theta) - f(\theta)| \rightarrow 0$ as $n \rightarrow \infty$.
	Let $\{ \theta_n \}_{n=1}^\infty$ be any sequence s.t. $\theta_n \in \argmin_{\theta \in \Theta} f_n(\theta)$ for all sufficiently large $n$.
	Then $\theta_n \rightarrow \theta_*$ as $n \rightarrow \infty$.
\end{lemma}

\begin{proof}
	The following argument is similar to that used in \citet[Theorem~5.7]{Vaart1998} and \citet[Theorem~2.1]{Newey1994}.
	Fix $\eta > 0$ and consider $n$ sufficiently large that $\theta_n$ is well-defined.
	From (ii), for all sufficiently large $n$, we have the uniform bound $| f(\theta) - f_n(\theta)| < \eta / 2$ over $\theta \in \Theta$.
	Since $\theta_n$ is a minimiser of $f_n$, we therefore have $f(\theta_n) < f_n(\theta_n) + \eta / 2 < f_n(\theta_*) + \eta / 2 < f(\theta_*) + \eta$.
	Since $\eta > 0$ was arbitrary, we may take $\eta = \inf_{\{ \theta \in \Theta : \| \theta - \theta_* \|_2 \ge \epsilon \}} f(\theta) - f(\theta_*)$, where $\eta > 0$ from (i), to see that $f(\theta_n) < \inf_{\{ \theta \in \Theta : \| \theta - \theta_* \|_2 \ge \epsilon \}} f(\theta)$.
	Thus we have shown that $\theta_n \in \{ \theta \in \Theta : \| \theta - \theta_* \|_2 < \epsilon \}$ for all sufficiently large $n$.
	Since the argument holds for $\epsilon > 0$ arbitrarily small, the result is established.
\end{proof}

Now we can prove \Cref{lem:sc_ksd}:

\begin{proof}[Proof of \Cref{lem:sc_ksd}]
	Let $f_n(\theta) := \operatorname{KSD}^2(\P_{\theta} \| \P_n)$ and $f(\theta) := \operatorname{KSD}^2(\P_{\theta} \| \P)$.
	From \Cref{asmp:sc_cnd}, there exists an unique $\theta_* \in \Theta$ s.t. $f(\theta_*) < \inf_{\{ \theta \in \Theta : \| \theta - \theta_* \|_2 \ge \epsilon \}} f(\theta)$ for any $\epsilon > 0$, and $\theta_n \in \Theta$ minimises $f_n$ a.s. for all sufficiently large $n$.
	Since \Cref{asmp:an_cnd} ($r_{\max}=1$) hold, $f_n$ is continuous a.s. and $\sup_{\theta \in \Theta} | f_n(\theta) - f(\theta)| \overset{a.s.}{\rightarrow} 0$ by \Cref{thm:uc}.
	Thus the conditions of \Cref{lem:consis_det} are a.s. satisfied, from which it follows that $\theta_n \overset{a.s.}{\longrightarrow} \theta_*$.
\end{proof}

\subsubsection{Proof of \Cref{lem:an_ksd} (Asymptotic Normality)} \label{apx:proof_an_ksd}

\begin{proof}
	Let $f_n(\theta) := \operatorname{KSD}^2(\P_{\theta} \| \P_n)$ and $f(\theta) := \operatorname{KSD}^2(\P_{\theta} \| \P)$.
	It was assumed that, for any $h \in \H$ and $x \in \X$, the map $\theta \mapsto \S_{\P_\theta}[h](x)$ is three times continuously differentiable, from which it follows that $f_n$ is three times continuously differentiable as well.
	Since $\theta_n$ minimises $f_n$ for all sufficiently large $n$, we have $\nabla f_n(\theta_n) = 0$.
	Hence a second order Taylor expansion around $\theta_*$ yields
	\begin{align*}
	0 = \nabla f_n(\theta_n) = \nabla f_n(\theta_*) + \nabla^2 f_n(\theta_*) (\theta_n - \theta_*) + (\theta_n - \theta_*) \cdot \nabla^3 f_n( \theta_n' ) (\theta_n - \theta_*) 
	\end{align*}
	where $\theta_n' = \alpha \theta_* + (1-\alpha) \theta_n$ for some $\alpha \in [0,1]$.
	By transposing the terms properly and scaling the both side by $\sqrt{n}$, we have 
	\begin{align*}
	\sqrt{n}( \theta_n - \theta_* ) = \Big( \underbrace{ \nabla^2 f_n(\theta_*) }_{(*_1)} + \underbrace{ (\theta_n - \theta_*) \cdot \nabla^3 f_n( \theta_n' ) }_{(*_2)} \Big)^{-1} \Big( - \underbrace{ \sqrt{n} \nabla f_n(\theta_*) }_{(*_3)} \Big) .
	\end{align*}
	In the remainder, we show the convergence of $(*_1)$, $(*_2)$ and $(*_3)$, and apply the Slutsky's theorem to see the convergence in distribution of $\sqrt{n}( \theta - \theta_n )$.
	
	\vspace{5pt}
	\noindent \textbf{Term $(*_1)$:}
	From the auxiliary result \Cref{lem:deriv_KSD_1st2nd_con} in \Cref{sec: proof_pre_result_4}, we have that $\nabla^2 f_n(\theta_*) \overset{a.s.}{\to} \nabla^2 f(\theta_*) = H_*$ where $H_*$ is positive semi-definite.
	
	\vspace{5pt}
	\noindent \textbf{Term $(*_2)$:}
	From the Cauchy--Schwarz inequality and auxiliary result \Cref{lem:deriv_KSD_2nd3rd_up} in \Cref{sec: proof_pre_result_4},
	\begin{align*}
	\limsup_{n \rightarrow \infty} \big\| (\theta_n - \theta_*) \cdot \nabla^3 f_n( \theta_n' ) \big\|_2 & \le \limsup_{n \rightarrow \infty} \sup_{\theta \in \B} \| \nabla^3 f_n(\theta) \|_2 \| \theta_n - \theta_* \|_2 \\
	& \leq \underbrace{ \limsup_{n \rightarrow \infty} \sup_{\theta \in \B} \| \nabla^3 f_n(\theta) \|_2 }_{< \infty\ \text{a.s.}} \times \limsup_{n \rightarrow \infty} \| \theta_n - \theta_* \|_2 
	\end{align*}
	Since \Cref{lem:sc_ksd} implies that $\| \theta_n - \theta_* \|_2 \overset{a.s.}{\to} 0$, we have $(*_2) \overset{a.s.}{\to} 0$.
	
	\vspace{5pt}
	\noindent \textbf{Term $(*_3)$:}
	Let $F(x, x') := \nabla_\theta ( \S_{\P_{\theta}} \S_{\P_{\theta}} K(x, x') ) |_{\theta=\theta_*} \in \R^p$ and recall that $S(x, \theta_*) = \E_{X \sim \P}\left[ F(x, X) \right] \in \R^p$.
	Then
	\begin{align*}
	\sqrt{n} \nabla f_n(\theta_*) & = \sqrt{n} \left( \frac{1}{n^2} \sum_{i=1}^{n} F(x_i, x_i) + \frac{1}{n^2} \sum_{i=1}^{n} \sum_{j \ne i}^{n} F(x_i, x_j) \right) \\
	& = \frac{1}{\sqrt{n}} \underbrace{ \frac{1}{n} \sum_{i=1}^{n} F(x_i, x_i) }_{(*_a)} + \frac{n-1}{n} \underbrace{ \frac{\sqrt{n}}{n (n-1)} \sum_{i=1}^{n} \sum_{j \ne i}^{n} F(x_i, x_j) }_{(*_b)} .
	\end{align*}
	First, it follows from the strong law of large number \cite[Theorem 2.5.10]{Durrett2010b} that $(*_a) \overset{a.s.}{\to} \E_{X \sim \P}[ F(X, X) ]$ whenever $\E_{X \sim \P}[ \| F(X, X) \big) \|_2 ] < \infty$.
	Second, since $(*_b)$ is a U-statistic multiplied by $\sqrt{n}$, it follows from \citet[Theorem~12.3]{Vaart1998} that $(*_b) \overset{p}{\to} (1 / \sqrt{n}) \sum_{i=1}^{n} S(x_i, \theta_*)$ whenever $\E_{X, X' \sim \P}[ \| F(X, X') \|_2^2 ] < \infty$.
	(Here $\overset{p}{\to}$ denotes convergence in probability.)
	Both the required conditions indeed hold from the auxiliary result \Cref{lem:score_finite_fm} in \Cref{sec: proof_pre_result_4}.
	Thus we have 
	\begin{align*}
	\sqrt{n} \nabla f_n(\theta_*) & = \frac{1}{\sqrt{n}} (*_a) + \frac{n-1}{n} (*_b) \overset{p}{\longrightarrow} \frac{1}{\sqrt{n}} \sum_{i=1}^{n} S(x_i, \theta_*) .
	\end{align*} 
	This convergence in probability implies that $\sqrt{n} \nabla f_n(\theta_*)$ and $(1 / \sqrt{n}) \sum_{i=1}^{n} S(x_i, \theta_*)$ converge in distribution to the same limit.
	Therefore we may apply the central limit theorem for $(1 / \sqrt{n}) \sum_{i=1}^{n} S(x_i, \theta_*)$ to obtain the asymptotic distribution of $\sqrt{n} \nabla f_n(\theta_*)$.
	Again from \citet[Theorem~12.3]{Vaart1998}, we have
	\begin{align*}
	\frac{1}{\sqrt{n}} \sum_{i=1}^{n} S(x_i, \theta_*) \overset{d}{\longrightarrow} \mathcal{N}\left(0, J_* \right) , \qquad J_* = \E_{X \sim \P}\left[ S(X, \theta_*) S(X, \theta_*)^\top \right]
	\end{align*}
    whenever $\E_{X, X' \sim \P}\left[ \| F(X, X') \|_2^2 \right] < \infty$, which implies the existence of the covariance matrix $J_*$.
	Hence $\sqrt{n} \nabla f_n(\theta_*) \overset{d}{\rightarrow} \mathcal{N}\left(0, J_* \right)$.

	\vspace{5pt}
	Collecting together these results, we have shown that
	\begin{align*}
	(*_1) \overset{a.s.}{\longrightarrow} H_*, \qquad (*_2) \overset{a.s.}{\longrightarrow} 0, \qquad (*_3) \overset{d}{\longrightarrow} \mathcal{N}\left(0, J_* \right) .
	\end{align*}
	Since $H_*$ is guaranteed to be at least positive semi-definite, it is in fact strictly positive definite if $H_*$ is non-singular, as we assumed.
	Finally, Slutsky's theorem allows us to conclude that $\sqrt{n}( \theta - \theta_n ) \overset{d}{\rightarrow} \mathcal{N}\left(0, H_*^{-1} J_* H_*^{-1} \right)$ as claimed.
\end{proof}

\subsubsection{Verifying \Cref{asmp:an_cnd} for the Langevin Stein Operator} \label{ap: verify A3}

Here we compute the quantities involved in \Cref{asmp:an_cnd} for the Langevin Stein operator $\S_{\P_{\theta}}$ with $\P_{\theta} \in \mathcal{P}_{\text{S}}(\R^d)$. 
In this case,
\begin{align}
\partial^r \S_{\P_\theta}[h](x) = \partial^r \big( \nabla_{x} \log p_{\theta}(x) \cdot h(x) \big) + \partial^r \big( \nabla_{x} \cdot h(x) \big) = \big( \partial^r \nabla_{x} \log p_{\theta}(x) \big) \cdot h(x) . \label{eq:eg_deriv1_score}
\end{align}
The operator $\partial^r \S_{\P_\theta}$ in \eqref{eq:eg_deriv1_score} is therefore well-defined and $\theta \mapsto \partial^r \S_{\P_\theta}[h](x)$ is continuous whenever $\theta \mapsto \nabla_{x} \log p_{\theta}(x)$ is $r$-times continuously differentiable over $\B$.
For each fixed $x \in \X$, it is clear that $h \mapsto ( \partial^r \S_{\P_\theta} )[h](x)$ is a continuous linear functional on $\H$.
Then the term $( \partial^r \S_{\P_\theta} ) ( \partial^r \S_{\P_\theta} ) K(x, x)$ appearing in the final part of \Cref{asmp:an_cnd} takes the explicit form
\begin{align}
( \partial^r \S_{\P_\theta} ) ( \partial^r \S_{\P_\theta} ) K(x, x) = \big( \partial^r \nabla_{x} \log p_{\theta}(x) \big) \cdot K(x, x) \big( \partial^r \nabla_{x} \log p_{\theta}(x) \big) . \label{eq:eg_deriv1_kernel}
\end{align}
The regularity of \eqref{eq:eg_deriv1_kernel} therefore depends on $K$ and $\P_\theta$.
See \Cref{apx: expfam_asmp}, where \eqref{eq:eg_deriv1_kernel} is computed for an exponential family model.

\subsection{Proof of \Cref{thm:pc} (Posterior Consistency)} \label{apx:proof_pc}

The following preliminary lemma is required, which takes inspiration from \citet{Alquier2016,Cherief-Abdellatif2019}.
Let $f_n(\theta) = \operatorname{KSD}^2(\P_\theta \| \P_n)$ and $f(\theta) = \operatorname{KSD}^2(\P_\theta \| \P)$.

\begin{lemma} \label{lem:pc_lem}
	Suppose \Cref{asmp:sc_cnd} and \Cref{asmp:pc_cnd} hold.
	For all $\delta \in (0,1]$, with probability at least $1 - \delta$,
	\begin{align*}
	\int_{\Theta} f(\theta) \pi_{n}^{D}(\theta) \mathrm{d} \theta \le f(\theta_*) + \left( \alpha_1 + \alpha_2 + \frac{ 8 \sup_{\theta \in \Theta} \sigma(\theta) }{\delta} \right) \frac{1}{\sqrt{n}} .
	\end{align*}
	where the probability is taken with respect to realisations of the dataset $\{x_i\}_{i=1}^n \overset{i.i.d.}{\sim} \P$.
\end{lemma}

\begin{proof}
	From the auxiliary result \Cref{thm:1st_con_ksd2} in \Cref{sec: proof_pre_result_4}, we have a concentration inequality 
	\begin{align}
	\mathbb{P}\left( | f_n(\theta) - f(\theta) | \ge \delta \right) & \le \frac{4 \sigma(\theta)}{\delta \sqrt{n}} \le \frac{4 \sup_{\theta \in \Theta} \sigma(\theta)}{\delta \sqrt{n}} \label{eq:pc_eq1}
	\end{align}
	for each $\theta \in \Theta$, where the probability is taken with respect to the samples $X_1, \ldots, X_n \overset{i.i.d.}{\sim} \P$.
	Taking the complement and re-scaling $\delta$, \eqref{eq:pc_eq1} is equivalent to
	\begin{align}
	\mathbb{P}\bigg( | f_n(\theta) - f(\theta) | \le \frac{4 \sup_{\theta \in \Theta} \sigma(\theta)}{\delta \sqrt{n}} \bigg) & \ge 1 - \delta. \label{eq:pc_eq2}
	\end{align}
	Notice that by virtue of the absolute value, the following inequalities hold simultaneously with probability at least $1-\delta$:
	\begin{align}
	f(\theta) & \le f_n(\theta) + \frac{4 \sup_{\theta \in \Theta} \sigma(\theta)}{\delta \sqrt{n}} . \label{eq:pc_eq2_a1} \\
	f_n(\theta) & \le f(\theta) + \frac{4 \sup_{\theta \in \Theta} \sigma(\theta)}{\delta \sqrt{n}} . \label{eq:pc_eq2_a2}
	\end{align}
	Taking an expectation with respect to the generalised posterior on both side of \eqref{eq:pc_eq2_a1} yields, with probability at least $1 - \delta$,
	\begin{align*}
	\int_\Theta f(\theta)  \pi_{n}^{D}(\theta) \mathrm{d} \theta & \le \int_\Theta f_n(\theta)  \pi_{n}^{D}(\theta) \mathrm{d} \theta + \frac{4 \sup_{\theta \in \Theta} \sigma(\theta)}{\delta \sqrt{n}}
	\end{align*}
	In order to apply the identity \eqref{eq:gen-posterior-opt} of \citet[Theorem~1]{Knoblauch2019}, we add the term $(1 / n) \operatorname{KL}(\pi_{n}^{D} \| \pi) \ge 0$ in the right hand side and see that, with probability at least $1 - \delta$,
	\begin{align*}
	\int_\Theta f(\theta)  \pi_{n}^{D}(\theta) \mathrm{d} \theta \le \frac{1}{n} \left\{ \int_\Theta n f_n(\theta)  \pi_{n}^{D}(\theta) \mathrm{d} \theta + \operatorname{KL}(\pi_{n}^{D} \| \pi) \right\} + \frac{4 \sup_{\theta \in \Theta} \sigma(\theta)}{\delta \sqrt{n}} .
	\end{align*}
	Then from the identity \eqref{eq:gen-posterior-opt}, the bracketed term on the right hand side is the solution to the following variational problem over $\mathcal{P}(\Theta)$:
	\begin{align}
	\int_\Theta f(\theta)  \pi_{n}^{D}(\theta) \mathrm{d} \theta & \le \frac{1}{n} \inf_{\rho \in \mathcal{P}(\Theta)} \left\{ \int_\Theta n f_n(\theta)  \rho(\theta) \mathrm{d} \theta + \text{KL}(\rho \| \pi) \right\} + \frac{4 \sup_{\theta \in \Theta} \sigma(\theta)}{\delta \sqrt{n}} \nonumber \\
	& = \inf_{\rho \in \mathcal{P}(\Theta)} \left\{ \int_\Theta f_n(\theta)  \rho(\theta) \mathrm{d} \theta + \frac{1}{n} \text{KL}(\rho \| \pi) \right\} + \frac{4 \sup_{\theta \in \Theta} \sigma(\theta)}{\delta \sqrt{n}}. \label{eq:pc_eq3}
	\end{align}
	Plugging \eqref{eq:pc_eq2_a2} in \eqref{eq:pc_eq3}, we have with probability at least $1 - \delta$,
	\begin{align}
	\int_\Theta f(\theta)  \pi_{n}^{D}(\theta) \mathrm{d} \theta & \le \inf_{\rho \in \mathcal{P}(\Theta)} \left\{ \int_\Theta f(\theta)  \rho(\theta) \mathrm{d} \theta + \frac{1}{n} \text{KL}(\rho \| \pi) \right\} + \frac{8 \sup_{\theta \in \Theta} \sigma(\theta)}{\delta \sqrt{n}}. \label{eq:pc_eq4}
	\end{align}
	Plugging the trivial bound $f(\theta) \le f(\theta_*) + | f(\theta) - f(\theta_*) |$ into \eqref{eq:pc_eq4}, we have
	\begin{align*}
	\eqref{eq:pc_eq4} \le f(\theta_*) + \inf_{\rho \in \mathcal{P}(\Theta)} \bigg\{ \int_\Theta \left| f(\theta) - f(\theta_*) \right|  \rho(\theta) \mathrm{d} \theta + \frac{1}{n} \text{KL}(\rho \| \pi) \bigg\} + \frac{8 \sup_{\theta \in \Theta} \sigma(\theta)}{\delta \sqrt{n}} .
	\end{align*}
	Notice that the infimum term can be upper bounded by any choice of $\rho \in \mathcal{P}(\Theta)$.
	Letting $\Pi(B_n) := \int_{B_n} \pi(\theta) \mathrm{d} \theta$, we take $\rho(\theta) =  \pi(\theta) / \Pi(B_n)$ for $\theta \in B_n$ and $\rho(\theta) = 0$ for $\theta \not\in B_n$.
	Then \Cref{asmp:pc_cnd} part (2) ensures that $\int_{B_n} | f(\theta) - f(\theta_*) |  \rho(\theta) \mathrm{d} \theta \leq \alpha_1/\sqrt{n}$ and that $\mathrm{KL}(\rho \| \pi) = \int_\Theta \log\left( \rho(\theta) / \pi(\theta) \right) \rho(\theta) \mathrm{d} \theta = \int_{B_n} - \log( \Pi(B_n) ) \pi(\theta) \mathrm{d} \theta / \Pi(B_n) = - \log \Pi(B_n) \le \alpha_2 \sqrt{n}$.
	Thus
	\begin{align}
	\int_\Theta f(\theta)  \pi_{n}^{D}(\theta) \mathrm{d} \theta \le f(\theta_*) + \left( \alpha_1 + \alpha_2 + \frac{8 \sup_{\theta \in \Theta} \sigma(\theta)}{\delta} \right) \frac{1}{\sqrt{n}} , \label{eq:pc_eq5}
	\end{align}
	with probability at least $1 - \delta$, as claimed.
\end{proof}

\vspace{5pt}
Now we turn to the proof of \Cref{thm:pc}:

\vspace{5pt}
\begin{proof}[Proof of \Cref{thm:pc}]
	Since $\theta_*$ uniquely minimise $f$,
	\begin{align*}
	f(\theta) - f(\theta_*) \ge 0, \quad \forall \theta \in \Theta \qquad \Longrightarrow \qquad \int_\Theta f(\theta)  \pi_{n}^{D}(\theta) \mathrm{d} \theta - f(\theta_*) \ge 0 .
	\end{align*}
	Thus, from \Cref{lem:pc_lem}, 
	\begin{align*}
	\mathbb{P}\left( \left| \int_\Theta f(\theta) \pi_{n}^{D}(\theta) \mathrm{d} \theta - f(\theta_*) \right| \le \left( \alpha_1 + \alpha_2 + \frac{8 \sup_{\theta \in \Theta} \sigma(\theta)}{\delta} \right) \frac{1}{\sqrt{n}} \right) \ge 1 - \delta .
	\end{align*}
	Applying the simplifying upper bound
	$$
	\alpha_1 + \alpha_2 + \frac{8 \sup_{\theta \in \Theta} \sigma(\theta)}{\delta} \leq \frac{\alpha_1 + \alpha_2 + 8 \sup_{\theta \in \Theta} \sigma(\theta)}{\delta} ,
	$$
	taking complement of the probability and performing a change of variables, we obtain the stated result.
\end{proof}

\subsection{Proof of \Cref{thm:lan} (Bernstein--von Mises)} \label{apx:proof_lan}

In this section we define the notation $f_n(\theta) := \operatorname{KSD}^2(\P_{\theta} \| \P_n)$ and $f(\theta) := \operatorname{KSD}^2(\P_{\theta} \| \P)$.
Similarly, denote $H_n := \nabla_{\theta}^2 f_n(\theta_n)$ and $H_* = \nabla_{\theta}^2 f(\theta_*)$.
Our aim is to verify the conditions of Theorem 4 in \citet{Miller2019}.
The following technical lemma lists and establishes the conditions that are required:

\begin{lemma} \label{lem:bvm_precnd}
	Suppose that \Cref{asmp:an_cnd} ($r_{\max}=3$), \Cref{asmp:sc_cnd}, and part (1) of \Cref{asmp:pc_cnd} hold.
	Assume that $H$ is nonsingular.
	Then the following statements almost surely hold:
	\begin{enumerate}
		\item the prior density $\pi$ is continuous at $\theta_*$ and $\pi(\theta_*) > 0$,
		\item $\theta_n \to \theta_*$, 
		\item the Taylor expansion $f_n(\theta) = f_n(\theta_n) + \frac{1}{2}(\theta - \theta_n) \cdot H_n (\theta - \theta_n) + r_n( \theta - \theta_n )$ holds on $\B$, where the remainder $r_n$ satisfies $| r_n( \vartheta ) | \le C \| \vartheta \|_2^3$ for all $\|\vartheta\|_2 \leq \epsilon$, all sufficiently large $n$ and some $C$ and $\epsilon > 0$,
		\item $H_n \to H_*$, where $H_n$ is symmetric and $H_*$ is positive definite,
		\item $\liminf_{n\to\infty} ( \inf_{\{\theta \in \Theta : \|\theta - \theta_n\|_2 \ge \epsilon \}} f_n(\theta) - f_n(\theta_n) ) > 0$ for any $\epsilon > 0$.
	\end{enumerate}
\end{lemma}

\begin{proof}
	We sequentially prove each statement in the list.
	
	\vspace{5pt}
	\noindent \textbf{Part (1):} Directly assumed in \Cref{asmp:pc_cnd} part (1).
	
	\vspace{5pt}
	\noindent \textbf{Part (2):} \Cref{asmp:an_cnd} ($r_{\max}=3$) and \ref{asmp:sc_cnd} are sufficient for \Cref{lem:sc_ksd} and hence $\theta_n \overset{a.s.}{\to} \theta_*$.
	
	\vspace{5pt}
	\noindent \textbf{Part (3):} From \Cref{asmp:an_cnd} ($r_{\max}=3$), for all $h \in \H$ and $x \in \X$ the map $\theta \mapsto \S_{\P_\theta}[h](x)$ is three times continuously differentiable, meaning that $f_n$ is three times continuously differentiable on $\B$. 
	Hence a second order Taylor expansion gives
	\begin{align*}
	f_n(\theta) = f_n(\theta_n) + \nabla f_n(\theta_n) (\theta - \theta_n) + \frac{1}{2}(\theta - \theta_n) \cdot H_n (\theta - \theta_n) + r_n( \theta - \theta_n )
	\end{align*}
	where, for all sufficiently large $n$, $\nabla f_n(\theta_n) = 0$ was assumed and the mean value form of the remainder term $r_n$ in the Taylor expansion provides a bound
	\begin{align*}
	| r_n( \theta - \theta_n ) | \le \sup_{\theta \in \B} \| \nabla^3 f_n(\theta) \|_2 \| \theta - \theta_n \|_2^3 .
	\end{align*}
	Finally, $\limsup_{n \to \infty} \sup_{\theta \in \B} \| \nabla^3 f_n(\theta) \|_2 < \infty$ a.s. by the auxiliary \Cref{lem:deriv_KSD_2nd3rd_up} in \Cref{sec: proof_pre_result_4}.
	
	\vspace{5pt}
	\noindent \textbf{Part (4):} 
	$H_n$ is symmetric since the assumed regularity of $f_n$ allows the mixed second order partial derivatives of $f_n$ to be interchanged.
	The auxiliary \Cref{lem:deriv_KSD_1st2nd_con} in in \Cref{sec: proof_pre_result_4} establishes that $H_n \overset{a.s.}{\to} H_*$ where $H_*$ is positive semi-definite. 
	Thus, since we assumed $H_*$ is nonsingular, it follows that $H_*$ is positive definite.
	
	\vspace{5pt}
	\noindent \textbf{Part (5):} The inequality $\liminf_{n\to\infty}(a_n + b_n) \ge \liminf_{n\to\infty} a_n + \liminf_{n\to\infty} b_n$ holds for any sequences of $a_n, b_n \in \R$.
	Combining the property $\liminf_{n\to\infty}( - b_n ) = - \limsup_{n\to\infty} b_n$, we have that $\liminf_{n\to\infty}(a_n - b_n) \ge \liminf_{n\to\infty} a_n - \limsup_{n\to\infty} b_n$.
	Applying this inequality, 
	\begin{align*}
	\liminf_{n\to\infty} \left( \inf_{\{\theta \in \Theta : \|\theta - \theta_n\|_2 \ge \epsilon \}} f_n(\theta) - f_n(\theta_n) \right) \ge \liminf_{n\to\infty} \inf_{\{\theta \in \Theta : \|\theta - \theta_n\|_2 \ge \epsilon \}} f_n(\theta) - \limsup_{n\to\infty} f_n(\theta_n) =: (*)
	\end{align*}
	Since $f_n(\cdot) \overset{a.s.}{\to} f(\cdot)$ uniformly on $\B$ by \Cref{thm:uc} and $\theta_n \overset{a.s.}{\to} \theta_*$ by \Cref{lem:sc_ksd},
	\begin{align*}
	(*) & \overset{a.s.}{=} \inf_{\{\theta \in \Theta : \|\theta - \theta_*\|_2 \ge \epsilon \}} f(\theta) - f(\theta_*) > 0
	\end{align*}
	where the last inequality follows from \Cref{asmp:sc_cnd}.
\end{proof}

Now we turn to the main proof:

\vspace{5pt}
\begin{proof}[Proof of \Cref{thm:lan}]
	Our aim is to verify the conditions of Theorem 4 in \cite{Miller2019}.
	Note that this result in \cite{Miller2019} views $\{ f_n \}_{n=1}^{\infty}$ as a deterministic sequence; we therefore aim to show that the conditions of Theorem 4 in \cite{Miller2019} are a.s. satisfied by our random sequence $\{ f_n \}_{n=1}^{\infty}$.
	
	Recall that the generalised posterior has p.d.f. $\pi_{n}^{D}(\theta) \propto \exp\left( - n f_n(\theta) \right) \pi(\theta)$ defined on $\Theta \subset \R^p$. 
	This p.d.f. can be trivially extended to a p.d.f. on $\R^p$ by defining $\pi(\theta) = 0$ and (e.g.) $f_n(\theta) = \inf_{\theta \in \Theta} f_n(\theta) + 1$ for all $\theta \in \R^p \setminus \Theta$.
	This brings us into the setting of  \citet{Miller2019}.
	The assumptions of \citet[Theorem 4]{Miller2019} are precisely the list in the statement of \Cref{lem:bvm_precnd}, and the conclusion is that
	\begin{align}
	\int_{\R^p} \left| \hat{\pi}_n^D(\theta) - \frac{1}{\det( 2 \pi H_*^{-1}  )^{1/2}} \exp\left( - \frac{1}{2} \theta \cdot H_* \theta \right) \right| \mathrm{d} \theta \rightarrow 0 . \label{eq: miller a.s.}
	\end{align}
	Thus, since from \Cref{lem:bvm_precnd} the assumptions of \citet[Theorem 4]{Miller2019} are a.s. satisfied, the conclusion in \Cref{eq: miller a.s.} a.s. holds, as claimed.
\end{proof}

\subsection{Proof of Robustness Results} \label{apx:bias-robust}

\subsubsection{Proof of \Cref{lem:bias-robust_cnd}} \label{apx:proof_bias-robust_cnd}

\begin{proof}
	First of all, (17) of \cite{Ghosh2016} demonstrates that
	\begin{align*}
	\operatorname{PIF}(y, \theta, \P_n) & = \beta n \pi_n^L(\theta) \left( - \operatorname{D}L(y, \theta, \P_n) + \int_\Theta \operatorname{D}L(y, \theta', \P_n) \pi_n^L(\theta') \mathrm{d}\theta' \right) .
	\end{align*}
	By Jensen's inequality, we have an upper bounded 
	\begin{align*}
	\sup_{\theta \in \Theta} \sup_{y \in \X} | \operatorname{PIF}(y, \theta, \P_n) | & \le \beta n \sup_{\theta \in \Theta}  \pi_n^L(\theta) \left( \sup_{y \in \X} \left| \operatorname{D}L(y, \theta, \P_n) \right| + \int_\Theta  \sup_{y \in \X} \left| \operatorname{D}L(y, \theta', \P_n) \right| \pi_n^L(\theta') \mathrm{d}\theta' \right) .
	\end{align*}
	Recall that $\pi_n^L(\theta) = \pi(\theta) \exp( - \beta n L(\theta; \P_n) ) / Z$ where $0 < Z < \infty$ is the normalising constant.
	Thus we can obtain an upper bound $\pi_n^L(\theta) \le \pi(\theta) \exp( - \beta n \inf_{\theta \in \Theta} L_n(\theta; \P_n) ) / Z =: C \pi(\theta)$ for some constant $0 < C < \infty$, since $L_n(\theta; \P_n)$ is lower bounded by assumption and $n$ is fixed.
	From this upper bound, we have
	\begin{align}
	\sup_{\theta \in \Theta} \sup_{y \in \X} | \operatorname{PIF}(y, \theta, \P_n) | & \le \beta n C \sup_{\theta \in \Theta} \pi(\theta) \left( \sup_{y \in \X} \left| \operatorname{D}L(y, \theta, \P_n) \right| + C \int_{\Theta}  \sup_{y \in \X} \left| \operatorname{D}L(y, \theta', \P_n) \right| \pi(\theta') \mathrm{d} \theta' \right) \nonumber \\
	& \le \beta n C \sup_{\theta \in \Theta} \left( \pi(\theta) \sup_{y \in \X} \left| \operatorname{D}L(y, \theta, \P_n) \right| \right) + \nonumber \\
	& \hspace{90pt} \beta n C^2 \left( \sup_{\theta \in \Theta} \pi(\theta) \right)  \int_{\Theta}  \sup_{y \in \X} \left| \operatorname{D}L(y, \theta', \P_n) \right| \pi(\theta') \mathrm{d} \theta' . \nonumber
	\end{align}
	Since $\sup_{\theta \in \Theta} \pi(\theta) < \infty$ by assumption in the statement of \Cref{lem:bias-robust_cnd},  
	it follows that
	\begin{align*}
	\sup_{\theta \in \Theta} \left( \pi(\theta) \sup_{y \in \X} \left| \operatorname{D}L(y, \theta, \P_n) \right| \right) < \infty \quad \text{ and } \quad \int_{\Theta} \pi(\theta) \sup_{y \in \X} \left| \operatorname{D}L(y, \theta, \P_n) \right| \mathrm{d} \theta < \infty
	\end{align*}
	are sufficient conditions for $\sup_{\theta \in \Theta} \sup_{y \in \X} | \operatorname{PIF}(y, \theta, \P_n) | < \infty$, as claimed.
\end{proof}

\subsubsection{The Form of $\operatorname{D}L(y, \theta, \P_n)$ for KSD} \label{sec:GD_KSD}

The following lemma clarifies the form of $\operatorname{D}L(y, \theta, \P_n)$ for KSD:

\begin{lemma} \label{lem:bias-robust_GD}
	For $L(\theta; \P_{n,\epsilon,y}) = \operatorname{KSD}^2(\P_{\theta} \| \P_{n,\epsilon,y})$, we have
	\begin{align}
	\operatorname{D}L(y, \theta, \P_n) = 2  \E_{X \sim \P_n}\big[ \S_{\P_\theta} \S_{\P_\theta} K(X, y) \big] - 2 \E_{X,X' \sim \P_n}\big[ \S_{\P_\theta} \S_{\P_\theta} K(X, X') \big] . \label{eq: sensi bound}
	\end{align}
\end{lemma}

\begin{proof}
	From the definition of the $\epsilon$-contamination model as a mixture model, and using the symmetry of $K$, we have
	\begin{align*}
	\operatorname{KSD}^2(\P_{\theta} \| \P_{n,\epsilon,y}) & = {\E}_{X, X' \sim \P_{n,\epsilon,y}}\left[ \S_{\P_\theta} \S_{\P_\theta} K(X, X') \right] \\
	& = (1 - \epsilon)^2 \E_{X,X' \sim \P_n}\left[ \S_{\P_\theta} \S_{\P_\theta} K(X, X') \right] + 2 (1 - \epsilon) \epsilon \E_{X \sim \P_n}\left[ \S_{\P_\theta} \S_{\P_\theta} K(X, y) \right] \\
	& \hspace{40pt}  + \epsilon^2 \S_{\P_\theta} \S_{\P_\theta} K(y, y) .
	\end{align*}
	Direct differentiation then yields
	\begin{align*}
	\operatorname{D}L(y, \theta, \P_n) & = \frac{\mathrm{d}}{\mathrm{d} \epsilon} \operatorname{KSD}^2(\P_{\theta} \| \P_{n,\epsilon,y}) \bigg|_{\epsilon=0} \\
	& = 2  \E_{X \sim \P_n}\big[ \S_{\P_\theta} \S_{\P_\theta} K(X, y) \big] - 2  \E_{X,X' \sim \P_n}\big[ \S_{\P_\theta} \S_{\P_\theta} K(X, X') \big] ,
	\end{align*}
	as claimed.
\end{proof}

\subsubsection{Proof of \Cref{thm:bias-robust}} \label{sec:proof_bias-robust}

\begin{proof}
	From \Cref{lem:bias-robust_cnd} with $\X = \R^d$, it is sufficient to show that
	\begin{align*}
	\text{(i)}\ \sup_{\theta \in \Theta} \left( \pi(\theta) \sup_{y \in \R^d} \left| \operatorname{D}L(y, \theta, \P_n) \right| \right) < \infty \quad \text{ and } \quad \text{(ii)}\ \int_{\Theta} \sup_{y \in \R^d} \left| \operatorname{D}L(y, \theta, \P_n) \right| \pi(\theta) \mathrm{d} \theta < \infty .
	\end{align*}
	To establish (i) and (ii) we exploit the expression for $\operatorname{D}L(y, \theta, \P_n)$ in  \Cref{lem:bias-robust_GD}.
	This furnishes us with the bound
	\begin{align}
	\big| \operatorname{D}L(y, \theta, \P_n) \big| \le 2 \E_{X \sim \P_n}\big[ \underbrace{ | \S_{\P_\theta} \S_{\P_\theta} K(X, y) | }_{=:(*_1)} \big] + 2 \underbrace{ \E_{X,X' \sim \P_n}\big[ \S_{\P_\theta} \S_{\P_\theta} K(X, X') \big] }_{=:(*_2)} . \label{eq:bias-robust_GD_upb}
	\end{align}
	From \Cref{lem:ub_kd2}, $(*_1) \le \sqrt{ \S_{\P_{\theta}} \S_{\P_{\theta}} K(y, y) } \sqrt{ \S_{\P_{\theta}} \S_{\P_{\theta}} K(X, X) }$  and $(*_2) \le \E_{X \sim \P_n}[ \S_{\P_{\theta}} \S_{\P_{\theta}} K(X, X) ]$.
	Plugging these bounds into \eqref{eq:bias-robust_GD_upb} and using Jensen's inequality gives
	\begin{align}
	\eqref{eq:bias-robust_GD_upb} & \le 2 \sqrt{ \S_{\P_{\theta}} \S_{\P_{\theta}} K(y, y) } \sqrt{ \E_{X \sim \P_n}\left[ \S_{\P_{\theta}} \S_{\P_{\theta}} K(X, X) \right] } + 2 \E_{X \sim \P_n}\left[ \S_{\P_{\theta}} \S_{\P_{\theta}} K(X, X) \right] . \label{eq: robust interm bound}
	\end{align}
	Now, observing that 
	\begin{align}
	\E_{X \sim \P_n}\big[ \S_{\P_{\theta}} \S_{\P_{\theta}} K(X, X) \big] \le \E_{X \sim \P_n} \Big[ \sup_{y \in \R^d} ( \S_{\P_{\theta}} \S_{\P_{\theta}} K(y, y) ) \Big] = \sup_{y \in \R^d} \S_{\P_{\theta}} \S_{\P_{\theta}} K(y, y) \label{eq: sup y bound}
	\end{align}
	and taking a supremum over $y$ in \eqref{eq: robust interm bound}, we obtain the bound
	\begin{align}
	\sup_{y \in \R^d} \left| \operatorname{D}L(y, \theta, \P_n) \right| & \le 4 \sup_{y \in \R^d} \S_{\P_{\theta}} \S_{\P_{\theta}} K(y, y) . \label{eq:proof_pd_upb}
	\end{align}
	Therefore, from \eqref{eq:proof_pd_upb}, it suffices to verify the conditions
	\begin{align*}
	\text{(I)}\ \sup_{\theta \in \Theta} \left( \pi(\theta) \sup_{y \in \R^d} \S_{\P_{\theta}} \S_{\P_{\theta}} K(y, y) \right) < \infty \quad \text{ and } \quad \text{(II)}\ \int_{\Theta}  \sup_{y \in \R^d} \S_{\P_{\theta}} \S_{\P_{\theta}} K(y, y) \pi(\theta) \mathrm{d} \theta < \infty ,
	\end{align*}
	which imply the original conditions (i) and (ii).
	To this end, in the remainder we (a) exploit the specific form of $\S_{\P_{\theta}}$ to derive the an explicit upper bound on $\sup_{y \in \R^d} \S_{\P_{\theta}} \S_{\P_{\theta}} K(y, y)$, then (b) verify the conditions (I) and (II) based on this upper bound.
	
	\vspace{5pt}
	\noindent \textbf{Part (a):} 
	By the reproducing property of $K$, the definition of the diffusion Stein operator $\S_{\P_\theta}$, and the fact $(a_1 + a_2)^2 \le 2 ( a_1^2 + a_2^2 )$ for $a_1, a_2 \in \R$, we have the bound 
		\begin{align*}
	\S_{\P_{\theta}} \S_{\P_{\theta}} K(y, y) = \| \S_{\P_{\theta}} K(y, \cdot) \|_{\H}^2 & = \left\| \nabla_{y} \log p_{\theta}(y) \cdot K(y, \cdot) + \nabla_{y} \cdot K(y, \cdot) \right\|_{\H}^2 \\
	& \le 2 \| \nabla_{y} \log p_{\theta}(y) \cdot K(y, \cdot) \|_\H^2 + 2 \left\| \nabla_{y} \cdot K(y, \cdot) \right\|_{\H}^2 . 
	\end{align*}
	For the first term, the reproducing property of $K$ gives that
	\begin{align*}
	\| \nabla_{y} \log p_{\theta}(y) \cdot K(y, \cdot) \|_{\H}^2 & = \nabla_{y} \log p_{\theta}(y) \cdot K(y, y) \nabla_{y} \log p_{\theta}(y) ,
	\end{align*}
	while for the second term, and letting $R(x, x') := \nabla_{x} \cdot ( \nabla_{x'} \cdot K(x, x') )$, the reproducing property gives that
	\begin{align*}
	\left\| \nabla_{y} K(y, \cdot) \right\|_{\H}^2 & = \big\langle \nabla_{y} \cdot K(y, \cdot), \nabla_{y} \cdot K(y, \cdot) \big\rangle_{\H} = R(y, y) .
	\end{align*}
	Thus, taking the supremum with respect to $y \in \R^d$ yields the upper bound, 
	\begin{align*}
	\sup_{y \in \R^d} \S_{\P_{\theta}} \S_{\P_{\theta}} K(y, y) \le 2 \sup_{y \in \R^d} \big( \nabla_{y} \log p_{\theta}(y) \cdot K(y, y) \nabla_{y} \log p_{\theta}(y) \big) + 2 \sup_{y \in \R^d} R(y, y) .
	\end{align*}
	Since $K \in C_b^{1 \times 1}(\R^d \times \R^d)$ by assumption, it follows that $C_{MK} := \sup_{y \in \R^d} R(y, y) < \infty$.
	Thus we have arrived at
	\begin{align}
	\sup_{y \in \R^d} \S_{\P_{\theta}} \S_{\P_{\theta}} K(y, y) & \le 2 \gamma(\theta) + 2 C_{MK} , \label{eq: robust final bound}
	\end{align}
	where $\gamma(\theta)$ was defined in the statement of \Cref{thm:bias-robust}.
	
	\vspace{5pt}
	\noindent \textbf{Part (b):} 
	Now we are in a position to verify conditions (I) and (II).
	For condition (I), we use \eqref{eq: robust final bound} to obtain
	\begin{align*}
	\sup_{\theta \in \Theta} \Big( \pi(\theta) \sup_{y \in \R^d} \S_{\P_{\theta}} \S_{\P_{\theta}} K(y, y) \Big) \le 2 \sup_{\theta \in \Theta} \pi(\theta) \gamma(\theta) + 2 C_{MK} \sup_{\theta \in \Theta} \pi(\theta) 
	\end{align*}
	which is finite by assumption.
	Similarly, for condition (II), we use \eqref{eq: robust final bound} to obtain
	\begin{align*}
	\int_{\Theta}  \sup_{y \in \R^d} \S_{\P_{\theta}} \S_{\P_{\theta}} K(y, y) \pi(\theta) \mathrm{d} \theta \le 2 \int_{\Theta} \pi(\theta) \gamma(\theta) \mathrm{d} \theta + 2 C_{MK} \int_{\Theta} \pi(\theta) \mathrm{d} \theta ,
	\end{align*}
	which is also finite by assumption.
	This completes the proof.
\end{proof}

\subsection{Verifying \Cref{asmp:sc_cnd,asmp:an_cnd,asmp:pc_cnd}} \label{apx: expfam_asmp}

In this appendix we demonstrate how \Cref{asmp:sc_cnd,asmp:an_cnd,asmp:pc_cnd} can be verified for the exponential family model when the Langevin Stein operator is employed.
For simplicity, consider the case where the data dimension is $d = 1$, the parameter dimension is $p = 1$, and the conjugate prior $\pi(\theta) \propto \exp(- \theta^2 / 2)$ is used.
From \eqref{eq:expfam}, a canonical exponential family model with $\eta(\theta) = \theta$ and $\X = \R$ is given by
\begin{align*}
	p_\theta(x) = \exp( \theta \cdot t(x) - a(\theta) + b(x) )
\end{align*} 
where $t: \R \to \R$, $a: \Theta \to \R$ and $b: \R \to \R$. 
Accordingly, the log derivative is given by $\nabla \log p_\theta(x) = \nabla t(x) \theta + \nabla b(x)$.
Identical calculations to \Cref{prop:post_expfam} 
show that the KSD of the exponential family model with the Langevin Stein operator takes a quadratic form
\begin{align*}
\operatorname{KSD}^2(\P_{\theta} \| \P_n) = C_{1,n} \theta^2 + C_{2,n} \theta + C_{3,n} \quad \text{ and } \quad \operatorname{KSD}^2(\P_{\theta} \| \P) = C_{1} \theta^2 + C_{2} \theta + C_{3} .
\end{align*}
where $C_{i,n} = (1 / n^2)\sum_{i,j=1}^{n} c_i(x_i, x_j)$ and $C_{i} = \E_{X,X' \sim \P}[ c_i(X, X') ]$ and
\begin{align*}
c_1(x, x') & := \nabla t(x) \cdot K(x, x') \nabla t(x') \\
c_2(x, x') & := \nabla t(x) \cdot \big( \nabla_{x'} \cdot K(x, x') \big) + \nabla t(x') \cdot \big( \nabla_{x} \cdot K(x, x') \big) + 2 \nabla t(x) \cdot K(x, x') \nabla b(x') \\
c_3(x, x') & := b(x) \cdot K(x, x') b(x') + \nabla_{x} \cdot ( \nabla_{x'} \cdot K(x, x') ) \\
& \hspace{160pt} + b(x) \cdot ( \nabla_{x'} \cdot K(x, x') ) + b(x') \cdot ( \nabla_{x} \cdot K(x, x') ) .
\end{align*}
Note that $C_{1,n} > 0$ and $C_1 > 0$ if a positive definite kernel $K$ is used.

\vspace{5pt}
\noindent \textbf{Verifying \Cref{asmp:an_cnd} ($r_{\max}=3$):}
First, note that $H_* = \nabla_{\theta}^2 \operatorname{KSD}^2(\P_{\theta} \| \P) |_{\theta = \theta_*}$ is non-singular since $\nabla_{\theta} \operatorname{KSD}^2(\P_{\theta} \| \P) = \nabla_{\theta}^2 (C_1 \theta^2 + C_2 \theta + C_3) = 2 C_1 > 0$.
Now, as demonstrated in \Cref{sec:prop_KSD}, when $\S_{\P_\theta}$ is the Langevin Stein operator, we have $( \partial^r \S_{\P_\theta} )[h](x) = ( \partial^r \nabla_{x} \log p_{\theta}(x) ) \cdot h(x)$ and $h \mapsto ( \partial^r \S_{\P_\theta} )[h](x)$ is a continuous linear functional on $\H$ for each fixed $x \in \X$.
In the exponential family case, the map $\theta \mapsto \nabla_{x} \log p_{\theta}(x)$ is infinitely differentiable over $\B$ since it is polynomial, leading to
\begin{align*}
	\partial^1 \nabla_{x} \log p_{\theta}(x) = \nabla t(x), \qquad \partial^2 \nabla_{x} \log p_{\theta}(x) = 0, \qquad \partial^3 \nabla_{x} \log p_{\theta}(x) = 0 .
\end{align*}
It is then clear that $\E_{X \sim \P}[ \sup_{\theta \in \B} ( ( \partial^r \S_{\P_\theta} ) ( \partial^r \S_{\P_\theta} ) K(X, X) ) ] < \infty$ for $r = 2, 3$.
For $r = 1$, 
\begin{align}
	\E_{X \sim \P}\left[ \sup_{\theta \in \B} \big( ( \partial^1 \S_{\P_\theta} ) ( \partial^1 \S_{\P_\theta} ) K(X, X) \big) \right] = \E_{X \sim \P}\left[ \nabla t(X) \cdot K(X, X) \nabla t(X) \right] . \label{eq: b8_eq1}
\end{align}
For the remaining term in \Cref{asmp:an_cnd}, by essentially same calculations as \Cref{prop:post_expfam},
\begin{align}
	\E_{X \sim \P}\left[ \sup_{\theta \in \B} \big( \S_{\P_\theta} \S_{\P_\theta} K(X, X) \big) \right] & = \E_{X \sim \P}\left[ \sup_{\theta \in \B} \big( c_1(X, X) \theta^2 + c_2(X, X) \theta + c_3(X, X) \big) \right] \nonumber \\
	& \le \E_{X \sim \P}\left[ c_1(X, X) \sup_{\theta \in \B} \theta^2 + c_2(X, X) \sup_{\theta \in \B} \theta + c_3(X, X) \right] \label{eq: b8_eq2}
\end{align}
Since $\Theta$ is a bounded set in $\R$, it is clear that $\sup_{\theta \in \B} \theta < \infty$.
The finiteness of \eqref{eq: b8_eq1} and \eqref{eq: b8_eq2} can therefore be interpreted as finite moment conditions involving $t$, $b$, $K$ and $\P$.

\vspace{5pt}
\noindent \textbf{Verifying \Cref{asmp:sc_cnd}:}
If both $\operatorname{KSD}^2(\P_{\theta} \| \P_n)$ and $\operatorname{KSD}^2(\P_{\theta} \| \P)$ are of quadratic form with $C_{1,n} > 0$ and $C_1 > 0$, the estimator $\theta_n$ exists and the minimiser $\theta_*$ is unique over $\R$.
It depends on $C_1, C_2, C_3$ whether $\theta_*$ is contained in $\Theta$, but if we are free to select $\Theta$ then we may select it such that $\theta^* \in \Theta$.
Since $C_1 > 0$, the \emph{well-separated} property of $\theta_*$ is automatically satisfied; i.e. $\operatorname{KSD}(\P_{\theta_*}, \P) < \inf_{\{\theta \in \Theta : \|\theta - \theta_*\|_2 \ge \epsilon \}} \operatorname{KSD}(\P_\theta,\P)$ for all $\epsilon > 0$.

\vspace{5pt}
\noindent \textbf{Verifying \Cref{asmp:pc_cnd}:}
Part (1) is immediately satisfied since the prior density $\pi(\theta) \propto \exp(- \theta^2 / 2)$ is continuous and positive on $\Theta$.
For part (2), we first have
\begin{align*}
	| \operatorname{KSD}^2(\P_{\theta} \| \P) - \operatorname{KSD}^2(\P_{\theta_*} \| \P) | = | C_1 \theta^2 + C_2 \theta - C_1 \theta_*^2 - C_2 \theta_* | = C_1 \left| ( \theta + Z_2 )^2 - Z_1 \right|
\end{align*}
where $Z_1 := C_2^2 / (4 C_1^2) + \theta_*^2 + (C_2 / C_1) \theta_* $ and $Z_2 := C_2 / (2 C_1)$ by completing the square.
By the simple calculation, the set $B_n(\alpha_1) = \{ \theta \in \Theta : | \operatorname{KSD}^2(\P_{\theta} \| \P) - \operatorname{KSD}^2(\P_{\theta_*} \| \P) | \le \alpha_1 / \sqrt{n} \}$ is then given by
\begin{align*}
	B_n(\alpha_1) = \left\{ \theta \in \Theta : - \left( \frac{\alpha_1}{C_1 \sqrt{n}} + Z_1 \right)^{\frac{1}{2}} - Z_2 \le \theta \le \left( \frac{\alpha_1}{C_1 \sqrt{n}} + Z_1 \right)^{\frac{1}{2}} - Z_2 \right\}
\end{align*}
While it is difficult to derive an explicit inequality between $\Pi(B_n)$ and $\exp(- \alpha_2 \sqrt{n})$, since it requires division into cases according to the values of $C_1,C_2,C_3$, $\theta_*$, and the set $\Theta$, the explicit form of $B_n$ renders it straightforward to numerically determine which values for $\alpha_1 > 0$ and $\alpha_2 > 0$ ensure that $\Pi(B_n) \ge \exp(- \alpha_2 \sqrt{n})$ holds for all $n \in \mathbb{N}$.

\vspace{5pt}
\noindent \textbf{Quantities $S_n(x, \theta)$ and $J_n$:}
Here we provide the explicit form of $S_n(x, \theta)$ and $J_n$ used to determine the value of $\beta$ for exponential family model.
From the definition,
\begin{align*}
S_n(x, \theta) & = \frac{1}{n} \sum_{i = 1}^{n} \nabla_{\theta}\big( \S_{\P_\theta} \S_{\P_\theta} K(x, x_i) \big) = 2 \left( \frac{1}{n} \sum_{i = 1}^{n} c_1(x, x_i) \right) \theta + \left( \frac{1}{n} \sum_{i = 1}^{n} c_2(x, x_i) \right) .
\end{align*}
Let $c_{1,n}(x) := (1 / n) \sum_{i = 1}^{n} c_1(x, x_i)$ and $c_{2,n}(x) := (1 / n) \sum_{i = 1}^{n} c_2(x, x_i)$. 
From the definition,
\begin{align*}
J_n & = \frac{1}{n} \sum_{i = 1}^{n} S_n(x_i, \theta_n) S_n(x_i, \theta_n)^\top = \frac{1}{n} \sum_{i = 1}^{n} \big( 2 c_{1,n}(x_i) \theta_n + c_{2,n}(x_i) \big) \big( 2 c_{1,n}(x_i) \theta_n + c_{2,n}(x_i) \big)^\top .
\end{align*}
Together with $H_n = C_{1,n}$, the default choice of $\beta$ is given by \eqref{eq:beta_choice_approx} in \Cref{subsec: choice of k}.


\section{Auxiliary Theoretical Results} \label{sec: proof_pre_result_4}

In \Cref{sec:appendix} we exploited a number of auxiliary results, the details of which are now provided. 
Recall that Standing Assumptions 1 and 2 continue to hold throughout.

\subsection{Derivative Bounds} \label{sec: proof_pre_result_4_pre}

Our auxiliary results mainly concern moments of derivative quantities, and the aim of \Cref{sec: proof_pre_result_4_pre} is to establish the main bounds that will be used.
Recall that $\partial^1$, $\partial^2$ and $\partial^3$ denote the partial derivatives $(\partial / \partial \theta_h)$, $(\partial^2 / \partial \theta_h \partial \theta_k)$ and $(\partial^3 / \partial \theta_h \partial \theta_k \partial \theta_l)$ respectively.
For the proofs in \Cref{sec: proof_pre_result_4_pre}, we make the index explicit by re-writing them as $\partial_{(h)}^1$, $\partial_{(h,k)}^2$ and $\partial_{(h,k,l)}^3$.
For $x \in \X$ and $(h,k,l) \in \{ 1, \ldots, p \}^3$, we define
\begin{align*}
m^0(x) & := \sup_{\theta \in \B} \sqrt{ \S_{\P_{\theta}} \S_{\P_{\theta}} K(x, x) } , \qquad
m^1(x)  := \sup_{\theta \in \B} \sqrt{ \sum_{h=1}^{p} ( \partial_{(h)}^1 \S_{\P_{\theta}} ) ( \partial_{(h)}^1 \S_{\P_{\theta}} ) K(x, x) }, \\
m^2(x) & := \sup_{\theta \in \B} \sqrt{ \sum_{h,k=1}^{p} ( \partial_{(h,k)}^2 \S_{\P_{\theta}} ) ( \partial_{(h,k)}^2 \S_{\P_{\theta}} ) K(x, x) }, \\
m^3(x) & := \sup_{\theta \in \B} \sqrt{ \sum_{h,k,l=1}^{p} ( \partial_{(h,k,l)}^3 \S_{\P_{\theta}} ) ( \partial_{(h,k,l)}^3 \S_{\P_{\theta}} ) K(x, x) } .
\end{align*}
where we continue to use the convention that the first and second operator in expressions such as $( \partial_{(h)}^1 \S_{\P_{\theta}} ) ( \partial_{(h)}^1 \S_{\P_{\theta}} ) K(x, x')$ are respectively applied to the first and second argument of $K$.
Further define
\begin{align*}
M^1(x, x') & := m^1(x) m^0(x') + m^0(x) m^1(x') , \\
M^2(x, x') & := m^2(x) m^0(x') + 2 m^1(x) m^1(x') + m^0(x) m^2(x') , \\
M^3(x, x') & := m^3(x) m^0(x') + 3 m^2(x) m^1(x') + 3 m^1(x) m^2(x') + m^0(x) m^3(x') .
\end{align*}
Based on these quantities, we now provide three technical results, \Cref{lem:sk_deriv_upb}, \Cref{lem:sk_deriv_upb_conv} and \Cref{lem:sk_deriv_exp_bnd}.

\vspace{5pt}
\begin{lemma} \label{lem:sk_deriv_upb}
	Suppose \Cref{asmp:an_cnd} ($r_{max}=3$) holds.
	For each $r = 1, 2, 3$, and for any $x, x' \in \X$,
	\begin{align}
	\sup_{\theta \in \B} \big\| \nabla_{\theta}^r \big( \S_{\P_{\theta}} \S_{\P_\theta} K(x, x') \big) \big\|_2 & \le M^r(x, x') . \label{eq: deriv bd ap c}
	\end{align}
	If instead \Cref{asmp:an_cnd} ($r_{max} = 1$) holds, then \eqref{eq: deriv bd ap c} holds for $r = 1$.
\end{lemma}

\begin{proof}
	We first derive the upper bound for $r = 1$ and then apply the same argument for the remaining upper bound for $r = 2$ and $r = 3$.
	By the definition of $\nabla_{\theta}$,
	\begin{align}
	\sup_{\theta \in \B} \Big\| \nabla_{\theta} \big( \S_{\P_{\theta}} \S_{\P_\theta} K(x, x') \big) \Big\|_2 & = \sup_{\theta \in \B} \sqrt{ \sum_{h=1}^{p} \left( \partial_{(h)}^1 \big( \S_{\P_{\theta}} \S_{\P_\theta} K(x, x') \big) \right)^2 } . \label{eq:sk_deriv_norm}
	\end{align}
	By \Cref{lem:ub_kd2} and Standing Assumption 2, we have $\S_{\P_{\theta}} K(x, \cdot) \in \H$ for any $x \in \X$ and 
	\begin{align}
	(*_1) := \partial_{(h)}^1 \Big( \S_{\P_{\theta}} \S_{\P_\theta} K(x, x') \Big) & = \partial_{(h)}^1 \Big( \left\langle \S_{\P_{\theta}} K(x, \cdot), \S_{\P_{\theta}} K(x', \cdot) \right\rangle_{\H} \Big) . \label{eq:sk_deriv_inp}
	\end{align}
	From \Cref{asmp:an_cnd} ($r_{max}=1$), the operator $( \partial_{(h)}^1 \S_{\P_{\theta}} )$ exists over $\B$ and satisfies the preconditions of \Cref{lem:ub_kd2}.
	Hence, by setting $\S_\Q = ( \partial_{(h)}^1 \S_{\P_{\theta}} )$ in \Cref{lem:ub_kd2}, we have that $( \partial_{(h)}^1 \S_{\P_{\theta}} ) K(x, \cdot) \in \H$ for each $x \in \X$.
	Let $f_\theta(\cdot) = \S_{\P_{\theta}} K(x, \cdot)$ and $g_\theta(\cdot) = \S_{\P_{\theta}} K(x', \cdot)$.
	Then the following product rule holds:
	\begin{align}
	\partial_{(h)}^1 \langle f_{\theta}, g_{\theta} \rangle_{\H} & = \langle \partial_{(h)}^1 f_{\theta}, g_{\theta} \rangle_{\H} + \langle f_{\theta}, \partial_{(h)}^1 g_{\theta} \rangle_{\H} , \label{eq:formula_deriv_inp}
	\end{align}
	which is verified from definition of differentiation as a limit and continuity of the inner product.
	Note that $\partial_{(h)} f_\theta(\cdot) = ( \partial_{(h)}^1 \S_{\P_{\theta}} ) K(x, \cdot) \in \H$ and $\partial_{(h)} g_\theta(\cdot) = ( \partial_{(h)}^1 \S_{\P_{\theta}} ) K(x', \cdot) \in \H$.
	Therefore by \eqref{eq:formula_deriv_inp} and the Cauchy--Schwarz inequality,
	\begin{align*}
	(*_1) & = \left\langle \partial_{(h)}^1 \S_{\P_{\theta}} K(x, \cdot), \S_{\P_{\theta}} K(x', \cdot) \right\rangle_{\H} + \left\langle \S_{\P_{\theta}} K(x, \cdot), \partial_{(h)}^1 \S_{\P_{\theta}} K(x', \cdot) \right\rangle_{\H} \\
	& \le \underbrace{ \left\| ( \partial_{(h)}^1 \S_{\P_{\theta}} ) K(x, \cdot) \right\|_\H }_{(*_a)} \underbrace{\vphantom{\left|\partial_{(h)}^1\right\|_\H} \left\| \S_{\P_{\theta}} K(x', \cdot) \right\|_\H }_{(*_b)} + \underbrace{\vphantom{\left|\partial_{(h)}^1\right\|_\H} \left\| \S_{\P_{\theta}} K(x, \cdot) \right\|_\H }_{(*_c)} \underbrace{ \left\| ( \partial_{(h)}^1 \S_{\P_{\theta}} ) K(x', \cdot) \right\|_\H }_{(*_d)} .
	\end{align*}
	For the original term \eqref{eq:sk_deriv_norm}, by the triangle inequality,
	\begin{align*}
	\sup_{\theta \in \B} \sqrt{ \sum_{h=1}^{p} (*_1)^2 } & \le \sup_{\theta \in \B} \sqrt{ \sum_{h=1}^{p} \Big( (*_a) (*_b) + (*_c) (*_d) \Big)^2 } \le \sup_{\theta \in \B} \sqrt{ \sum_{h=1}^{p} (*_a)^2 (*_b)^2 } + \sup_{\theta \in \B} \sqrt{ \sum_{h=1}^{p} (*_c)^2 (*_d)^2 } .
	\end{align*}
	For the term $(*_a)$, expanding the norm yields that
	\begin{align*}
	(*_a)^2 & = \left\langle ( \partial_{(h)}^1 \S_{\P_{\theta}} ) K(x, \cdot), ( \partial_{(h)}^1 \S_{\P_{\theta}} ) K(x, \cdot) \right\rangle_{\H} = ( \partial_{(h)}^1 \S_{\P_{\theta}} ) ( \partial_{(h)}^1 \S_{\P_{\theta}} ) K(x, x) .
	\end{align*}
	A similar argument applied to $(*_b)^2$, $(*_c)^2$ and $(*_d)^2$ leads to the overall bound
	\begin{align*}
	\sup_{\theta \in \B} \big\| \nabla_{\theta} \big( \S_{\P_{\theta}} \S_{\P_\theta} K(x, x') \big) \big\|_2 & \le m^1(x) m^0(x') + m^0(x) m^1(x') = M^1(x, x') .
	\end{align*}
	
	The upper bounds for $r = 2$ and $r = 3$ are obtained by an analogous argument.
	Indeed, from the definition of $\nabla_{\theta}^2$ and $\nabla_{\theta}^3$,
	\begin{align*}
	\sup_{\theta \in \B} \left\| \nabla_{\theta}^2 \big( \S_{\P_{\theta}} \S_{\P_\theta} K(x, x') \big) \right\|_2 & = \sup_{\theta \in \B} \sqrt{ \sum_{h,k=1}^{p} \left( \partial_{(h,k)}^2 \big( \S_{\P_{\theta}} \S_{\P_\theta} K(x, x') \big) \right)^2 } =: (*'') , \\
	\sup_{\theta \in \B} \left\| \nabla_{\theta}^3 \big( \S_{\P_{\theta}} \S_{\P_\theta} K(x, x') \big) \right\|_2 & = \sup_{\theta \in \B} \sqrt{ \sum_{h,k,l=1}^{p} \left( \partial_{(h,k,l)}^3 \big( \S_{\P_{\theta}} \S_{\P_\theta} K(x, x') \big) \right)^2 } =: (*''') .
	\end{align*}
	From \Cref{asmp:an_cnd} ($r_{max}=3$), the operators $( \partial_{(h,k)}^2 \S_{\P_{\theta}} )$ and $( \partial_{(h,k,l)}^3 \S_{\P_{\theta}} )$ exist over $\B$ and satisfy the preconditions of \Cref{lem:ub_kd2}.
	Hence from \Cref{lem:ub_kd2}, $\partial_{(h,k)}^2 f_\theta(\cdot) = ( \partial_{(h,k)}^2 \S_{\P_{\theta}} ) K(x, \cdot) \in \H$ and $\partial_{(h,k,l)}^3 f_\theta(\cdot) = ( \partial_{(h,k,l)}^3 \S_{\P_{\theta}} ) K(x, \cdot) \in \H$ for any $x \in \X$, and in turn $\partial_{(h,k)}^2 g_\theta(\cdot) \in \H$ and $\partial_{(h,k,l)}^3 g_\theta(\cdot) \in \H$.
	Repeated application of the product rule \eqref{eq:formula_deriv_inp} gives that
	\begin{align*}
	\partial_{(h,k)}^2 \langle f_{\theta}, g_{\theta} \rangle_{\H} = & \langle \partial_{(h,k)}^2 f_{\theta}, g_{\theta} \rangle_{\H} + \langle \partial_{(h)}^1 f_{\theta}, \partial_{(k)}^1 g_{\theta} \rangle_{\H} + \langle \partial_{(k)}^1 f_{\theta}, \partial_{(h)}^1 g_{\theta} \rangle_{\H} + \langle f_{\theta}, \partial_{(h,k)}^2 g_{\theta} \rangle_{\H} ,\\
	\partial_{(h,k,l)}^3 \langle f_{\theta}, g_{\theta} \rangle_{\H} = & \langle \partial_{(h,k,l)}^3 f_{\theta}, g_{\theta} \rangle_{\H} + \langle \partial_{(h,k)}^2 f_{\theta}, \partial_{(l)}^1 g_{\theta} \rangle_{\H} + \langle \partial_{(h,l)}^2 f_{\theta}, \partial_{(k)}^1 g_{\theta} \rangle_{\H} + \langle \partial_{(k,l)}^2 f_{\theta}, \partial_{(h)} g_{\theta} \rangle_{\H}  \nonumber \\
	& + \langle \partial_{(h)}^1 f_{\theta}, \partial_{(k,l)}^2 g_{\theta} \rangle_{\H} + \langle \partial_{(k)}^1 f_{\theta}, \partial_{(h,l)}^2 g_{\theta} \rangle_{\H} + \langle \partial_{(l)}^1 f_{\theta}, \partial_{(h,k)}^2 g_{\theta} \rangle_{\H} + \langle f_{\theta}, \partial_{(h,k,l)}^3 g_{\theta} \rangle_{\H} . 
	\end{align*}
	Following the same argument as the preceding upper bound for $r = 1$, the triangle inequality and Cauchy--Schwarz imply that
	\begin{align*}
	(*'') & \le m^2(x) m^0(x') + m^1(x) m^1(x') + m^1(x) m^1(x') + m^0(x) m^2(x') \\
	& = m^2(x) m^0(x') + 2 m^1(x) m^1(x') + m^0(x) m^2(x') = M^2(x, x') , \\
	(*''') & \le m^3(x) m^0(x') + m^2(x) m^1(x') + m^2(x) m^1(x') + m^2(x) m^1(x') \\
	& \qquad  + m^1(x) m^2(x') + m^1(x) m^2(x') + m^1(x) m^2(x') + m^0(x) m^3(x') \\
	& = m^3(x) m^0(x') + 3 m^2(x) m^1(x') + 3 m^1(x) m^2(x') + m^0(x) m^3(x') = M^3(x, x') ,
	\end{align*}
	which are the claimed upper bounds for the cases $r = 2$ and $r = 3$.
\end{proof}

\vspace{5pt}
\begin{lemma} \label{lem:sk_deriv_exp_bnd}
	Suppose \Cref{asmp:an_cnd} ($r_{max}=3$) holds.
	For $r = 0, 1, 2, 3$, $\E_{X \sim \P}[ | m^r(X) | ] < \infty$ and $\E_{X \sim \P}[ | m^r(X) |^2 ] < \infty$. 
	For $r = 1, 2, 3$, $\E_{X, X' \sim \P}[ | M^r(X, X') | ] < \infty$ and $\E_{X \sim \P}[ | M^r(X, X) | ] < \infty$.
	If instead \Cref{asmp:an_cnd} ($r_{max}=1$) holds, these results hold for $0 \leq r \leq 1$.
\end{lemma}

\begin{proof}
	First, note that positivity of $m^r(\cdot)$ and $M^r(\cdot)$ implies that the absolute value signs can be neglected.
	Moreover, from Jensen's inequality $( \E_{X \sim \P}[ m^r(X) ] )^2 \le \E_{X \sim \P}[ m^r(X)^2 ]$. 
	Thus it is sufficient to show that (a) $\E_{X \sim \P}[ m^r(X)^2 ] < \infty$, (b) $\E_{X, X' \sim \P}[ M^r(X, X') ] < \infty$ and (c) $\E_{X \sim \P}[ M^r(X, X) ] < \infty$.
	
	\vspace{5pt}
	\noindent \textbf{Part (a):}
	The argument is analogous for each $r = 0,1,2,3$ and we present it with $r = 3$.
	The bounded follows from Jensen's inequality and the triangle inequality:
	\begin{align*}
	\E_{X \sim \P}\left[ m^3(X)^2 \right] & \le \E_{X \sim \P}\left[ \sup_{\theta \in \B} \sum_{h,k,l = 1}^{p} ( \partial_{(h,k,l)}^3 \S_{\P_{\theta}} ) ( \partial_{(h,k,l)}^3 \S_{\P_{\theta}} ) K(X, X) \right] \\
	& \le \sum_{h,k,l = 1}^{p} \E_{X \sim \P}\left[ \sup_{\theta \in \B} \big( ( \partial_{(h,k,l)}^3 \S_{\P_{\theta}} ) ( \partial_{(h,k,l)}^3 \S_{\P_{\theta}} ) K(X, X) \big) \right]
	\end{align*}
	where the terms in the sum are finite by \Cref{asmp:an_cnd} ($r_{max}=3$).
	
	\vspace{5pt}
	\noindent \textbf{Part (b):}
	Since $X, X'$ are independent in the expectation $\E_{X, X' \sim \P}[ M^r(X, X') ]$, it is clear from the definition of $M^r$ that $\E_{X, X' \sim \P}[ M^r(X, X') ]$ exists if the expectation of each term $m^s(X)$, $s \leq r$, exists.
	Thus by part (a), $\E_{X, X' \sim \P}[ M^r(X, X') ] < \infty$ for $r = 1, 2, 3$.
	
	\vspace{5pt}
	\noindent \textbf{Part (c):}
	From the definition of $M^r(x, x)$ for $r = 1, 2, 3$,
	\begin{align*}
	\E_{X \sim \P}[ M^1(X, X) ] & = 2 \E_{X \sim \P}[ m^1(X) m^0(X) ] , \\
	\E_{X \sim \P}[ M^2(X, X) ] & = 2 \E_{X \sim \P}[ m^2(X) m^0(X) ] + 2 \E_{X \sim \P}[ m^1(X) m^1(X) ] , \\
	\E_{X \sim \P}[ M^3(X, X) ] & = 2 \E_{X \sim \P}[ m^3(X) m^0(X) ] + 6 \E_{X \sim \P}[ m^2(X) m^1(X) ] .
	\end{align*}
	Applying the Cauchy Schwartz inequality for each term
	\begin{align*}
	\E_{X \sim \P}[ M^1(X, X) ] & \le 2 \sqrt{ \E_{X \sim \P}[ m^1(X)^2 ] } \sqrt{ \E_{X \sim \P}[ m^0(X)^2 ] } , \\
	\E_{X \sim \P}[ M^2(X, X) ] & \le 2 \sqrt{ \E_{X \sim \P}[ m^2(X)^2 ] } \sqrt{ \E_{X \sim \P}[ m^0(X)^2 ] } + 2 \sqrt{ \E_{X \sim \P}[ m^1(X)^2 ] } \sqrt{ \E_{X \sim \P}[ m^1(X)^2 ] } , \\
	\E_{X \sim \P}[ M^3(X, X) ] & \le 2 \sqrt{ \E_{X \sim \P}[ m^3(X)^2 ] } \sqrt{ \E_{X \sim \P}[ m^0(X)^2 ] } + 6 \sqrt{ \E_{X \sim \P}[ m^2(X)^2 ] } \sqrt{ \E_{X \sim \P}[ m^1(X)^2 ] } .
	\end{align*}
	Since each of the latter expectations is finite by part (a), $\E_{X \sim \P}[ M^r(X, X) ] < \infty$ for $r = 1, 2, 3$. 
	
	\vspace{5pt}
	Inspection of the proof reveals that these results hold for $r = 0,1$ if instead \Cref{asmp:an_cnd} ($r_{max}=1$) holds.
\end{proof}

\vspace{5pt}
\begin{lemma} \label{lem:sk_deriv_upb_conv}
	Suppose \Cref{asmp:an_cnd} ($r_{max}=3$) holds.
	Then, for $r = 1, 2, 3$,
	\begin{align}
		\frac{1}{n^2} \sum_{i=1}^{n} \sum_{j=1}^{n} M^r(x_i, x_j) \overset{a.s}{\longrightarrow} \E_{X, X' \sim \P}[ M^r(X, X') ] < \infty . \label{eq: deriv bd sum b}
	\end{align}
	If instead \Cref{asmp:an_cnd} ($r_{max}=1$) holds, then \eqref{eq: deriv bd sum b} holds for $r = 1$.
\end{lemma}

\begin{proof}
	The proof is based on the strong law of large numbers,
	the sufficient conditions for which are provided by \Cref{lem:sk_deriv_exp_bnd}, which shows that $\E_{X \sim \P}\left[ | m^r(X) | \right] < \infty$ for $r = 0, 1, 2, 3$ under \Cref{asmp:an_cnd} ($r_{max}=3$).
	Then the strong law of large numbers \cite[Theorem 2.5.10]{Durrett2010b} yields that $(1/n) \sum_{i = 1}^{n} m^r(x_i) \overset{a.s.}{\to} \E_{X \sim \P}\left[ m^r(X) \right] =: (*_r)$ for $r = 0, 1, 2, 3$.
	Then, from the definition of $M^1$,
	\begin{align*}
		& \lim_{n \to \infty} \frac{1}{n^2} \sum_{i=1}^{n} \sum_{j=1}^{n} M^1(x_i, x_j) = \lim_{n \to \infty} \frac{1}{n^2} \sum_{i=1}^{n} \sum_{j=1}^{n} \Big( m^1(x_i) m^0(x_j) + m^0(x_i) m^1(x_j) \Big) \\
		& \hspace{5pt} = \lim_{n \to \infty} \frac{1}{n} \sum_{i=1}^{n} m^1(x_i) \times \lim_{n \to \infty} \frac{1}{n} \sum_{j=1}^{n} m^0(x_j) + \lim_{n \to \infty} \frac{1}{n} \sum_{i=1}^{n} m^0(x_i) \times \lim_{n \to \infty} \frac{1}{n} \sum_{j=1}^{n} m^1(x_j) .
	\end{align*}
	Since each limit in the right hand side converges a.s. to either $(*_0)$ or $(*_1)$, so that 
	\begin{align*}
	\frac{1}{n^2} \sum_{i=1}^{n} \sum_{j=1}^{n} M^1(x_i, x_j) \overset{a.s.}{\longrightarrow} & \E_{X \sim \P}[ m^1(X) ] \times \E_{X \sim \P}[ m^0(X) ] + \E_{X \sim \P}[ m^0(X) ] \times \E_{X \sim \P}[ m^1(X) ] \\
	= & \E_{X, X' \sim \P}[ m^1(X) m^0(X') + m^0(X) m^1(X') ] = \E_{X, X' \sim \P}[ M^1(X, X') ] ,
	\end{align*}
	where $X, X'$ are independent.
	An analogous argument holds for $M^2(x_i, x_j)$ and $M^3(x_i, x_j)$, giving that
	\begin{align*}
	\frac{1}{n^2} \sum_{i=1}^{n} \sum_{j=1}^{n} M^2(x_i, x_j) & \overset{a.s.}{\longrightarrow} (*_2) (*_0) + 2 (*_1) (*_1) + (*_0) (*_2) = \E_{X, X' \sim \P}[ M^2(X, X') ] , \\
	\frac{1}{n^2} \sum_{i=1}^{n} \sum_{j=1}^{n} M^3(x_i, x_j) & \overset{a.s.}{\longrightarrow} (*_3) (*_0) + 3 (*_2) (*_1) + 3 (*_1) (*_2) + (*_0) (*_3) = \E_{X, X' \sim \P}[ M^3(X, X') ] .
	\end{align*}
	Inspection of the proof reveals that \eqref{eq: deriv bd sum b} still holds for $r = 1$ if \Cref{asmp:an_cnd} ($r_{max}=1$) holds instead.
\end{proof}

\subsection{Proof of Auxiliary Results} \label{sec:proof_deriv_KSD_2nd3rd_up}

Throughout this section we let $f_n(\theta) := \operatorname{KSD}^2(\P_{\theta} \| \P_n)$ and $f(\theta) := \operatorname{KSD}^2(\P_{\theta} \| \P)$.
Similarly to $\nabla_{\theta}^2$, we let $\nabla_{\theta}^3 := \nabla_{\theta} \otimes \nabla_{\theta} \otimes \nabla_{\theta}$ denote the tensor product $\otimes$ where each component is given by $\partial_{h,k,l}^3$.
For a matrix $a \in \R^{p \times p}$ and tensor $b \in \R^{p \times p \times p}$, denote their Euclidean norms by $\| a \|_2$ and $\| b \|_2$.

\begin{lemma}[Derivatives a.s. Bounded] \label{lem:deriv_KSD_2nd3rd_up}
	Suppose \Cref{asmp:an_cnd} ($r_{max}=3$) holds.
	Then $\limsup_{n \to \infty} \sup_{\theta \in \B} \left\| \nabla_{\theta}^r f_n(\theta) \right\|_2 < \infty$ a.s. for $r = 1, 2, 3$.
	If instead \Cref{asmp:an_cnd} ($r_{max}=1$) holds, then the result holds for $r = 1$.
\end{lemma}

\begin{proof}
	First of all, for finite $n$ we have
	\begin{align*}
	\nabla_{\theta}^r f_n(\theta) = \nabla_{\theta}^r \frac{1}{n^2} \sum_{i = 1}^{n} \sum_{j = 1}^{n} \S_{\P_{\theta}} \S_{\P_\theta} K(x_i, x_j) = \frac{1}{n^2} \sum_{i = 1}^{n} \sum_{j = 1}^{n} \nabla_{\theta}^r \big( \S_{\P_{\theta}} \S_{\P_\theta} K(x_i, x_j) \big) .
	\end{align*}
	From the triangle inequality and \Cref{lem:sk_deriv_upb}, we further have
	\begin{align*}
	\sup_{\theta \in \B} \| \nabla_{\theta}^r f_n(\theta) \|_2 = \frac{1}{n^2} \sum_{i = 1}^{n} \sum_{j = 1}^{n} \sup_{\theta \in \B} \left\| \nabla_{\theta}^r \big( \S_{\P_{\theta}} \S_{\P_\theta} K(x_i, x_j) \big) \right\|_2 \le \frac{1}{n^2} \sum_{i = 1}^{n} \sum_{j = 1}^{n} M^r(x_i, x_j) .
	\end{align*}
	It follows from \Cref{lem:sk_deriv_upb_conv} that $(1 / n^2) \sum_{i = 1}^{n} \sum_{j = 1}^{n} M^r(x_i, x_j) \overset{a.s.}{\longrightarrow} \E_{X, X' \sim \P}[ M^r(X, X') ] < \infty$.
	Therefore, a.s. $\limsup_{n \to \infty} \sup_{\theta \in \B} \left\| \nabla_{\theta}^r f_n(\theta) \right\|_2 < \infty$.
	Inspection of the proof reveals that the argument still holds for $r = 1$ if \Cref{asmp:an_cnd} ($r_{max}=1$) holds instead.
\end{proof}

\vspace{5pt}
\begin{lemma}[A.S. Convergence of Derivatives] \label{lem:deriv_KSD_1st2nd_con} 
	Suppose \Cref{asmp:an_cnd} ($r_{max}=3$) and \ref{asmp:sc_cnd} hold.
	Then we have $\nabla_{\theta}^r f_n(\theta_*) \overset{a.s.}{\to} \nabla_{\theta}^r f(\theta_*)$ for $r = 1, 2, 3$.
	Let $H_n := \nabla_{\theta}^2 f_n(\theta_n)$ and $H_* := \nabla_{\theta}^2 f(\theta_*)$.
	We further have $H_n \overset{a.s.}{\to} H_*$, where $H_n$ and $H_*$ are symmetric and $H_*$ is semi positive definite.
\end{lemma}

\begin{proof}
    The proof is structured as follows: First we show (a) $\nabla_{\theta}^r f_n(\theta_*) \overset{a.s.}{\to} \nabla_{\theta}^r f(\theta_*)$ for $r = 1, 2, 3$.
	Then we show (b) $H_n \overset{a.s.}{\to} H_*$.
	Finally we show (c) $H_n$ is symmetric and $H_*$ is semi-positive definite.
	
	\vspace{5pt}
	\noindent \textbf{Part (a):}
	The argument here is analogous to that used to prove \Cref{thm:pw}, based on the decomposition
	\begin{align*}
	\nabla_{\theta}^r f_n(\theta) = \nabla_{\theta}^r \frac{1}{n^2} \sum_{i = 1}^{n} \sum_{j = 1}^{n} \S_{\P_{\theta}} \S_{\P_\theta} K(x_i, x_j) = \frac{1}{n^2} \sum_{i = 1}^{n} \sum_{j = 1}^{n} \nabla_{\theta}^r \big( \S_{\P_{\theta}} \S_{\P_\theta} K(x_i, x_j) \big) .
	\end{align*}
	Let $F(x, x') := \nabla_{\theta}^r \big( \S_{\P_{\theta}} \S_{\P_\theta} K(x, x') \big)$ to see that 
	\begin{align*}
	\nabla_{\theta}^r f_n(\theta) & = \frac{1}{n} \underbrace{ \frac{1}{n} \sum_{i = 1}^{n} F(x_i, x_i) }_{ (*_1) } + \frac{n - 1}{n} \underbrace{ \frac{1}{n (n - 1)} \sum_{i = 1}^{n} \sum_{j \ne i}^{n} F(x_i, x_j) }_{ (*_2) } .
	\end{align*}
	It follows from the strong law of large number \cite[Theorem~2.5.10]{Durrett2010b} that $(*_1) \overset{a.s.}{\to} \E_{X \sim \P}[ F(X, X) ]$ provided $E_{X \sim \P}[ \| F(X, X) \|_2 ] < \infty$. 
	Similarly, it follows from the strong law of large number for U-statistics \citep{Hoeffding1961} that $(*_2) \overset{a.s.}{\to} \E_{X,X' \sim \P}[ F(X, X') ]$ provided $E_{X, X' \sim \P}[ \| F(X, X') \|_2 ] < \infty$. 
	Both the required conditions holds by \Cref{lem:sk_deriv_exp_bnd} and the fact that $\| F(x, x') \|_2 \le \sup_{\theta \in \Theta} \| \nabla_{\theta}^r ( \S_{\P_\theta} \S_{\P_\theta} K(x, x') ) \|_2 \le M^r(x, x')$ from \Cref{lem:sk_deriv_upb}.
	Thus
	\begin{align*}
	\nabla_{\theta}^r f_n(\theta) & \overset{a.s.}{\longrightarrow} \E_{X,X' \sim \P}[ F(X, X') ] = \E_{X,X' \sim \P}[ \nabla_{\theta}^r \big( \S_{\P_{\theta}} \S_{\P_\theta} K(x_i, x_j) \big) ] .
	\end{align*}
	Since $\E_{X, X' \sim \P}[ \| F(X, X') \|_2 ] < \infty$, we may apply the dominated convergence theorem to interchange expectation and differentiation:
	\begin{align*}
	\E_{X,X' \sim \P}[ \nabla_{\theta}^r ( \S_{\P_{\theta}} \S_{\P_\theta} K(X, X') ) ] = \nabla_{\theta}^r \E_{X,X' \sim \P}[ \S_{\P_{\theta}} \S_{\P_\theta} K(X, X') ] = \nabla_{\theta}^r f(\theta) .
	\end{align*}
	Therefore, setting $\theta = \theta_*$, we conclude that $\nabla_{\theta}^r f_n(\theta_*) \overset{a.s.}{\to} \nabla_{\theta}^r f(\theta_*)$.
	
	\vspace{5pt}
	\noindent \textbf{Part (b):} 
	First of all, by the triangle inequality,
	\begin{align*}
	\left\| \nabla_{\theta}^2 f_n(\theta_n) - \nabla_{\theta}^2 f(\theta_*) \right\|_2 \le \underbrace{ \left\| \nabla_{\theta}^2 f_n(\theta_n) - \nabla_{\theta}^2 f_n(\theta_*) \right\|_2 }_{(**_1)} + \underbrace{ \left\| \nabla_{\theta}^2 f_n(\theta_*) - \nabla_{\theta}^2 f(\theta_*) \right\|_2 }_{(**_2)} .
	\end{align*}
	By the mean value theorem applied to $(**_1)$ and \Cref{lem:deriv_KSD_2nd3rd_up} (i.e. $\lim_{n \to \infty} \sup_{\theta \in \B} \| \nabla_{\theta}^3 f_n(\theta) \|_2 < \infty$ a.s.), there a.s. exists a constant $0 < C < \infty$ s.t., for all sufficiently large $n$,
	\begin{align*}
	(**_1) = \left\| \nabla_{\theta}^2 f_n(\theta_n) - \nabla_{\theta}^2 f_n(\theta_*) \right\|_2 & \le \sup_{\theta \in \B} \| \nabla_{\theta}^3 f_n(\theta) \|_2 \| \theta_n - \theta_* \|_2 \le C \| \theta_n - \theta_* \|_2 .
	\end{align*}
	Then applying \Cref{lem:sc_ksd} (i.e. $\| \theta_n - \theta_* \|_2 \overset{a.s.}{\to} 0$), we have $(**_1) \overset{a.s.}{\to} 0$.
	Further the preceding part (a) implied that $(**_2) \overset{a.s.}{\to} 0$.
	Therefore, we conclude that $\nabla_{\theta}^2 f_n(\theta_n) \overset{a.s.}{\to} \nabla_{\theta}^2 f(\theta_*)$.
	 
	\vspace{5pt}
	\noindent \textbf{Part (c):} 
	Since $f_n$ is twice continuously differentiable over $\B$ by assumption, commutation of two partial derivatives $\partial_{(h)} \partial_{(k)} f_n(\theta) = \partial_{(k)} \partial_{(h)} f_n(\theta)$ holds over $\B$ by the Clairaut's theorem.
	Therefore the $(h,k)$-th entry and $(k,h)$-th entry of $H_n = \nabla_{\theta}^2 f_n(\theta_n)$ are equal. 
	An analogous argument applies to $H_* = \nabla_{\theta}^2 f(\theta_*)$, so that both $H_n$ and $H_*$ are symmetric.
	Furthermore, the Hessian $H_*$ is semi positive definite since $\theta_*$ is the minimiser of $f$ from \Cref{asmp:sc_cnd}.
\end{proof}

\vspace{5pt}
\begin{lemma}[Moment Condition for Asymptotic Normality] \label{lem:score_finite_fm} 
	Suppose that \Cref{asmp:an_cnd} ($r_{max}=3$) holds. 
	Let $F(x, x') := \nabla_{\theta} ( \S_{\P_{\theta}} \S_{\P_{\theta}} K(x, x') )$ for any fixed $\theta \in \Theta$.
	Then we have $\E_{X,X' \sim \P}\left[ \left\| F(X, X') \right\|_2^2 \right] < \infty$ and $\E_{X \sim \P}\left[ \left\| F(X, X) \right\|_2 \right] < \infty$.
\end{lemma}

\begin{proof}
	First of all, it follows from \Cref{lem:sk_deriv_upb} that for any $x, x' \in \X$,
	\begin{align*}
	\| F(x, x') \|_2 \le \sup_{\theta \in \B} \| \nabla_{\theta} \big( \S_{\P_{\theta}} \S_{\P_{\theta}} K(x, x') \big) \|_2 \le M^1(x, x') .
	\end{align*}
	Thus for the first moment we have $\E_{X \sim \P}\left[ \| F(X, X) \|_2 \right] \le \E_{X \sim \P}[ M^1(X, X) ] < \infty$ from \Cref{lem:sk_deriv_exp_bnd}.
	For the second moment, $\E_{X,X' \sim \P}\left[ \left\| F(X, X') \right\|_2^2 \right] \le \E_{X,X' \sim \P}\left[ M^1(X, X')^2 \right] =: (*)$.
	By definition,
	\begin{align*}
	(*) & = \E_{X,X' \sim \P}\left[ \big( m^1(X) m^0(X') + m^0(X) m^1(X') \big)^2 \right] = 4 \E_{X \sim \P}\left[ m^1(X)^2 \right] \E_{X \sim \P}\left[ m^0(X)^2 \right] .
	\end{align*}
	Each of these latter expectations is finite by \Cref{lem:sk_deriv_exp_bnd}, which completes the proof.
\end{proof}

\vspace{5pt}
\begin{theorem}[Concentration Inequality for KSD] \label{thm:1st_con_ksd2} 
	Let $\sigma(\theta) :=  {\E}_{X \sim \P}[ \S_{\P_{\theta}} \S_{\P_{\theta}} K(X, X) ]$.
	Then
	\begin{align*}
	\mathbb{P}\left( \left| f_n(\theta) - f(\theta) \right| \ge \delta \right) \le \frac{4 \sigma(\theta)}{\delta \sqrt{n}}, \quad \forall \theta \in \Theta ,
	\end{align*}
	where the probability is with respect to realisations of the dataset $\{ x_i \}_{i=1}^{n} \stackrel{i.i.d.}{\sim} \P$.
\end{theorem}

\begin{proof}
	Since $| a^2 - b^2 | = | (a + b) (a - b) | = ( a + b ) | a - b |$ for all $a, b \in [0,\infty)$, we have the bound
	\begin{align*}
	\underbrace{ \left| \operatorname{KSD}^2(\P_\theta \| \P_n) - \operatorname{KSD}^2(\P_\theta \| \P) \right| }_{=:(*)} = \underbrace{ \left( \operatorname{KSD}(\P_\theta \| \P_n) + \operatorname{KSD}(\P_\theta \| \P) \right) }_{=:(*_1)} \underbrace{ \left| \operatorname{KSD}(\P_\theta \| \P_n) - \operatorname{KSD}(\P_\theta \| \P) \right| }_{=:(*_2)}.
	\end{align*}
	In what follows we use $\E$ to denote an expectation with respect to the dataset $\{x_i\}_{i=1}^n \stackrel{i.i.d.}{\sim} \P$.
	Applying Markov's inequality followed by Cauchy--Schwarz, we have
	\begin{align}
	\mathbb{P}( (*) \ge \delta ) \le \frac{1}{\delta}  \E[(*)] = \frac{1}{\delta}  \E[ (*_1) (*_2) ] \le \frac{1}{\delta} \sqrt{  \E[ (*_1)^2 ] } \sqrt{  \E[  (*_2)^2 ] } . \label{eq:prf1_eq0}
	\end{align}
	To conclude the proof, we bound the two expectations one the right hand side.
	
	\vspace{5pt}
	\noindent \textbf{Bounding $\E[ (*_1)^2 ]$:} 
	From the fact that $(a + b)^2 \leq 2 (a^2 + b^2)$ for $a, b \in \R$,
	\begin{align*}
	\E[ (*_1)^2 ]  \le 2 \E\left[ \operatorname{KSD}^2(\P_\theta \| \P_n) + \operatorname{KSD}^2(\P_\theta \| \P) \right] = 2 \Big( \E\left[ \operatorname{KSD}^2(\P_\theta \| \P_n) \right] + \operatorname{KSD}^2(\P_\theta \| \P) \Big) .
	\end{align*}
	The preconditions of \Cref{lem:ub_kd2} holds due to Standing Assumption 2.
	Thus from \Cref{lem:ub_kd2} part (iii), together with Jensen's inequality, we have the two bounds $\operatorname{KSD}^2(\P_\theta \| \P_n) \le (1 / n) \sum_{i = 1}^{n} \S_{\P_{\theta}} \S_{\P_{\theta}} K(x_i, x_i)$ and $\operatorname{KSD}^2(\P_\theta \| \P) \le  {\E}_{X \sim \P}[ \S_{\P_{\theta}} \S_{\P_{\theta}} K(X, X) ]$.
	Plugging these into the previous inequality, and exploiting independence of $x_i$ and $x_j$ whenever $i \neq j$, we have
	\begin{align*}
	\E[ (*_1)^2 ] & \le 2 \left( \E\left[ \frac{1}{n} \sum_{i = 1}^{n} \S_{\P_{\theta}} \S_{\P_{\theta}} K(x_i, x_i) \right] +  {\E}_{X \sim \P}[ \S_{\P_{\theta}} \S_{\P_{\theta}} K(X, X) ] \right) \\
	& = 2 \Big( \E_{X \sim \P}[ \S_{\P_{\theta}} \S_{\P_{\theta}} K(X, X) ] + \E_{X \sim \P}[ \S_{\P_{\theta}} \S_{\P_{\theta}} K(X, X) ] \Big) = 4 \sigma(\theta) ,
	\end{align*}
	where existence of $\sigma(\theta)$ for all $\theta \in \Theta$ is ensured by Standing Assumption 2.

	\vspace{5pt}
	\noindent \textbf{Bounding $\E[ (*_2)^2 ]$:} From the fact $|\sup_x | f(x) | - \sup_y | g(y) | | \leq \sup_x |f(x) - g(x)|$ for functions $f$ and $g$, the term $(*_2)$ is upper bounded by
	\begin{align*}
	(*_2) & = \left| \sup_{\|h\|_\H \leq 1} \left| \frac{1}{n} \sum_{i=1}^{n} \S_{\P_\theta}[h](x_i) \right| - \sup_{\|h\|_\H \leq 1} \Bigg| \E_{X \sim \P}[ \S_{\P_\theta}[h](X) ] \Bigg| \right| \\
	& \le \sup_{\|h\|_\H \leq 1} \left | \frac{1}{n} \sum_{i=1}^{n} \S_{\P_\theta}[h](x_i) - \E_{X \sim \P}[ \S_{\P_\theta}[u](X) ] \right| = \sup_{f \in \F} \left| \frac{1}{n} \sum_{i=1}^{n} f(x_i) - \E_{X \sim \P}[ f(X) ] \right| . 
	\end{align*}
	where $\F := \{ \S_{\P_\theta}[h] \mid  \| h \|_{\H} \le 1 \}$. We can see from this expression that standard arguments in the context of Rademacher complexity theory can be applied.
	Noting that $| \cdot |^2$ is a convex function, Proposition 4.11 in \citet[][]{Wainwright2019} gives that
	\begin{align*}
	\E\left[ (*_2)^2 \right] & \le  \E\left[ \left( \sup_{f \in \F} \left| \frac{1}{n} \sum_{i=1}^{n} f(x_i) - \E_{X \sim \P}[ f(X) ] \right| \right)^2 \right] \le \E \E_\epsilon\left[ 2^2 \left( \sup_{f \in \F} \left| \frac{1}{n} \sum_{i=1}^{n} \epsilon_i f(x_i) \right| \right)^2 \right] 
	\end{align*}
	where $\{ \epsilon_i \}_{i=1}^{n}$ are independent random variables taking values in $\{-1, +1\}$ with equiprobability $1 / 2$ and $\E_{\epsilon}$ is the expectation over $\{ \epsilon_i \}_{i=1}^{n}$.
	From the essentially same derivation as \Cref{prop:derivation_KSD}, the following equality holds:
	\begin{multline*}
	\sup_{f \in \F} \left| \frac{1}{n} \sum_{i=1}^{n} \epsilon_i f(x_i) \right| = \sup_{\|h\|_\H \leq 1} \left| \frac{1}{n} \sum_{i=1}^{n} \epsilon_i \S_{\P_\theta}[h](x_i) \right| = \sup_{\|h\|_\H \leq 1} \left| \left\langle h, \frac{1}{n} \sum_{i=1}^{n} \epsilon_i \S_{\P_\theta} K(x_i, \cdot) \right\rangle_{\H} \right| \\
	= \left\| \frac{1}{n} \sum_{i=1}^{n} \epsilon_i \S_{\P_\theta} K(x_i, \cdot) \right\|_{\H} 
	= \sqrt{ \frac{1}{n^2} \sum_{i=1}^n \sum_{j=1}^n \epsilon_i \epsilon_j \S_{\P_{\theta}} \S_{\P_{\theta}} K(x_i, x_j) } .
	\end{multline*}
	Plugging this equality into the upper bound of $ \E\left[ (*_2)^2 \right]$, we have
	\begin{align*}
	\E\left[ (*_2)^2 \right] & \le 4 \E \E_\epsilon\left[ \frac{1}{n^2} \sum_{i=1}^n \sum_{j=1}^n \epsilon_i \epsilon_j \S_{\P_{\theta}} \S_{\P_{\theta}} K(x_i, x_j) \right] \\ 
	& = 4 \E \left[ \frac{1}{n^2} \sum_{i=1}^{n} \S_{\P_{\theta}} \S_{\P_{\theta}}[ K(X_i, X_i) ] \right] = \frac{4}{n}  {\E}_{X \sim \P}[ \S_{\P_{\theta}} \S_{\P_{\theta}} K(X, X) ] = \frac{4 \sigma(\theta)}{n}.
	\end{align*}
	
	\vspace{5pt}
	\noindent \textbf{Bounding $\E[ (*)^2 ]$:} 
	Returning to \eqref{eq:prf1_eq0}, we have the overall bound
	\begin{align*}
	\mathrm{P}( (*) \ge \delta ) \le \frac{\sqrt{4 \sigma(\theta)} \sqrt{4 \sigma(\theta)}}{\delta \sqrt{n}} \le \frac{4 \sigma(\theta)}{\delta \sqrt{n}}
	\end{align*}
	as claimed.
\end{proof}


\color{black}

\section{Additional Empirical Results} 
\label{sec: emp appendix}

This appendix contains additional empirical results referred to in the main text.
\Cref{subsec: sensitivity kernel choice} investigates the sensitivity of the generalised posterior to the choice of parameters employed in the kernel $K$.
\Cref{subsec: sampling disn beta} investigates the sampling distribution of $\beta$, controlling the scale of the generalised posterior, when estimated using the approach proposed in \Cref{subsec: beta setting}.
An extended discussion of the choice of weighting function, $M$, and the associated trade-off between statistical efficiency and robustness, is contained in \Cref{subsec: eff robust}.
A comparison of KSD-Bayes with other generalised Bayesian procedures developed for \textit{tractable} likelihood is presented in \Cref{subsec: gen bayes comparison}.
Finally, the use of KSD-Bayes in the context of discrete state spaces is demonstrated in \Cref{subsec: other-space}.

\subsection{Sensitivity to Kernel Parameters}
\label{subsec: sensitivity kernel choice}

\begin{figure}[t!]
\centering
\includegraphics[width = 0.8\textwidth,clip,trim = 3cm 9.8cm 3cm 9cm]{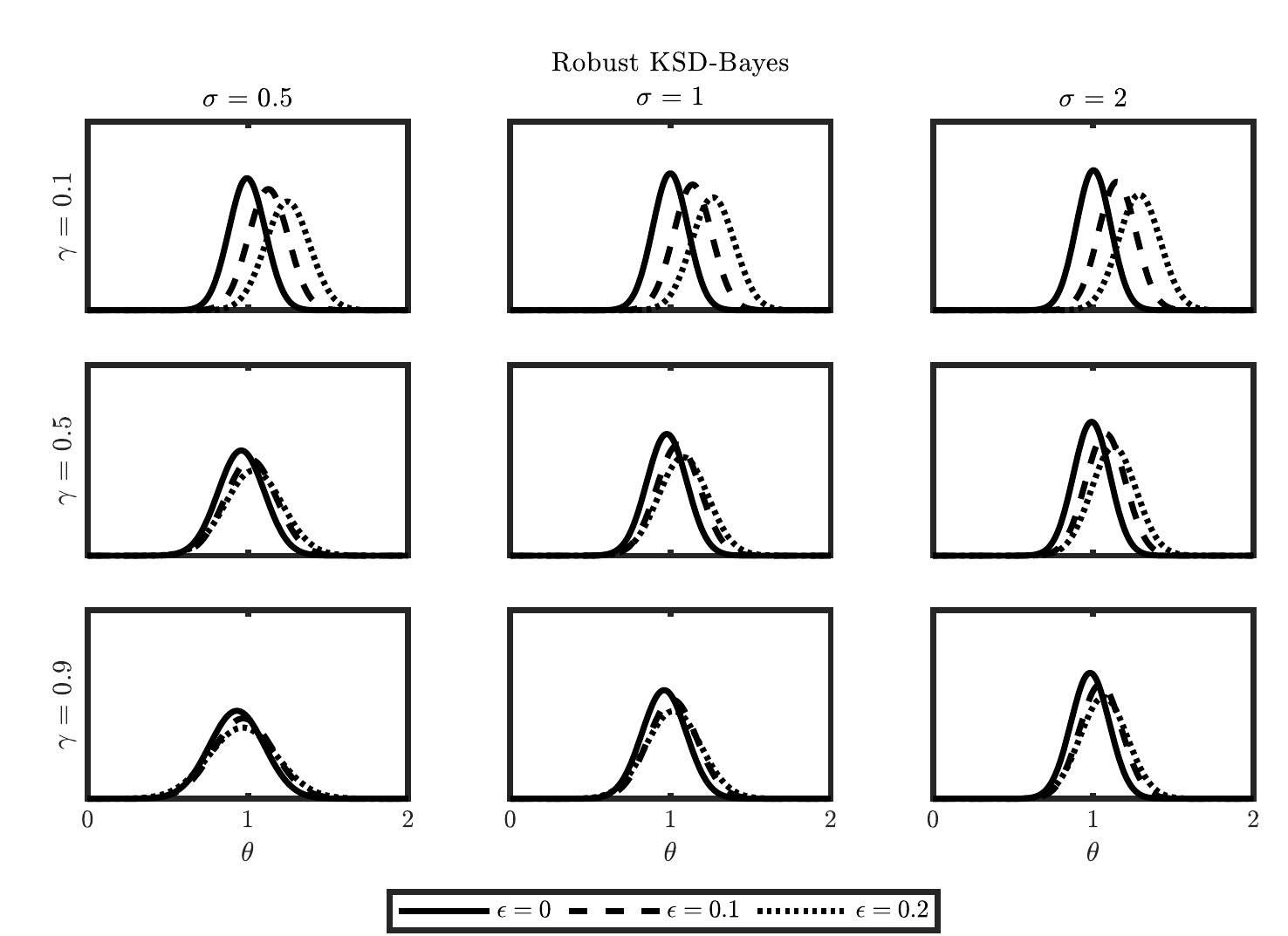}
\caption{\textcolor{black}{
Sensitivity to kernel parameters:
Kernels of the form \eqref{eq: param ker ap}, with length-scale parameter $\sigma$ and exponent $\gamma$, are considered in the context of the normal location model in \Cref{sec:experiment-nl}.
The settings $\sigma \approx 1$, $\gamma = 0.5$ (central panel) were used in the main text.
The true parameter value is $\theta = 1$, while a proportion $\epsilon$ of the data were contaminated by noise of the form $\mathcal{N}(y,1)$.
Here $y=10$ is fixed and $\epsilon \in \{0,0.1,0.2\}$ are considered.
}}
\label{fig: kernel sens Gauss}
\end{figure}

The kernel $K$ that we recommend as a default in \Cref{subsec: default SO and K} has no degrees of freedom to be specified (with the exception of the weighting function $M$, whose choice is further explored in \Cref{subsec: eff robust}).
Nevertheless, it is interesting to ask whether the generalised posterior is sensitive to our recommended choice of kernel.
To this end, we considered the family of kernels of the form
\begin{align}
K(x,x') = \left( 1 + \sigma^{-2} \|x-x'\|_2^2 \right)^{-\gamma} \times I_d \label{eq: param ker ap}
\end{align}
where $\sigma > 0$ and $\gamma \in (0,1)$.
Our recommended kernel sets $\sigma$ equal to a regularised version of the sample standard deviation of the dataset and $\gamma = 1/2$.
To investigate how the generalised KSD-Bayes posterior depends on the choice of $\sigma$ and $\gamma$, we re-ran the normal location model experiment from \Cref{sec:experiment-nl} using values $\sigma \in \{0.5,1,2\}$ and $\gamma \in \{0.1,0.5,0.9\}$.
To limit scope, we consider the performance of the robust version of KSD-Bayes from \Cref{sec:experiment-nl}, with weight function $M(x) = (1+x^2)^{-1/2}$, in the case where the contaminant is fixed to $y=10$ and the proportion of contamination is varied in $\epsilon \in \{0,0.1,0.2\}$.
Results in \Cref{fig: kernel sens Gauss} indicate that the generalised posterior is insensitive to $\sigma$, with almost identical output for each value of $\sigma$ considered.
The results for $\gamma \in \{0.5,0.9\}$ were almost identical, but the generalised posterior appeared to be less robust to contamination when $\gamma = 0.1$.
These results support the default choices recommended in the main text ($\sigma \approx 1$, $\gamma = 0.5$) and provide reassurance that the generalised posterior is not overly sensitive to how these values are specified.

\subsection{Sampling Distribution of $\beta$}
\label{subsec: sampling disn beta}

An important component of the KSD-Bayes method is the use of a data-adaptive $\beta$, as specified in \Cref{subsec: beta setting}.
In this appendix the sampling distribution of this data-adaptive $\beta$ is investigated.
Of particular interest are (1) the extent to which $\beta$ varies at small sample sizes, and (2) how the behaviour of $\beta$ changes when the data-generating model is mis-specified.
To investigate, we considered multiple independent realisations of the dataset in the context of the normal location model from \Cref{sec:experiment-nl}, collecting the corresponding estimates of $\beta$ together into box plots, so that the sampling distribution of $\beta$ can be visualised.
To limit scope, we consider the performance of the standard version of KSD-Bayes from \Cref{sec:experiment-nl} (i.e. with weight function $M(x) = 1$), in the case where the contaminant is fixed to $y=10$ and the proportion of contamination is varied in $\epsilon \in \{0,0.1,0.2\}$.
The dataset sizes $n \in \{10,50,100\}$ were considered.
Results in \Cref{fig: beta sampling} show that, in the case $\epsilon = 0$ where the model is well-specified, the value $\beta = 1$ is typically selected.
This value ensures that the scale of the KSD-Bayes posterior matches that of the standard posterior in this example, so that the approach used to select $\beta$ can be considered successful.
In the mis-specified regimes $\epsilon \in \{0.1,0.2\}$, with small $n$ the estimation of an appropriate weight $\beta$ is expected to be difficult and indeed the default choice of $\beta = 1$ in \eqref{eq:beta_choice_approx} is automatically adopted.
At larger values of $n$ it is possible to reliably estimate a weight $\beta < 1$ and this weight is seen to be smaller on average when data are more contaminated.
These results support our recommended approach to selecting $\beta$ in \eqref{eq:beta_choice_approx}.

\begin{figure}[t!]
\centering
\includegraphics[width = 0.8\textwidth,clip,trim = 1.5cm 11cm 2cm 12.9cm]{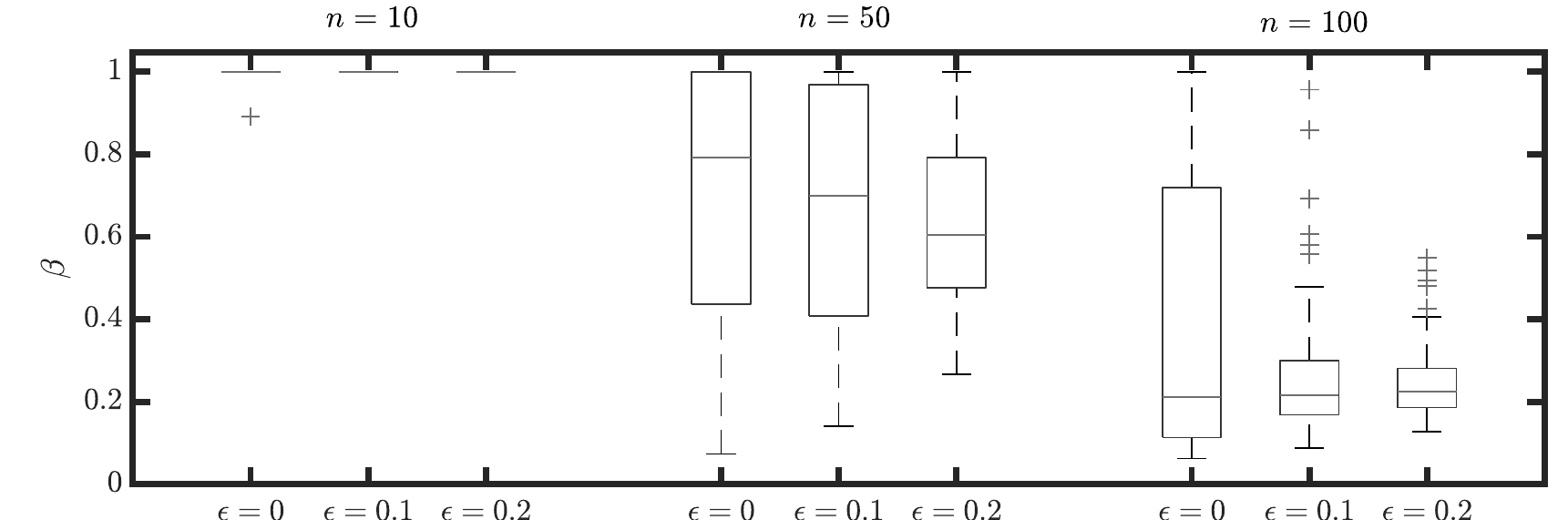}
\caption{\textcolor{black}{
Sampling distribution of $\beta$:
Box plots are used to summarise the sampling distribution of $\beta$ in the context of the normal location model in \Cref{sec:experiment-nl}.
The sample size $n$ and the contamination proportion $\epsilon$ were each varied.
}}
\label{fig: beta sampling}
\end{figure}

\subsection{Efficiency/Robustness Trade-Off}
\label{subsec: eff robust}

There is a well-known trade-off between statistical efficiency and robustness to model mis-specification, as exemplified by the data-agnostic statistician who is robust by not learning from data.
Minimum distance estimation, which can be considered the frequentist analogue of generalised Bayesian inference, can strike an attractive balance between these competing goals \citep[see e.g.][]{lindsay1994efficiency,basu2019statistical}.
In \Cref{sec:robustness} it was demonstrated that global bias-robustness can be achieved using KSD-Bayes through the inclusion of an appropriate weighting function $M$ in the kernel, and in \Cref{sec:experiment} it was demonstrated that KSD-Bayes can learn from data whilst being bias-robust.
However, it remains to investigate the extent to which statistical efficiency is lost in KSD-Bayes, compared to standard Bayesian inference, in the case where the data-generating model is correctly specified.
In this appendix we return to the normal location model of \Cref{sec:experiment-nl} and explore the effect of the choice of weighting function $M$ on the efficiency of the inferences that are produced.

Recall from \Cref{thm:bias-robust} that KSD-Bayes is globally bias-robust if there is a function $\gamma : \Theta \rightarrow \mathbb{R}$ such that
\begin{align}
	\sup_{y \in \R^d} \Big( \nabla_{y} \log p_{\theta}(y) \cdot K(y, y) \nabla_{y} \log p_{\theta}(y) \Big) \le \gamma(\theta) \label{eq: sup term}
\end{align}
where $\sup_{\theta \in \Theta} | \pi(\theta) \gamma(\theta) | < \infty$ and $\int_\Theta \pi(\theta) \gamma(\theta) \mathrm{d} \theta < \infty$.
For our recommended kernel $K$ in \eqref{eq: suggested kernel}, the expression on the left hand side of \eqref{eq: sup term} reduces to
$$
\sup_{y \in \mathbb{R}^d} \| M(y)^\top \nabla_y \log p_\theta(y) \|_2^2 .
$$
For the normal location model in \Cref{sec:experiment-nl} we have $\nabla_y \log p_\theta(y) = \theta - y$ and thus, with our recommended kernel from \Cref{eq: suggested kernel}, we have
\begin{align}
	\| M(y)^\top \nabla_y \log p_\theta(y) \|_2^2
	= (y - \theta)^2 M(y)^2 . \label{eq: robustness condition nl}
\end{align}
In order that \eqref{eq: robustness condition nl} is bounded over $y \in \mathbb{R}$ we require $M(y)$ to decay at the rate $\mathcal{O}(|y|^{-1})$ as $|y| \rightarrow \infty$.
This decay is achieved, for example, by functions of the form 
\begin{align}
M(y) = \left( \frac{a^2}{a^2+(y-b)^2} \right)^{c/2} \label{eq: more general m}
\end{align} 
for any $a \neq 0$, $b \in \mathbb{R}$ and any $c \geq 1$, although of course there are infinitely many other such functions that could be considered.
The particular value $c = 1$, which we considered in \Cref{sec:experiment-nl} of the main text and consider here in the sequel, represents the smallest value of $c$ for which \eqref{eq: robustness condition nl} is bounded over $y \in \R$.
For this choice we have that \eqref{eq: robustness condition nl} is maximised by $y = \theta \pm \sqrt{a^2 + (\theta - b)^2}$ and
\begin{align*}
	\sup_{y \in \R^d} (y - \theta)^2 M(y)^2 = \frac{[a^2 + (\theta - b)^2] a^2}{a^2 + [\theta - b \pm \sqrt{a^2 + (\theta - b)^2}]^2} \leq a^2 + (\theta - b)^2 =: \gamma(\theta) .
\end{align*}
For this bound $\gamma(\theta)$, all conditions of \Cref{thm:bias-robust} are satisfied.
The aim in what follows is to investigate how the performance of KSD-Bayes depends on the specific choices of $a$ and $b$ and in \eqref{eq: more general m}.

To limit scope, we consider performance in the case where the contaminant is fixed to $y=10$ and the proportion of contamination is varied in $\epsilon \in \{0,0.1,0.2\}$.
The dataset sizes was fixed at $n =100$ as per the main text.
Recall from \Cref{sec:experiment-nl} of the main text that the choices $a = 1$, $b=0$ lead to statistical efficiency comparable to that of standard Bayesian inference.
Results in \Cref{fig: eff robust trade} show that $a = 0.1$ led to almost total robustness to contamination at the expense of inefficient estimation, with the spread of the generalised posterior approximately twice as large as the case where $a = 1$.
The setting $a = 10$ causes the generalised posterior to approximate the non-robust KSD-Bayes approach with $M \equiv 1$, as would be expected from inspection of \eqref{eq: more general m}.
The generalised posterior was somewhat insensitive to $b$, though we note that the choice $b = -5$ conferred additional robustness at the expense of efficiency, while the choice $b = 5$ sacrificed both robustness and efficiency, in both cases relative to $b=0$.
These results broadly support the choices of $a = 1$ and $b = 0$ for this inference problem, as we considered in the main text.

\begin{figure}[t!]
\centering
\includegraphics[width = 0.8\textwidth,clip,trim = 1.5cm 9.8cm 1.5cm 9cm]{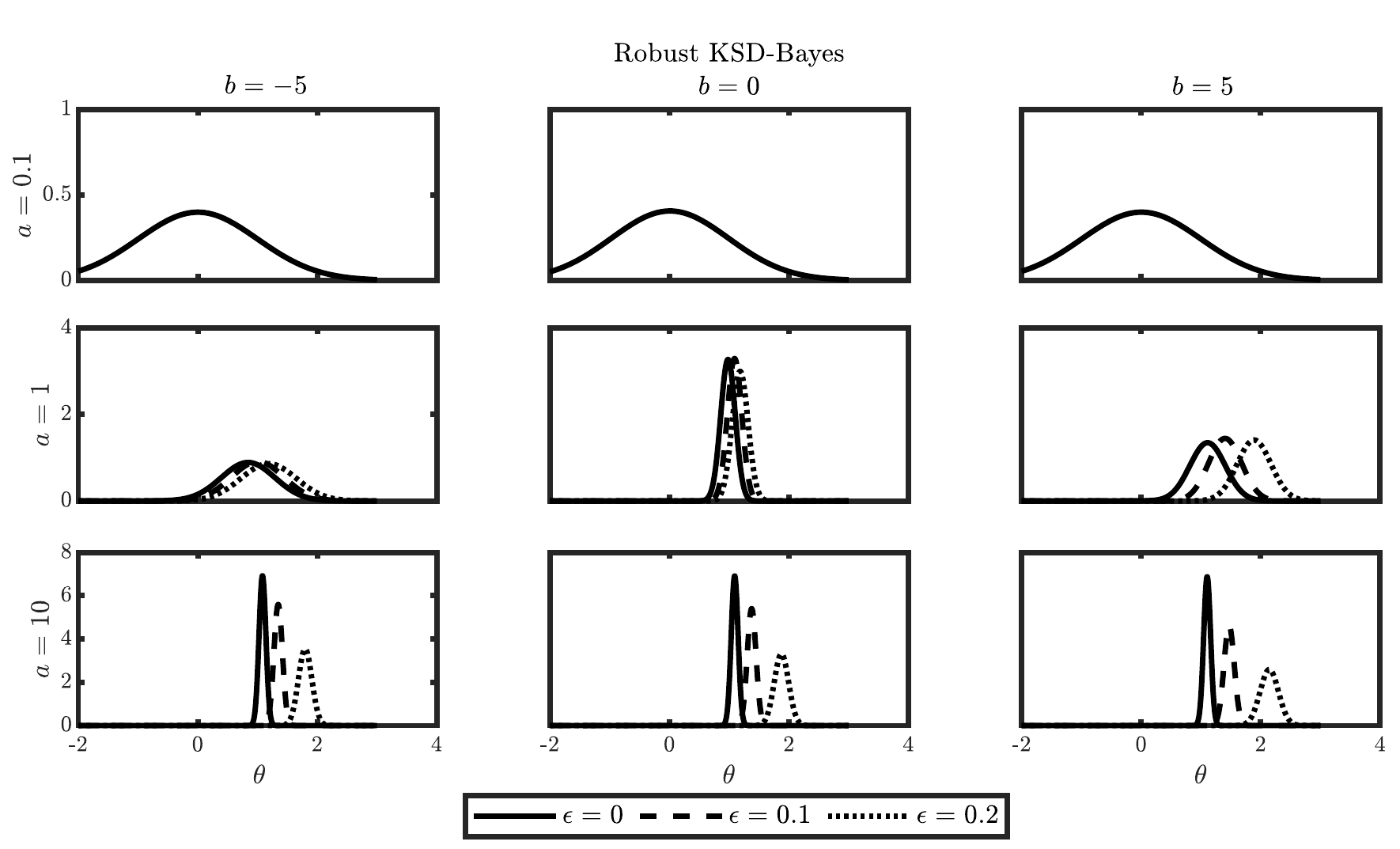}
\caption{\textcolor{black}{
Efficiency/robustness trade-off:
Weight functions of the form \eqref{eq: more general m}, with length-scale parameter $a$ and location parameter $b$, are considered in the context of the normal location model in \Cref{sec:experiment-nl}.
The settings $a = 1$, $b = 0$ (central panel) were used in the main text.
The true parameter value is $\theta = 1$, while a proportion $\epsilon$ of the data were contaminated by noise of the form $\mathcal{N}(y,1)$.
Here $y=10$ is fixed and $\epsilon \in \{0,0.1,0.2\}$ are considered.
}}
\label{fig: eff robust trade}
\end{figure}

\subsection{Comparison with Robust Generalised Bayesian Procedures}
\label{subsec: gen bayes comparison}

This paper presented a generalised Bayesian approach to inference for models that involve an intractable likelihood.
However, several generalised Bayesian approaches exist for \textit{tractable} likelihood and it is interesting to ask how the performance of KSD-Bayes compares to these existing approaches in the case of a tractable likelihood.
To this end, we return to the normal location model of \Cref{sec:experiment-nl}, which has a tractable likelihood, and consider two distinct generalised Bayesian procedures that have been developed in this context; the \textit{power posterior} approach of \citet{holmes2017assigning} and the \textit{MMD-Bayes} approach of  \citet{Cherief-Abdellatif2019}.
These approaches are representative of two of the main classes of robust statistical methodology; data-adaptive scaling parameters $\beta$ and minimum discrepancy methods.
Both approaches are briefly recalled:

\paragraph{Power Posteriors}

Motivated by the \textit{coherence} argument of \citet{Bissiri2016}, the authors \citet{holmes2017assigning} consider a generalised posterior of the form, for some $\beta > 0$,
\begin{align*}
    \pi_n(\theta) \propto \pi(\theta) \exp\left\{ \beta \sum_{i=1}^n \log p_\theta(x_i) \right\} ,
\end{align*}
which we call a \textit{power posterior} \citep[e.g. following][]{friel2008marginal}.
To select an appropriate value for $\beta$, with the intention to ``allow for Bayesian learning under model misspecification'', the authors first introduce the function
\begin{align*}
    \Delta(x) = \int_{\Theta} \pi(\theta) \; \| \partial^1 \log p_\theta(x) \|_2^2 \; \mathrm{d}\theta ,
\end{align*}
where we recall that, in our notation, $\partial^1 = (\partial_{\theta_1},\dots,\partial_{\theta_p})$.
Then the authors set
\begin{align}
    \beta = \left\{ \frac{\int_{\X} p_{\hat{\theta}_n}(x) \Delta(x) \mathrm{d}x}{ \frac{1}{n} \sum_{i=1}^n \Delta(x_i) } \right\}^{\frac{1}{2}} , \label{eq: Holmes beta}
\end{align}
where $\hat{\theta}_n$ is a maximiser of the likelihood.
The motivation for \eqref{eq: Holmes beta} is quite involved, so we refer the reader to \citet{holmes2017assigning} for further background.
The authors prove that $\beta \rightarrow 1$ in probability when the model is well-specified \citep[][Lemma 2.1]{holmes2017assigning}, and present empirical evidence of robustness when the model is mis-specified.

For the normal location model of \Cref{sec:experiment-nl} we can compute $\partial^1 \log p_\theta(x) = x - \theta$, $\Delta(x) = 1 + x^2$, $\hat{\theta}_n = \frac{1}{n} \sum_{i=1}^n x_i$, and $\int_{\X} p_{\hat{\theta}_n}(x) \Delta(x) \mathrm{d}x = 2 + (\hat{\theta}_n)^2$, leading to the recommended weight
\begin{align*}
    \beta = \left\{ \frac{ 2 + \left( \frac{1}{n} \sum_{i=1}^n x_i \right)^2 }{ 1 + \frac{1}{n} \sum_{i=1}^n x_i^2 } \right\}^{\frac{1}{2}} 
\end{align*}
and an associated generalised posterior that is again Gaussian with mean $(\frac{\beta n}{1 + \beta n}) ( \frac{1}{n} \sum_{i=1}^n x_i )$ and variance $\frac{1}{1 + \beta n}$.

\paragraph{MMD-Bayes}

An analogue of KSD-Bayes for tractable likelihood is provided by the MMD-Bayes approach of \citet{Cherief-Abdellatif2019}, where a \textit{maximum mean discrepancy} (MMD) is employed in place of KSD.
In identical notation to that used in \Cref{subsec: sd bayes}, the MMD-Bayes generalised posterior is defined, for some $\beta > 0$, as
\begin{align}
	\pi_n^D(\theta) \propto \pi(\theta) \exp\left\{- \beta n \operatorname{MMD}^2(\P_{\theta} , \P_n) \right\}
    \label{eq: mmd bayes def}
\end{align}
where, for a given reproducing kernel Hilbert space $\mathcal{H}$ with reproducing kernel $k : \X \times \X \rightarrow \R$, the MMD between distributions $\P$ and $\Q$ on $\X$ is defined as
\begin{align*}
    \operatorname{MMD}(\P , \Q) = \| \mu_\P - \mu_\Q \|_\H ,
\end{align*}
where the Bochner intergals $\mu_\P(\cdot) = \int_{\X} k(\cdot,x) \mathrm{d}\P(x)$ and $\mu_\Q(\cdot) = \int_{\X} k(\cdot,x) \mathrm{d}\Q(x)$ are the \textit{kernel mean embeddings} of $\P$ and $\Q$ in $\H$.
The authors prove a generalisation bound for MMD-Bayes \citep[][Theorem 1]{Cherief-Abdellatif2019}, which they interpret as showing ``the MMD-Bayes posterior distribution is
robust to misspecification''.
The authors do not recommend a default choice of $\beta$ in the main text\footnote{\color{black} \citet{Cherief-Abdellatif2019} absorbed the $n$ factor in \eqref{eq: mmd bayes def} into their definition of $\beta$, but for convenience of the reader we have adjusted the presentation of MMD-Bayes to match that used for KSD-Bayes in the main text.}, but in private correspondence they recommend $\beta = O(1)$, and we use $\beta = 1$ as a default.
The kernel $k(x,y) = \exp(-\|x-y\|_2^2/d)$ was used in our experiment, following Appendix F in \citet{Cherief-Abdellatif2019}.

For the normal location model of \Cref{sec:experiment-nl} we can compute the kernel mean embeddings $\mu_{\P_\theta}(x) = \sqrt{\frac{1}{3}} \exp(-\frac{1}{3}(x - \theta)^2)$, $\mu_{\P_n}(x) = \frac{1}{n} \exp(- (x - x_i)^2)$, obtaining an overall expression for the MMD:
\begin{align*}
    \operatorname{MMD}(\P_\theta , \P_n)^2 = \frac{1}{3} \exp\left( - \frac{\theta^2}{6} \right) - \frac{2}{n} \sum_{i=1}^n \sqrt{\frac{1}{3}} \exp\left( - \frac{(\theta - x_i)^2}{3} \right) + \frac{1}{n^2} \sum_{i,j = 1}^n \exp\left( - (x_i-x_j)^2 \right)
\end{align*}
The un-normalised density associated with this generalised posterior can be pointwise evaluated; we do this over a fine grid to approximate the normalisation constant in the experiments that we report.

\paragraph{Results}

The experiment of \Cref{sec:experiment-nl} was conducted using the power posterior and MMD-Bayes methods just described, with results shown in \Cref{fig: other methods}.
Power posteriors exhibited similar performance to (non-robust) KSD-Bayes (i.e. with $M \equiv 1$; see \Cref{fig:normal location model} in the main text), and was therefore less robust to contamination compared with robust KSD-Bayes (i.e. with $M(x) = (1+x^2)^{-1/2}$).
MMD-Bayes generalised posteriors provided similar performance to robust KSD-Bayes in this experiment, albeit exhibiting greater spread.
The spread of the MMD-Bayes generalised posterior might be improved if a data-adaptive learning rate $\beta$ is used, but such an approach was not proposed in \citet{Cherief-Abdellatif2019}.

\begin{figure}[t!]
\centering
\includegraphics[width = \textwidth,clip,trim = 0cm 10.2cm 0.5cm 10cm]{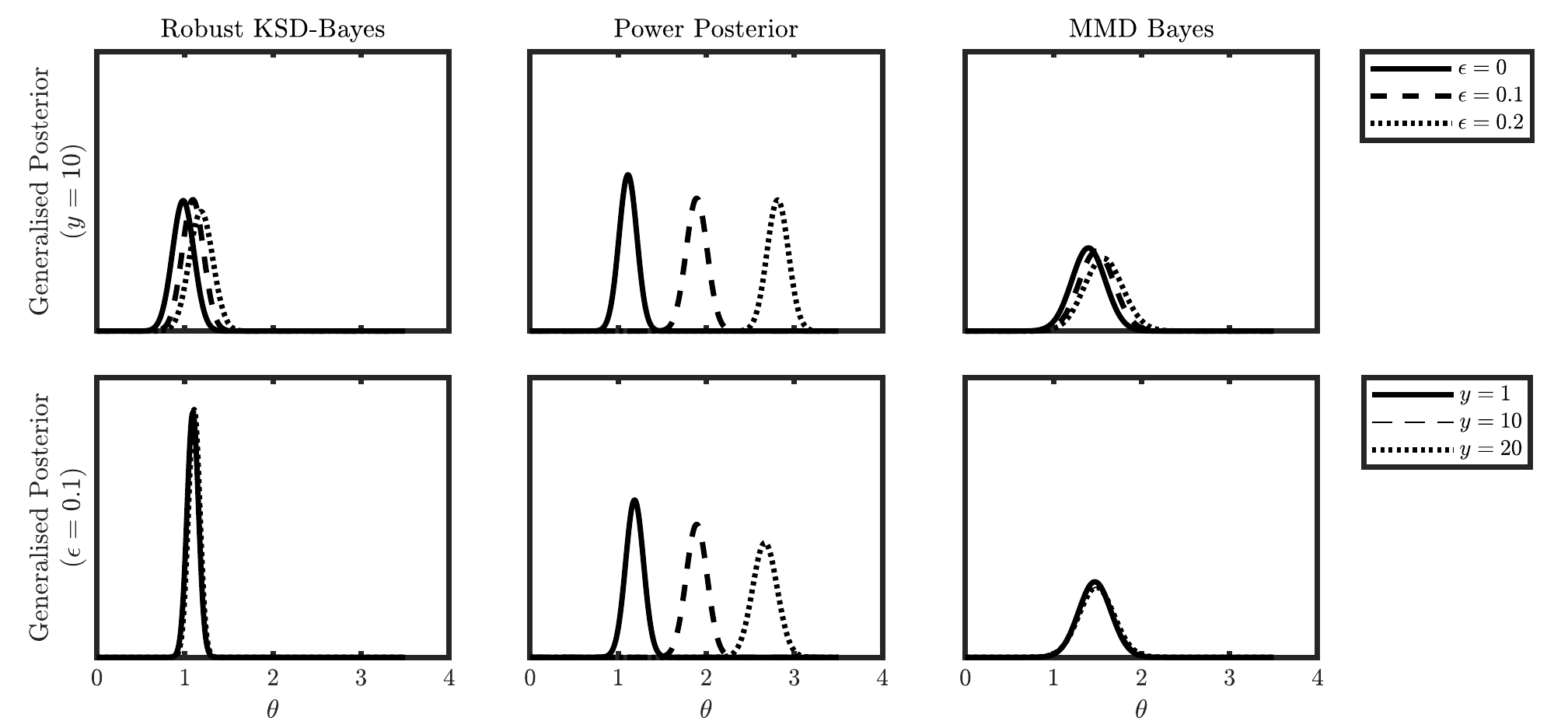}
\caption{\textcolor{black}{
Comparison with robust generalised Bayesian procedures:
Robust KSD-Bayes (this paper), \textit{power posterior} \citep{holmes2017assigning} and \textit{MMD-Bayes} \citep{Cherief-Abdellatif2019} approaches are considered in the context of the normal location model in \Cref{sec:experiment-nl}.
The true parameter value is $\theta = 1$, while a proportion $\epsilon$ of the data were contaminated by noise of the form $\mathcal{N}(y,1)$.
In the top row $y=10$ is fixed and $\epsilon \in \{0,0.1,0.2\}$ are considered, while in the bottom row $\epsilon = 0.1$ is fixed and $y \in \{1,10,20\}$ are considered.
}}
\label{fig: other methods}
\end{figure}

\subsection{Application to Discrete Data}
\label{subsec: other-space}

This section illustrates how KSD-Bayes may be applied to an intractable discrete-space model; note that the theoretical results in \Cref{sec:prop_KSD} and \Cref{sec:pc_bvm} cover both the discrete and continuous data context.
For demonstration purposes we consider a simple Ising model $\P_\theta$ on a vectorised $10 \times 10$ lattice $\X = \{-1,1\}^{100}$, with a \textit{temperature} parameter $\theta \in (0, \infty)$, whose density is
\begin{align}
	p_\theta(x) & \propto \exp\Big( \theta^{-1} \sum_{(i,j) \in E} x_{(i)} x_{(j)} \Big), \label{eq:disc_model}
\end{align}
where $x = \{ x_{(1)}, \dots, x_{(100)} \}$ and $E$ is a index set containing all pairs of adjacent states in the lattice $\mathcal{X}$.
The classical Ising model describes the statistical mechanics of molecular magnetic dipoles, with $\theta$ controlling the intensity of interaction between each adjacent magnetic dipole.
To construct a Stein operator in this setting we follow \cite{Yang2018}, defining the \emph{difference} operators $\nabla^+$ and $\nabla^-$ for a function $h: \{-1,1\}^{d} \to \R$ as
\begin{align*}
	\nabla^+ h(x) = \left[ \begin{array}{c}
		h(x^{+,1}) - h(x)\\
		\vdots\\
		h(x^{+,d}) - h(x)\\
	\end{array} \right] \quad \text{ and } \quad \nabla^- h(x) = \left[ \begin{array}{c}
		h(x^{-,1}) - h(x) \\
		\vdots\\
		h(x^{-,d}) - h(x)\\
	\end{array} \right] 
\end{align*}
where $x^{(+,i)}$ and $x^{(-,i)}$ are vectors whose $i$-th coordinate is $1$ if $x_{(i)} = -1$ and $-1$ if $x_{(i)} = 1$, with all other coordinates identical to their values in $x$.
The difference operators can be extend to act element-wise on vector-valued functions $h: \{0,1\}^{d} \to \R^d$, so that $\nabla^+ h(x)$ and $\nabla^- h(x)$ are $d \times d$ matrices whose $i$-th columns are given, respectively, by $\nabla^+ h_i(x)$ and $\nabla^- h_i(x)$.
Further, for a vector-valued function $h: \{0,1\}^{d} \to \R^d$ we let $\nabla^- \cdot h(x) = \sum_{i=1}^{d} h_i(x^{(-,i)}) - h_i(x) = \operatorname{Tr}( \nabla^- h(x) ) \in \R$.
The operator $\nabla^- \cdot$ will be applied to a matrix-valued kernel $K$; in the same manner as the divergence operator in continuous domain, $\nabla^- \cdot h(x)$ takes a value in $\R^d$ for a matrix-valued function $h: \{0,1\}^{d} \to \R^{d \times d}$ where $[ \nabla^- \cdot h(x) ]_i = \sum_{j=1}^{d} h_{ij}(x^{(-,j)}) - h_{ij}(x)$.
The Stein operator $\S_{\P_\theta}$ we consider in this example is as follows:
\begin{align*}
	\S_{\P_\theta} h (x) = \frac{\nabla^+ p_\theta(x)}{p_\theta(x)} \cdot h(x) + \nabla^- \cdot h(x) 
\end{align*}
For an empirical distribution $\P_n$ associated to a dataset $\{ x_i \}_{i=1}^{n}$, the corresponding $\operatorname{KSD}$ i.e. $\operatorname{KSD}^2(\P_\theta, \P_n) = (1 / n^2) \sum_{i=1}^{n} \sum_{j=1}^{n} \S_{\P_\theta} \S_{\P_\theta} K(x_i, x_j)$ is based on
\begin{multline*}
	\S_{\P_\theta} \S_{\P_\theta} K(x, x') =  \frac{\nabla^+ p_\theta(x)}{p_\theta(x)} \cdot K(x, x')  \frac{\nabla^+ p_\theta(x')}{p_\theta(x')} + \operatorname{tr}\big( \nabla^-_{x} \cdot \nabla^-_{x'} \cdot K(x, x') \big) \\
	+ \frac{\nabla^+ p_\theta(x)}{p_\theta(x)} \cdot \big( \nabla^-_{x'} \cdot K(x, x') \big) + \big( \nabla^-_{x} \cdot K(x, x') \big) \cdot \frac{\nabla^+ p_\theta(x')}{p_\theta(x' )} .
\end{multline*}
where $\nabla^-_{x'}$ denotes an action of the operator $\nabla^- \cdot$ with respect to the argument $x'$ and likewise for $\nabla^-_{x}$.
See \citet{Yang2018} for further detail.
The availability of a discrete KSD enables the application of our KSD-Bayes methodology to the Ising model.

\begin{figure}[t!]
	\centering
	\hfill
	\includegraphics[height=100pt]{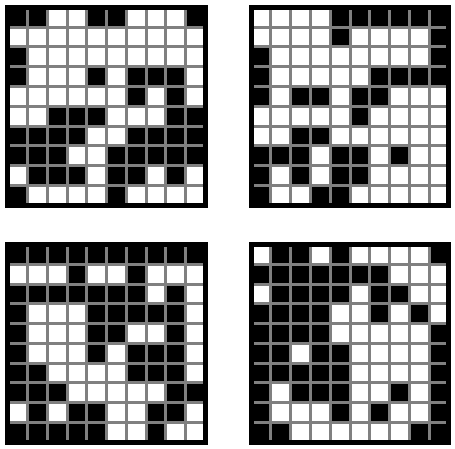}
	\hfill
	\hfill
	\includegraphics[height=100pt]{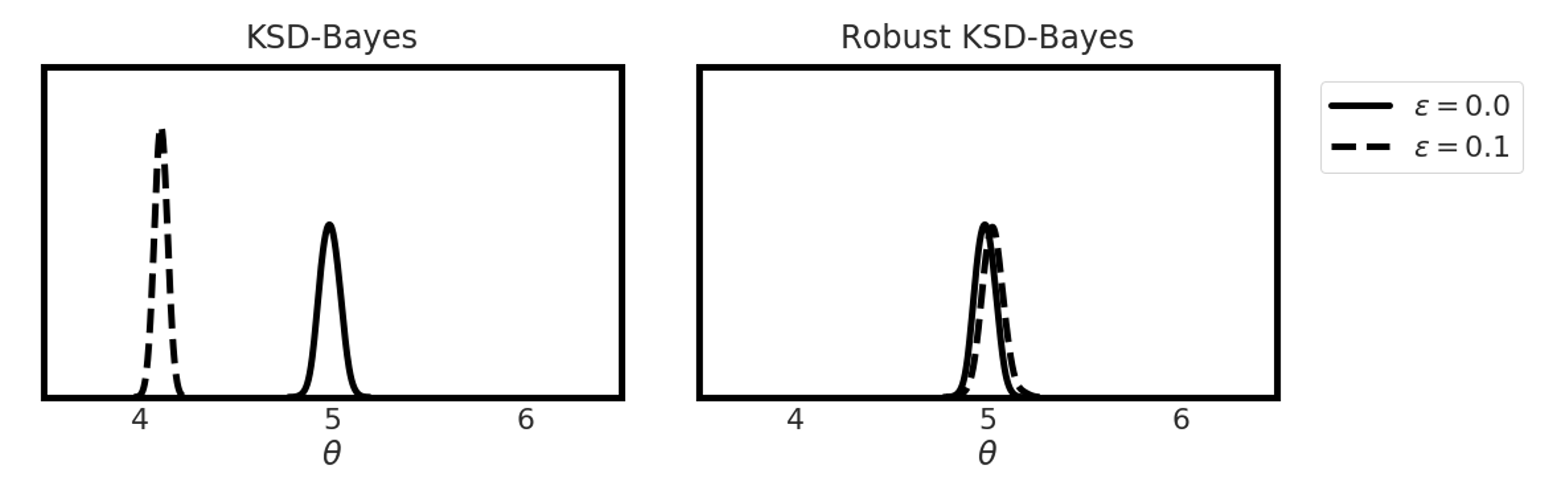}
	\caption{\textcolor{black}{Application to discrete data: 
	Left: Samples from the Ising model \eqref{eq:disc_model} at temperature $\theta=5$, defined on a $10 \times 10$ lattice, where white cells represent $+1$ and black cells represent $-1$. 
	Right: The KSD-Bayes and robust KSD-Bayes generalised posteriors obtained from an uncontaminated ($\epsilon=0$) and contaminated ($\epsilon=0.1$) dataset.}}
	\label{fig:ising_model}
\end{figure}

As an empirical demonstration, we consider the same setting as \cite{Yang2018}; we approximately draw 1000 samples $\{ x_i \}_{i=1}^{1000}$ from $\P_\theta$ with $\theta = 5$ using thinned MCMC (see \Cref{fig:ising_model}, left).
The prior $\pi$ was taken to be a half-normal distribution over $\Theta = [0, \infty)$ with the scale hyper-parameter $3.0$.
Our focus is on robustness of the generalised posterior, and for the contamination model we replaced a proportion $\epsilon$ of the data with the vector $( 1, 1, \cdots, 1 )$, corresponding to the all-white lattice (a configuration more typically observed at low values of the temperature parameter $\theta$).
For KSD-Bayes, the kernel in \cite{Yang2018} was used in combination with a weighting function $M(x)$, i.e. our kernel is
\begin{align*}
	K(x, x') = M(x) \exp\left( - \frac{1}{2d} \sum_{i=1}^d | x_{(i)} - x_{(i)}' | \right) M(x)^\top
\end{align*}
where $d = 100$.
For the weighting function $M(x)$, we examined two choices: (i) $M(x) = I_d$ and (ii) $M(x) = \mathbbm{1}\{ | \sum_i x_{(i)} | \leq 90 \} \times I_d$.
The kernel in case (i) coincides with the one used in \cite{Yang2018}.
The weighting function in case (ii) is designed to limit the influence of data whose coordinates are almost all equal.
The generalised posterior in cases (i) and (ii) will be called, respectively, the KSD-Bayes posterior and the robust KSD-Bayes posterior.
The KSD-Bayes and robust KSD-Bayes posteriors were approximated using Hamiltonian Monte Carlo.
For simplicity, the weight $\beta = 1$ was fixed in this experiment.
Results in \Cref{fig:ising_model} (right) present the generalised posteriors for a uncontaminated ($\epsilon = 0.0$) and contaminated ($\epsilon = 0.1$) dataset.
It can be observed that both the KSD-Bayes and robust KSD-Bayes posteriors place their mass near the true parameter $\theta = 5$ when there is no contamination $\epsilon = 0$.
Furthermore, when contamination is present, the robust KSD-Bayes posterior is not strongly affected.

The computational challenge associated with discrete intractable likelihoods, as exemplified by the Ising model, continues to attract attention \citep[e.g.][]{kim2021variational}.
Perhaps as a consequence, there has been little consideration of robust estimation in this context.
The nature of data contamination in discrete spaces, and the extent to which this can be mitigated by careful selection of the weighting function in KSD-Bayes, requires further careful examination and will be addressed in a sequel.
However, these preliminary results are an encouraging proof-of-concept.

\end{document}